\documentclass{amsart}

\usepackage{graphicx}
\usepackage{amsfonts}
\usepackage{amsmath}
\usepackage{amssymb}
\usepackage{color}
\usepackage{comment}

 \usepackage[utf8]{inputenc} 
\usepackage{pdfpages}

\usepackage{bbding}

\newcommand{\QQ}{{\mathbb{Q}}}

\newcommand{\NN}{{\mathbb{N}}}
\newcommand{\PP}{{\mathbb{P}}}
\newcommand{\KK}{{\mathbb{K}}}

\newcommand{\LL}{{\mathbb{L}}}

\newcommand{\bigOsoft}{\tilde{\mathcal{O}}}
\newcommand{\bigO}{\mathcal{O}}

\newtheorem{thm}{Theorem}
\newtheorem{lem}[thm]{Lemma}
\newtheorem{prop}[thm]{Proposition}
\newtheorem{cor}[thm]{Corollary}

\newtheorem{defi}[thm]{Definition}
\newtheorem{xx}[thm]{Example}

\newtheorem{rem}[thm]{Remark}

\numberwithin{equation}{section}
\begin{document}

\title{Symbolic Computations of  First Integrals for Polynomial Vector Fields%\thanks{Grants or other notes
}


%\author{Guillaume Chèze      \and
%        Thierry Combot %etc.
%}
\author{Guillaume Chèze}
\author{Thierry Combot}

\address{G. Chèze \\
Institut de Math\'ematiques de Toulouse, UMR5219\\
Universit\'e de Toulouse; CNRS \\
UPS IMT, F-31062 Toulouse Cedex 9, France \\
%              %Tel.: +33 5 61 55 76 31 \\
%              %Fax: +123-45-678910\\
             %\email{guillaume.cheze@math.univ-toulouse.fr}    
             }
\email{guillaume.cheze@math.univ-toulouse.fr} 

\address{T. Combot\\
           Université de Bourgogne\\
Bâtiment Mirande\\
9 avenue A. Savary BP~47870\\
21078 DIJON Cedex, France\\
}
 \email{thierry.combot@u-bourgogne.fr}                
%\address{G. Chèze  \at
%Institut de Math\'ematiques de Toulouse, UMR5219\\
%Universit\'e de Toulouse; CNRS \\
%UPS IMT, F-31062 Toulouse Cedex 9, France \\
%              %Tel.: +33 5 61 55 76 31 \\
%              %Fax: +123-45-678910\\
%              \email{guillaume.cheze@math.univ-toulouse.fr}           %  \\
%%             \emph{Present address:} of F. Author  %  if needed
%           T. Combot \at
%           Université de Bourgogne\\
%Bâtiment Mirande\\
%9 avenue A. Savary BP 47870\\
%21078 DIJON Cedex, France\\
% \email{thierry.combot@u-bourgogne.fr}                }


\maketitle

\begin{abstract}
In this article we show how to generalize to the Darbouxian, Liouvillian and Riccati case the extactic curve introduced by J. Pereira. With this approach, we get new algorithms for computing, if it exists, a rational, Darbouxian, Liouvillian or Riccati first integral with bounded degree of a polynomial planar vector field. We give probabilistic and  deterministic algorithms. 
The arithmetic complexity of our probabilistic algorithm is in $\bigOsoft(N^{\omega+1})$, where $N$ is the bound on the degree of a representation of the first integral and $\omega \in [2;3]$ is the exponent of linear algebra. This result improves previous algorithms. Our algorithms have been implemented in Maple and are available on authors' websites. In the last section, we give some examples showing the efficiency of these algorithms.
%\subclass{34A05 \and 68W30 \and 68W40}
\end{abstract}

\section{Introduction}
In this article, we design an algorithm that given a planar polynomial vector field with degree $d$
$$ (S)\; : \quad \left\lbrace \begin{array}{rcl} \dot{x} &=& A(x,y), \\ \dot{y} &=& B(x,y), \end{array}\right. \qquad A,B\in \KK[x,y],\qquad \deg(A), \deg(B) \leq d$$
and some bound $N\in\mathbb{N}$, computes first integrals of $(S)$ of ``size'' (for some appropriate definition) lower than $N$.\\ The field $\KK$ is an effective field of characteristic zero, i.e. one can perform arithmetic operations and test equality of two elements (typically, $\KK=\QQ$ or
$\QQ(\alpha)$, where $\alpha$ is an algebraic number).\\

First integrals are \emph{non-constant functions} $\mathcal{F}$ that are constant along the solutions $\big(x(t),y(t)\big)$ of $(S)$. This property can be rewritten as being solution of a partial differential equation
\begin{equation}\label{eqint}\tag{Eq}
A(x,y) \partial_x \mathcal{F}(x,y)+B(x,y) \partial_y \mathcal{F}(x,y)=0,
\end{equation}
which can be also written $D_0(\mathcal{F})=0$ with $D_0$ the derivation $$D_0=A(x,y) \partial_x +B(x,y) \partial_y.$$

Let us remark that multiplying $A,B$ by some arbitrary non zero polynomial does not change the solutions of equation (Eq). Thus in the rest of the article, we will always consider $\gcd(A,B)=1$, thus excluding the case $A=0$ or $B=0$ as a trivial one. Indeed, in this case $\mathcal{F}(x,y)=x \hbox{ or } y$ is then a first integral.\\

We need to specify in which class of functions we are searching $\mathcal{F}$. The simplest class is the class of rational first integrals, for which we can easily define the notion of size by the degree of its numerator and denominator.

However, we can compute first integrals in larger classes of functions.\\
 The first method for computing first integrals in a symbolic way can be credited to G. Darboux in 1878, see \cite{Da}. Darboux has introduced what we call nowadays Darboux's polynomials. These polynomials  allows one to compute  ``Darbouxian'' first integrals.  Darbouxian functions are written with rational functions and logarithms and thus generalize rational functions.  There exist even more general functions than Darbouxian functions, for example we can consider elementary or Liouvillian functions. There exist theoretical results about these kinds of first integrals, see \cite{PrelleSinger,Singer}, and also algorithms for computing these first integrals, see \cite{Man,ManMac,Avellar}.  Roughly speaking, these algorithms compute Darboux polynomials and then combine them in order to construct a first integral.
 Unfortunately, the computation of Darboux polynomials is a difficult problem. Indeed, computing a bound on the degree of the irreducible Darboux polynomials of a derivation is still an open problem. This problem is  called Poincar\'e's problem. Thus in practice, algorithms ask the users for a bound on the ``size'' of the first integral they want to compute. \\
Moreover, the recombination step leads to an exponential complexity in terms of the degree $d$ of the derivation.
%The computation of irreducible Darboux polynomial with degree smaller than $N$ can be done in polynomial time, see \cite{Cheze}. Unfortunately, it can be shown that the exponent in the complexity is bigger than $10$. Moreover, the algorithm proposed in \cite{Avellar} to compute Liouvillian first integrals has an exponential  complexity in terms of the degree $d$ of the derivation. Indeed, the recombination step used to construct the first integral from Darboux polynomials is of combinatorial nature.\\

In this article, we give an algorithm which computes a symbolic first integral: rational,  Darbouxian, Liouvillian or Riccati, see Definition~\ref{deffi} below,  with ``size" bounded by $N$, with $\bigOsoft(N^{\omega+1})$ arithmetic operations in $\mathbb{K}$.\\

  We recall that $\omega \in [2;3]$ is the exponent of linear algebra over $\KK$. This means that we assume that two matrices of size $N\times N$ with entries in $\mathbb{K}$ can be multiplied with $\bigO(N^{\omega})$ arithmetic operations. The soft-O notation $\bigOsoft()$ indicates that polylogarithmic factors are neglected. Furthermore, in the following we suppose that the bound $N$ tends to infinity and $d$ is fixed. The dependence in $d$ is also polynomial.  More precisely, the cost is at most $N^{\omega+1}+(d+N)^{\omega+1}+d^2N^2$ up to a constant factor and logarithmic factors in $N$. In particular, it shows that our algorithm is polynomial in $d$.\\
  Our algorithm is thus more efficient than the existing ones.\\

Our strategy generalizes to the Darbouxian, Liouvillian and Riccati cases the algorithm proposed in \cite{BCCW} for computing rational first integrals. Our method avoids the computation of Darboux polynomials and then does not need a recombination step.\\

Now, we recall the definition of Darbouxian, Liouvillian first integrals and introduce a new definition: Riccati first integrals. 

\begin{defi}\label{deffi}
\begin{itemize}
%A \textit{rational first integral} of $(S)$ is a first integral $\mathcal{F}\in\overline{\KK}(x,y)$.\\
\item[$\bullet$] A \textit{Darbouxian function}  $\mathcal{F}$ is an expression of the form
$$\mathcal{F}(x,y)=\int G(x,y)dx+F(x,y)dy$$
where $F,G\in\overline{\KK}(x,y)$ and $G(x,y)dx+F(x,y)dy$ is closed, i.e.  $\partial_y G=\partial_x F$, or equivalently
$$\mathcal{F}(x,y)=\frac{P(x,y)}{Q(x,y)}+\sum_i \lambda_i \ln H_i(x,y)$$
where $P,Q,H_i\in\overline{\KK}[x,y]$, $\lambda_i \in \overline{\KK}$.\\

\item[$\bullet$] A \textit{Liouvillian function}  $\mathcal{F}$ is an expression of the form
$$\mathcal{F}(x,y)=\int R(x,y)B(x,y)dx-R(x,y)A(x,y)dy$$
where:
\begin{itemize} 
\item[$\star$] $R(x,y)=\exp\int G(x,y)dx+F(x,y)dy$ (called \textit{the integrating factor}),\\
\item[$\star$] $F(x,y)$, $G(x,y)$ belong to $\overline{\KK}(x,y)$,\\ 
\item[$\star$] $G(x,y)dx+F(x,y)dy$ and $R(x,y)B(x,y)dx-R(x,y)A(x,y)dy$ are closed.\\
\end{itemize}

\item[$\bullet$] A \textit{Riccati function} is an expression of the form $\mathcal{F}_1/\mathcal{F}_2$ where $\mathcal{F}_1,\mathcal{F}_2$ are two independent solutions over $\overline{\KK(x)}$ of a second order differential equation
$$ (EqR) \quad \partial_y^2 \mathcal{F}_i+G(x,y)\partial_y \mathcal{F}_i+F(x,y) \mathcal{F}_i=0$$
with $F,G\in\overline{\KK}(x,y)$.\\

\item[$\bullet$] A Darbouxian (respectively Liouvillian, or Riccati) first integral of $(S)$ is a Darbouxian (respectively Liouvillian, or Riccati) function $\mathcal{F}$ satisfying \mbox{$D_0(\mathcal{F})=0$}.
\end{itemize}
\end{defi}

The classical result of the equivalence of the two representations of a Darbouxian first integral is proved in \cite{Picard,Christopher,Duarte3}, \cite[Satz 2]{Ruppert}, and in \cite[Lemma 2 p. 205]{Schinzel}. Singer \cite{Singer} proves that a vector field admitting a first integral built by successive integrations, exponentiations and algebraic extensions of $\KK(x,y)$ (so a Liouvillian function), also admits a Liouvillian first integral of the form given in Definition \ref{deffi}. Similarly, we will prove in Proposition \ref{prop:structure_successive_extension} that a vector field admitting a first integral built by successive integrations, exponentiations, algebraic and Riccati extensions of $\KK(x,y)$ (see Definition \ref{def:ric_extension} in Section \ref{sec:first_int_diff_inv}), also admits a Riccati first integral of the form given in Definition \ref{deffi}.\\

A vector field admitting a first integral built by successive integrations and algebraic extensions of $\KK(x,y)$ does not always admits a Darbouxian first integral of the form given in Definition \ref{deffi}. This is due to the possible appearance of algebraic extensions in the $1$-form $G(x,y)dx+F(x,y)dy$. However, as we will prove in Proposition \ref{prop:structure_successive_extension}, such vector field then admits what we call a $k$-Darbouxian first integral.
\begin{defi}\label{deffi2}
A $k$-\textit{Darbouxian first integral} of $(S)$ is a first integral $\mathcal{F}$ of $(S)$ of the form
$$\mathcal{F}(x,y)=\int G(x,y)dx+F(x,y)dy$$
where $k\in\mathbb{N}^*$, $F^k,G^k\in\overline{\KK}(x,y)$ and $G(x,y)dx+F(x,y)dy$ is closed.
\end{defi}
Putting $k=1$ recovers the classical Darbouxian first integrals.\\
 Using that $\mathcal{F}$ is a first integral and so $D_0(\mathcal{F})=0$, we have moreover
$$\frac{G(x,y)}{B(x,y)}=-\frac{F(x,y)}{A(x,y)},$$
and this defines a hyperexponential function $R(x,y)$. Now writing $\mathcal{F}(x,y)=\int R(x,y)B(x,y)dx-R(x,y)A(x,y)dy$, we recognize the form of a Liouvillian first integral. So $k$-Darbouxian functions define an intermediary class between Darbouxian and Liouvillian functions.\\

We will not consider elementary first integrals. In \cite{PrelleSinger}, Prelle and Singer have proved that the study of elementary first integrals can be reduced to the study of Liouvillian first integrals with an algebraic integrating factor.
 This meets our Definition of $k$-Darbouxian first integral. However, elementary first integrals require the additional condition that the $1$-form can be integrated in elementary terms, and such integration problems will not be considered in this article.\\

Rational, $k$-Darbouxian and Liouvillian first integrals are particular cases of a Riccati first integral, by simply taking $\mathcal{F}_1$ as the first integral and $\mathcal{F}_2=1$.
Indeed, a rational, $k$-Darbouxian or Liouvillian first integral always satisfies a second order differential equation in $y$. In equation $(EqR)$, it is always possible to multiply $\mathcal{F}_1,\mathcal{F}_2$ by the same hyperexponential function, leaving unchanged the quotient. This allows one to force the Wronskian to become $1$, so we can put  $G=0$. This suggests that one can  represent rational, Darbouxian, Liouvillian and Riccati first integrals as solutions of a differential equation in $y$.
\begin{itemize}
\item A rational first integral is a first integral solution of
\begin{equation}\label{eqtype0}\tag{Rat}
\mathcal{F}-F(x,y)=0\qquad F\in\overline{\KK}(x,y)\setminus \overline{\KK}.
\end{equation}
\item A $k$-Darbouxian first integral is a first integral solution of
\begin{equation}\label{eqtype1}\tag{D}
\partial_y \mathcal{F}-F(x,y) =0,\qquad F^k\in\overline{\KK}(x,y)\setminus \{0\}.
\end{equation}
\item A Liouvillian first integral is a first integral solution of
\begin{equation}\label{eqtype2}\tag{L}
\partial_y^2 \mathcal{F}-F(x,y)\partial_y \mathcal{F}=0\qquad F\in\overline{\KK}(x,y).
\end{equation}
\item A Riccati first integral is a first integral quotient of two independent solutions over $\overline{\KK(x)}$  of
\begin{equation}\label{eqtype3}\tag{Ric}
\partial_y^2 \mathcal{F}-F(x,y) \mathcal{F}=0,\qquad F\in\overline{\KK}(x,y).
\end{equation}
\end{itemize}

These four equations will be the four canonical equations representing respectively each type of first integral. Once one of the above equation is found, it is possible to recover the first integral by single variable integration and linear differential equation solving.\\

Each case is included in the next one, leading to a ranking on the classes of first integrals
$$\hbox{Rational}<k\hbox{-Darbouxian}<\hbox{Liouvillian}<\hbox{Riccati}.$$
Each type of equation can be represented by a single rational function, thus also giving us a notion of ``size''.
\begin{defi}
The degree of a rational, Darbouxian, Liouvillian, Riccati first integral is respectively the maximum of the degree of numerator and denominator of $F$ (or $F^k$ in the $k$-Darbouxian case) in the four above equations.
\end{defi}

The output of our algorithm will be an equation of the form (Rat), (D), (L) or (Ric).  We will see that we can always suppose $F$ with coefficients in $\KK$. This does not change the degree of $F$, see Corollaries \ref{cor:darbouxdansK}, \ref{cor:liouvdansK}, \ref{cor:RiccatidansK}. The main theorem of the article is the following:

\begin{thm}\label{thmmain1}
Let $d$ be the maximum of $\deg(A)$ and $\deg(B)$. The problem of finding symbolic (rational, $k$-Darbouxian, Liouvillian, Riccati) first integrals with degree smaller than $N$ can be solved in a probabilistic way with at most \mbox{$\bigOsoft(N^{\omega+1})$} arithmetic operations in $\mathbb{K}$, and the factorization of a univariate polynomial with degree at most $N$ and coefficients in $\KK$.\\ 
More precisely, there exists an algorithm with inputs $A,B$, $k\in\mathbb{N}^*$, a bound $N$, and parametrized by initial conditions $z\in \KK^{2}$ such that the possible outputs are:
\begin{itemize}
\item a differential equation of one of the forms \eqref{eqtype0}, \eqref{eqtype1}, \eqref{eqtype2}, \eqref{eqtype3} leading to a symbolic first integral, 
\item ``None'' meaning that there exists no symbolic first integral with degree smaller than $N$,
\item ``I don't know''.
\end{itemize}
Furthermore, if $z$ avoids the roots of a non-zero polynomial with degree $\bigO(N^4)$  then the algorithm does not return ``I don't know".\\
Moreover, if $(S)$ admits a symbolic first integral with degree smaller than $N$ then the output is a differential equation with minimal degree.
\end{thm}

The parameter $k$ is necessary for $k$-Darbouxian first integrals. An equation admitting a $k\geq 2$-Darbouxian first integral also admits a Liouvillian first integral. So for the (default) input $k=1$, we detect Darbouxian first integrals, but the possible $k\geq 2$-Darbouxian first integral could stay unnoticed or seen as a Liouvillian first integral. The reduction of such Liouvillian first integral to a $k$-Darbouxian first integral comes down to an integration problem, i.e. testing if the integrating factor is an algebraic function, which will not be considered in this article.\\

As we use the dense representation of polynomials and $\deg(F) \leq N$, the size of the output of our algorithm is in $\bigO(N^2)$. Thus our algorithm has a sub-quadratic complexity if we use linear algebra algorithms with $\omega <3$. \\

%As the  last possible output, ``I don't know", can only appear when $z$ is a root of a non-zero polynomial with degree $\bigO(dN^4)$ and  as we are considering fields of characteristic zero, we can say that for almost all $z$ the algorithm detects symbolic first integrals.\\
Furthermore, we can say that for almost all $z$ the algorithm detects symbolic first integrals. Indeed, we are considering fields of characteristic zero and the output "I don't know" can only appear when $z$ is a root of a non-zero polynomial.\\

Repeating the probabilistic algorithm in order to avoid bad values for $z$ provides a deterministic algorithm with a polynomial complexity:

\begin{cor}
The probabilistic algorithm can be turned into a deterministic one. The deterministic algorithm uses at most $\bigOsoft(N^{\omega+9})$ arithmetic operations in $\mathbb{K}$, and $\bigO(N^8)$ factorizations of univariate polynomials with coefficients in $\mathbb{K}$ and degree at most $N$.
\end{cor}
In practice, $z$ is chosen uniformly at random. This allows to avoid bad situations. Then we do not need to repeat the algorithm. Thus the practical timings obtained in Section~\ref{sec:examples} do not reflect this complexity result.\\
%In practice the complexity of the deterministic algorithm is better, see Section~\ref{sec:examples}.\\
 
At last, we mention that the algorithm can even sometimes returns equations of degree higher than $N$. This situation can appear for example when we are looking for a Darbouxian first integral with degree smaller than $N$ and there exists a rational first integral of degree $2N$. We give such examples in  Section \ref{sec:examples}.\\
% In these kinds of situations the degree of the first integral is bigger than $N$. This means that the algorithm gives a minimal solution in terms of the degree of the first integral and not in terms of its class.
%However, if the algorithm returns ``None" then the algorithm ensures that the equation $(S)$ does not admit a first integral in a lower class of degree $\leq N$.

\subsection{Strategy, description and theoretical contributions}
In this paper we generalize the approach given in \cite{BCCW} for computing rational first integrals. The main idea was to compute a solution $y(x_0,y_0;x)$ as a power series in $x$ with coefficients in $\KK(x_0,y_0)$ of 
$$(E)\; : \quad \partial_x y(x_0,y_0;x)=\frac{B\big(x,y(x_0,y_0;x)\big)}{A\big(x,y(x_0,y_0;x)\big)},  \textrm{ and }  y(x_0,y_0;x_0)=y_0,$$
and then find a rational function $F(x,y) \in \KK(x,y)$ such that 
$$F\big(x,y(x_0,y_0;x)\big)=F(x_0,y_0).$$
The key ingredient of this approach was the following:\\
 If we know $y(x_0,y_0;x) \mod (x-x_0)^{\sigma}$ with $\sigma$ big enough, then we can compute $F$ from this truncated power series.\\

 %and then find a polynomial $\mathcal{F}$ vanishing at this power series solution.
 %If the power series order is large enough, it is sufficient for recovering $\mathcal{F}$. The determination of $\mathcal{F}$ indeed comes down to a linear algebra problem. Furthermore, the polynomial $\mathcal{F}$ as the following form: $P(x,y)Q(x_0,y_0)-Q(x,y)P(x_0,y_0)$ where $P/Q$ is a rational first integral. Thus the computation of $\mathcal{F}$ gives a rational first integral.\\
 In \cite{BCCW}, in order to avoid computations in $\KK(x_0,y_0)$ two solutions of $(E)$ with random initial conditions are used and give a probabilistic and then a deterministic algorithm.\\

%  Here we use the same approach, but instead of computing a solution, we consider the flow, i.e. the solution of the following differential equation where $x_0,y_0$ are new variables:\\
%$$(E)\; : \quad \partial_x y(x_0,y_0;x)=\frac{B(x,y(x_0,y_0;x))}{A(x,y(x_0,y_0;x))},  \textrm{ and }  y(x_0,y_0;x_0)=y_0.$$
%We compute the flow $y$ as a series in $x$ with coefficients in $\KK(x_0,y_0)$.
% The computation of the solution issued from a particular initial condition can be obtained just by evaluating $(x_0,y_0)$ in the $y$ series expansion, provided that the series coefficients do not become singular. So this naturally generalizes the computation of a solution.\\
In this article,  the new ingredient is the following: we consider derivatives of the flow $y(x_0,y_0;x)$ relatively to $y_0$. We set
$$y_1(x_0,y_0;x)=\partial_{y_0}y(x_0,y_0;x), \quad y_2(x_0,y_0;x)=\partial_{y_0}^2y(x_0,y_0;x),$$
$$y_3(x_0,y_0;x)=\partial_{y_0}^3y(x_0,y_0;x).$$
In the following we will sometimes omit the dependence relatively to $x_0,y_0$ in the notations. We will write $y(x)$ instead of $y(x_0,y_0;x)$ and $y_r(x)$ instead of $y_r(x_0,y_0;x)$, for $r=1,2,3$.\\

With a direct computation we remark that 
the functions $y(x),y_1(x),y_2(x),y_3(x)$  are  solutions of the differential system: 
$$(S'_3)\,
 \begin{cases}

(S'_2) 
\begin{cases}
(S'_1) \,
\begin{cases}
(S'_0) \, \quad \partial_x y=\dfrac{B}{A}\\
\,\\
\partial_xy_1= y_1\partial_y \left(\dfrac{B}{A}\right)\\
 \end{cases}\\
\, \\
\partial_x y_2=y_2\partial_y \left(\dfrac{B}{A}\right)+y_1^2\partial_y^2 \left(\dfrac{B}{A}\right)\\

\end{cases}\\
\, \\
\partial_x y_3=y_3\partial_y \left(\dfrac{B}{A}\right)+3y_2y_1\partial_y^2 \left(\dfrac{B}{A}\right)+y_1^3\partial_y^3 \left(\dfrac{B}{A}\right).\\
\end{cases}\\
$$

The system $(S'_r)$ gives a method to compute  $y(x)$ and $y_r(x)$ for $r=1,2,3$ as power series in $x$. We only have to solve $(S'_r)$ using the Newton method for initial condition $x=x_0,y=y_0,y_1=1,y_2=0,y_3=0$.\\
% As in \cite{BCCW}, in order to get efficient algorithms in practice we will consider 
% random intial conditions $x_0^{\star},y_0^{\star}$. This leads to probabilistic algorithms and then to deterministic algorithms.\\

 %Furthermore,  we will also sometimes denote  $y_2$ (respectively $y_3$) by $y_1^{(2)}$ (respectively  $y_1^{(3)}$) in order to write statements and algorithms in a short and uniform way in terms of $y_1^{(r)}$, where $r \in [[0;3]]$.\\

Now, we introduce new variables $y_1$, $y_2$ and $y_3$ and we define a polynomial derivation $D_r$ which is the Lie derivative along the vector field of $(S'_r)$.
%associated to the  system $(S'_r)$. Here, associated means that $\mathcal{F}$ is a first integral of $(S'_r)$ if and only if $\mathcal{F}$ satisfies $D_r(\mathcal{F})=0$.

\begin{defi}$\,$
\begin{itemize}
\item The system $(S'_1)$ is associated to the derivation $D_1$ in $\KK[x,y,y_1]$:
$$D_1=A^2\partial_x + AB\partial_y+ y_1A^2\partial_y \left(\dfrac{B}{A}\right) \partial_{y_1}.$$

\item The system $(S'_2)$ is associated to the derivation $D_2$ in $\KK[x,y,y_1,y_2]$:
$$D_2=A^3\partial_x + A^2B\partial_y+ y_1A^3\partial_y \left(\dfrac{B}{A}\right) \partial_{y_1}+A^3\Big(y_2\partial_y \left(\dfrac{B}{A}\right)+y_1^2\partial_y^2 \left(\dfrac{B}{A}\right)\Big) \partial_{y_2}.$$

\item The system $(S'_3)$ is associated to the derivation $D_3$ in $\KK[x,y,y_1,y_2,y_3]$:
 
 \begin{eqnarray*}
 D_3&=&A^4\partial_x + A^3B\partial_y+ y_1A^4\partial_y \left(\dfrac{B}{A}\right) \partial_{y_1}+A^4\Big(y_2\partial_y \left(\dfrac{B}{A}\right)+y_1^2\partial_y^2 \left(\dfrac{B}{A}\right)\Big) \partial_{y_2}\\
 &&+A^4\Big(y_3\partial_y \left(\frac{B}{A}\right)+3y_2y_1\partial_y^2 \left(\dfrac{B}{A}\right)+y_1^3\partial_y^3 \left(\dfrac{B}{A}\right)\Big)\partial_{y_3}.\end{eqnarray*}\\
 \end{itemize}
\end{defi}

Let us now consider $\mathcal{F}$ a Darbouxian first integral of $D_0$ such that $\partial_y \mathcal{F}=F$.

We have 
$$\mathcal{F}\big(x,y(x)\big)=\mathcal{F}(x_0,y_0)$$ and thus the derivative relatively to $y_0$ of this equation gives:
$$\partial_y\mathcal{F}\big(x,y(x)\big) y_1(x)=\partial_{y_0}\mathcal{F}(x_0,y_0).$$
%with initial condition $y_{1,0}=1$ since $y(x_0,y_0;x_0)=y_0$.
 Therefore if $\mathcal{F}$ is a Darbouxian first integral of $D_0$ with $\partial_y \mathcal{F}=F$ we get 
$$F\big(x,y(x)\big)y_1(x)=F(x_0,y_0),$$
where $F \in \KK(x,y)$.\\

Now the computation of the rational function $F$ comes down to solving this equation knowing $y(x),y_1(x)$ as power series. This situation is similar to the one studied for rational first integrals.\\ Let us remark moreover that the rational function $F(x,y)y_1$ is constant on $\big(x,y(x),y_1(x)\big)$, where the initial condition is $y(x_0)=y_0$ and $y_1(x_0)=1$. In Section~\ref{sec:first_int_diff_inv}, we prove that $F(x,y)y_1$ is a rational first integral for $(S'_1)$ and the even more general result:

\begin{prop} \label{propinv}
$\,$
\begin{itemize}
\item The system $(S)$ admits a rational first integral associated to \eqref{eqtype0} if and only if $F(x,y)$ is a rational first integral of $(S'_0)$.\\

\item The system $(S)$ admits a Darbouxian first integral associated to \eqref{eqtype1} if and only if $y_1F(x,y)$ is a rational first integral of $(S'_1)$.\\

\item The system $(S)$ admits a Liouvillian first integral associated to \eqref{eqtype2} if and only if $y_1F(x,y)+y_2/y_1$ is a rational first integral of $(S'_2)$.\\

\item The system $(S)$ admits a Riccati first integral associated to \eqref{eqtype3} if and only if $4y_1^2F(x,y)-2y_3/y_1+3y_2^2/y_1^2$ is rational a first integral of $(S'_3)$.\\
\end{itemize}
\end{prop}

This proposition means that the computation of symbolic first integrals is reduced to the computation of a rational first integral with a given structure of a differential system $(S'_r)$. The existence and the computation of these rationals first integrals can then be done thanks to generalized extactic curves. More precisely, this can be done with linear algebra only. For example,  we get this kind of result, see Section~\ref{sec:extacurve}:

\begin{thm}[Liouvillian extactic curve Theorem]$\,$%\label{thm:exta+lfi}$\,$

We can construct from $(S)$ a matrix $\tilde{\mathcal{E}}_{D_2}^N(x_0,y_0)$ with entries in $\KK[x_0,y_0]$ such that its determinant denoted by $\tilde{E}^N_{D_2}(x_0,y_0)$, satisfies the following properties:
%Let $\tilde{\mathcal{E}}_{D_2}^N$ be the matrix 
%$$\tilde{\mathcal{E}}_{2,D_2}^N=\Big(D_2^k\big( x^iy^jy_1^{\alpha}(y_2)^{\beta}\big)\Big)$$
%where $0\leq k\leq 3(N+1)(N+2)/2$, $0\leq i+j\leq N$, $(\alpha,\beta) \in \{ (1,0); (2,0); (0,1) \}$. 
%We have $\det \tilde{\mathcal{E}}^N_{D_2} \in \KK[x,y,y_1,y_2]$ and we set $\tilde{E}^N_{D_2}(x,y)=\det \tilde{\mathcal{E}}^N_{D_2} (x,y,1,0)$.
\begin{enumerate}
\item If $\tilde{E}^N_{D_2}(x_0,y_0)=0$ then the derivation $D_0$ has a Liouvillian first integral with degree smaller than $N$ or a Darbouxian first integral with degree smaller than $2N+3d-1$ or a rational first integral with degree smaller than \mbox{$4N+8d-3$}.
\item If $D_0$ has a rational or a Darbouxian or a Liouvillian first integral with degree smaller than $N$ then $\tilde{E}^N_{D_2}(x_0,y_0)=0$.
\end{enumerate}

\end{thm}

We will explain in Section~\ref{sec:extacurve} how to obtain explicitly the matrix $\tilde{\mathcal{E}}_{D_2}^N(x_0,y_0)$ from the derivation $D_2$.
%see that the entries of $\tilde{\mathcal{E}}_{D_2}^N$ are explicit. Indeed, they have the following form: $$D_2^k\big( x^iy^jy_1^{\alpha}(y_2)^{\beta}\big),$$ where $0\leq k\leq 3(N+1)(N+2)/2$, $0\leq i+j\leq N$, $(\alpha,\beta) \in \{ (1,0); (2,0); (0,1) \}$.\\

 We call this kind of theorem an ``extactic curve theorem". Indeed, the matrix $\tilde{\mathcal{E}}_{D_2}^N$ corresponds to the study of a high order of contact between a solution of the differential system $(S)$ and a Liouvillian function. This generalizes the situation introduced by Pereira in \cite{Pereira} for rational first integrals.\\

The main step in our algorithms will be the computation of a non-trivial element in the kernel of a matrix of the previous form. From such an element we will show that we can easily construct  the rational function $F$ appearing in equations \eqref{eqtype0}, \eqref{eqtype1}, \eqref{eqtype2}, \eqref{eqtype3}. Therefore, the computation of symbolic first integrals with bounded degree is reduced to a linear algebra problem.\\

%
%%%%%%%%%%%%%%%%%%%%%%%%%%%%%%%%%%%%%%%%%%%%%%%%%%%%%%%%%%%%%%%%%%%

%%%%%%%%%%%%%%%%%%%%%%%%%%%%%%%%%%%%%%%%%%%%%%%%%%%%%%%%%%%%%%%%%%%%%
\subsection{Related results}\label{sec:related_results}
The computation of symbolic first integrals can be credited to G. Darboux in 1878, see \cite{Da}. In this paper Darboux introduced what we call nowadays \emph{Darboux polynomials}. A polynomial $f\in \KK[x,y]$ is a Darboux polynomial for $D_0$ means that  $f$ divides $D_0(f)$. Thus $f$ is an invariant algebraic curve. Darboux has shown how to find a Darbouxian first integral thanks to a recombination of Darboux polynomials.\\
This approach has been generalized in order to compute elementary first integrals by Prelle and Singer in \cite{PrelleSinger}. This method has been implemented and studied in \cite{Man,ManMac}.\\
In \cite{Singer}, Singer has given a theoretical characterization of Liouvillian first integrals. This characterization is the main ingredient of the algorithm proposed by Duarte et al. in \cite{Avellar,Duarte2,Duarte3}.\\
The interested reader can also consult the following surveys \cite{Schlomiuk,Goriely,DLA,zhang2017integrability} for more results about Darboux polynomials and first integrals.\\

Roughly speaking, all the previous algorithms proceed as follows: first compute Darboux polynomials with bounded degree and second recombine them in order to find a first integral.\\
These two steps correspond to two practical difficulties. The computation of Darboux polynomials with bounded degree can be performed in polynomial time, see \cite{Cheze}. This method is based on the so-called extactic curve introduced by Pereira in \cite{Pereira} and uses a number of binary operations that is polynomial in the bound $N$, the degree $d$ and the logarithm of the height of $A$ and $B$. Unfortunately, the arithmetic complexity of this computation is in $\bigO(N^{4\omega+4})$, see \cite{Cheze}.\\
 The recombination part can be solved with linear algebra  if we are looking for Darbouxian first integrals. However, if we are looking for a Liouvillian first integral then the recombination step used in \cite{Avellar,Duarte2,Duarte3} uses at least $2^d$ arithmetics operations. Indeed, this algorithm tries to solve a family of equation. Each equation of this family is constucted from a polynomial $\prod_i f_i^{e_i}$, where $f_i$ is a Darboux polynomial and $e_i$ is an unknown integer. A condition on the degree of the output leads to a condition on the degree of $\prod_i f_i^{e_i}$. With this approach if $D_0$ has $k$ Darboux polynomials and the bound on the degree of $\prod_i f_i^{e_i}$ is bigger than $k$ then we have to study at least $2^k$ situations.\\
In \cite{Avellar,Duarte2,Duarte3}, in order to find a Liouvillian first integral, the authors compute a Darbouxian integrating factor $\mathcal{R}=e^{P/Q}\prod_if_i^{c_i}$. With our approach the integrating factor $\mathcal{R}$ is related to the equation (\ref{eqtype2}) in the following way:
$$\dfrac{\partial_y \mathcal{R}}{\mathcal{R}}=F-\dfrac{\partial_y A}{A}.$$
Thus our bound on the degree of $F$ corresponds to a bound on the degree of the polynomials $P,Q,f_i$.

In \cite{FerragutGiacomini}, Ferragut and Giacomini have proposed a method to compute rational first integrals with bounded degree.  This approach does not follow the previous strategy. The idea is to computed a bivariate polynomial annihilating $y(x_0,y_0;x)$ written as a  power series solution of a first order differential  equation. From this polynomial we can then deduce a rational first integral if it exists. Unfortunately, the precision needed on the power series to get a correct output was not explicitly given.\\
In \cite{BCCW}, the authors have improved the Ferragut-Giacomini's method. They have given an explicit bound on the precision needed on the power series to get a rational first integral when it exists. Furthermore, the main step of this algorithm is reduced to linear algebra only. The complexity of the probabilistic algorithm is then in $\bigOsoft(N^{2\omega})$.  However, as remarked by G. Villard this complexity can be lowered to $\bigOsoft(N^{\omega+1})$ with an application of Hermite-Pad\'e approximation. This approach was just study as an heuristic in \cite{BCCW}.\\

The algorithm proposed in this article is based on a generalization of the extactic curve and follows the idea used in \cite{FerragutGiacomini} and \cite{BCCW}. We give then a uniform strategy with a uniform complexity to compute rational, Darbouxian, Liouvillian, and Riccati first integrals. Furthermore, we explain in Section \ref{sec:first_int_diff_inv} why our approach cannot be generalized to another class of functions.

%%%%%%%%%%%%%%%%%%%%%%%%%%%%%%%%%%%%%%%%%%%%%%%%%%%%%%%%%%%%%%%%%

%%%%%%%%%%%%%%%%%%%%%%%%%%%%%%%%%%%%%%%%%%%%%%%%%%%%%%%%%%%%%%%%%%

\subsection{Structure of the paper}
In the second section of this article we prove Proposition~\ref{propinv}, i.e. we show how the computation of a symbolic first integral can be reduced to the computation of a rational first integral of a differential system $(S'_r)$. In the third section, we define and study extactic hypersurfaces. We give a precise statement for the following idea: if an hypersurface has a sufficiently big order of contact with a generic solution of a differential system then this order of contact is infinite. This result will be useful in our algorithm in order to compute a solution with a sufficient precision in order to construct a first integral. As a byproduct we show that the computation of a rational first integral of a derivation in $\KK[x_1,\ldots,x_n]$ can be reduced to a linear algebra problem. In Section \ref{sec:extacurve},  we define the Darbouxian, Liouvillian and the Riccati extactic curve. We prove that these curves allow us to characterize the existence of symbolic first integrals with bounded degree. In Section \ref{sec:extaeval}, we study the evaluation of the extactic curves. In particular we characterize non-generic solutions. In Section~\ref{sec:algo}, we give and prove the correctness of our algorithms based on the previous results. In Section \ref{sec:complexity}, we study the complexity of our algorithms. At last, in Section \ref{sec:examples} we give some examples thanks to our implementation of these algorithms. Our implementation is freely available at:\\
 \texttt{http://combot.perso.math.cnrs.fr/software.html},\\
 \texttt{https://www.math.univ-toulouse.fr/$\sim$cheze/Programme.html}.
%%%%%%%%%%%%%%%%%%%%%%%%%%%%%%%%%%%%%%%%%%%%%%%%%%%%%%%%%%%%%
\subsection{Notations}
In this article, we suppose that $\gcd(A,B)=1$.\\
We denote by $\KK[x,y]_{\leq N}$  the vector space of polynomials in $x,y$ with coefficients in $\mathbb{K}$ of total degree less than $N$.\\
We set $div=\partial_x A +\partial_y B$.\\

In the following, $y(x_0,y_0;x)$ will be a power series solution of 
$$(E)\; : \quad \partial_x y(x_0,y_0;x)=\frac{B(x,y(x_0,y_0;x))}{A(x,y(x_0,y_0;x))},  \textrm{ and }  y(x_0,y_0;x_0)=y_0.$$

We set \vspace{-0.2cm}
\begin{eqnarray*}
y_1(x_0,y_0;x)&=&\partial_{y_0}y(x_0,y_0;x),\\
y_2(x_0,y_0,x)&=&\partial_{y_0}^2y(x_0,y_0;x),\\
y_3(x_0,y_0;x)&=&\partial_{y_0}^3y(x_0,y_0;x).
\end{eqnarray*}
%$$y_1(x_0,y_0;x)=\partial_{y_0}y(x_0,y_0;x), \quad y_2(x_0,y_0,x)=\partial_{y_0}^2y(x_0,y_0;x),$$
%$$ y_3(x_0,y_0;x)=\partial_{y_0}^3y(x_0,y_0;x).$$

In the following, we will sometimes omit the dependence relatively to $x_0$ and $y_0$. We will write $y(x)$ instead of $y(x_0,y_0;x)$ and $y_r(x)$ instead of $y_r(x_0,y_0;x)$, for $r=1,2,3$.\\

These functions $y(x),y_1(x),y_2(x),y_3(x)$  are  solutions of 
$$(S'_3)\,
 \begin{cases}

(S'_2) 
\begin{cases}
(S'_1) \,
\begin{cases}
(S'_0) \, \quad \partial_x y=\dfrac{B}{A}\\
\,\\
\partial_xy_1= y_1\partial_y \left(\dfrac{B}{A}\right)\\
 \end{cases}\\
\, \\
\partial_x y_2=y_2\partial_y \left(\dfrac{B}{A}\right)+y_1^2\partial_y^2 \left(\dfrac{B}{A}\right)\\

\end{cases}\\
\, \\
\partial_x y_3=y_3\partial_y \left(\dfrac{B}{A}\right)+3y_2y_1\partial_y^2 \left(\dfrac{B}{A}\right)+y_1^3\partial_y^3 \left(\dfrac{B}{A}\right),\\
\end{cases}\\
$$
with initial condition $y(x_0)=y_0$, $y_1(x_0)=y_{1,0}$, $y_2(x_0)=y_{2,0}$, $y_3(x_0)=y_{3,0}$.

%In the following we will sometimes omit the dependence relatively to $x_0,y_0$ in the notations. We will write $y(x)$ instead of $y(x_0,y_0;x)$ and $y_r(x)$ instead of $y_r(x_0,y_0;x)$, for $r=1,2,3$.
% Furthermore,  we will also sometimes denote  $y_2$ (respectively $y_3$) by $y_1^{(2)}$ (respectively  $y_1^{(3)}$) in order to write statements and algorithms in a short and uniform way in terms of $y_1^{(r)}$, where $r \in [[0;3]]$.\\
%We use the notation $y_2$  in order to obtain formulae with terms of the form $y_2^2$ instead of terms with the longer form $\big(y_1^{(2)}\big)^2$.
%%%%%%%%%%%%%%%%%%%%%%%%%%%%%%%%%%%%%%%%%%%%%%%%%%%

%%%%%%%%%%%%%%%%%%%%%%%%%%%%%%%%%%%%%%%%%%%%%%%%%%%%%%
\section{First integrals and differential invariants}\label{sec:first_int_diff_inv}
%%%%%%%%%%%%%%%%%%%%%%%%%%%%%%%%%%%%%%%%%%%%%%%%%%%%%%%%%%%%%%%%%%%%%%%%%%
\subsection{Representation of first integrals}

Let us first prove that equations \eqref{eqtype0}, \eqref{eqtype1}, \eqref{eqtype2}, \eqref{eqtype3} used to represent the first integrals allow one to recover them. The next proposition explains  why \eqref{eqtype0}, \eqref{eqtype1}, \eqref{eqtype2}, \eqref{eqtype3} are admissible outputs when we are looking for rational or $k$-Darbouxian or Liouvillian or Riccati first integrals.

\begin{prop} $\;$
\begin{itemize}
\item A rational first integral is uniquely defined by equation \eqref{eqtype0}.
\item A $k$-Darbouxian first integral is defined up to addition of a constant by equation \eqref{eqtype1}.
\item A Liouvillian first integral is defined up to affine transformation by equation~\eqref{eqtype2}.
\item A Riccati first integral is defined up to homography by equation \eqref{eqtype3}.
\end{itemize}
\end{prop}

\begin{proof}
\emph{The rational case.}\\
According to equation \eqref{eqtype0}, a rational first integral is simply $F$, and so defined uniquely by equation \eqref{eqtype0}.\\

\emph{The Darbouxian case.}\\
A $k$-Darbouxian first integral $\mathcal{F}$ satisfies an equation \eqref{eqtype1}. Indeed, with our definition we have $\partial_y \mathcal{F}=F$. We also know that it should satisfy the equation of first integrals. This gives
$$\partial_y \mathcal{F}=F(x,y),\qquad A(x,y) \partial_x \mathcal{F}(x,y)+B(x,y) \partial_y \mathcal{F}(x,y)=0.$$
Thus we know the derivative of $\mathcal{F}$ with respect to $x$ and $y$, and so \eqref{eqtype1} defines $\mathcal{F}$ up to an addition of a constant.\\

\emph{The Liouvillian case.}\\
A Liouvillian first integral $\mathcal{F}$ satisfies an equation \eqref{eqtype2}. Setting
$$R(x,y)=\partial_y \mathcal{F}(x,y),$$
 equation \eqref{eqtype2} becomes
$$\frac{\partial_y R}{R}=F(x,y).$$

We now use the first integral equation \eqref{eqint}, dividing it by $A$ and differentiating with respect to $y$, which gives
$$\partial_x\partial_y\mathcal{F}+\partial_y \left(\frac{B}{A} \partial_y\mathcal{F} \right)=0,$$
$$\Rightarrow \partial_x R+\partial_y \left(\frac{B}{A} R \right)=0,$$
$$ \Rightarrow \partial_x R+\partial_y \left(\frac{B}{A}\right) R + \frac{B}{A} F R=0,$$
$$\Rightarrow \frac{\partial_x R}{R}=-\partial_y \left(\frac{B}{A}\right) - \frac{B}{A} F.$$
Therefore we know the logarithmic derivatives of $R$ with respect to $x$ and $y$, and thus we obtain $R$ up to a multiplication by a constant. Then as we have 
$$\partial_y \mathcal{F}(x,y)=R(x,y), \quad \partial_x \mathcal{F}(x,y)=-\frac{B}{A}(x,y) R(x,y)$$
we obtain $\mathcal{F}$ from $R$ up to addition of a constant. Thus equation \eqref{eqtype2} defines $\mathcal{F}$ up to an affine transformation.\\

\emph{The Riccati case.}\\
A Riccati first integral is a quotient of two solutions  $\mathcal{F}_1,\mathcal{F}_2$ of equation \eqref{eqtype3} that are independent over $\overline{\KK(x)}$. Knowing that the quotient is a first integral, we have moreover
$$D_0\left(\frac{\mathcal{F}_1}{\mathcal{F}_2}\right)=\frac{1}{\mathcal{F}_2^2} \left( D_0(\mathcal{F}_1)\mathcal{F}_2-\mathcal{F}_1D_0(\mathcal{F}_2)\right)=0$$
and thus
$$\frac{D_0(\mathcal{F}_1)}{\mathcal{F}_1}=\frac{D_0(\mathcal{F}_2)}{\mathcal{F}_2}$$
Let us note $\Omega= D_0(\mathcal{F}_i)/\mathcal{F}_i$ (for $i=1$ or $2$ as they are equal), the functions $\mathcal{F}_1,\mathcal{F}_2$ are solutions of the PDE system
\begin{equation}\label{eqric0}
\partial_y^2 \mathcal{F}_i-F(x,y) \mathcal{F}_i=0, \quad D_0(\mathcal{F}_i)-\Omega(x,y)\mathcal{F}_i=0.
\end{equation}

 Let us now consider a solution of this system.\\
  Due to the first equation, we can write it
$$C_1(x)\mathcal{F}_1(x,y)+C_2(x)\mathcal{F}_2(x,y).$$
Now substituting this in the second equation gives
$$C_1\big(D_0(\mathcal{F}_1)-\Omega\mathcal{F}_1\big)+C_2\big(D_0(\mathcal{F}_2)-\Omega\mathcal{F}_2\big)+A\partial_x C_1\mathcal{F}_1+A\partial_x C_2\mathcal{F}_2=0.$$
Thus
$$\partial_x C_1(x)\mathcal{F}_1(x,y)+\partial_x C_2(x)\mathcal{F}_2(x,y)=0.$$
As the functions $\mathcal{F}_1,\mathcal{F}_2$ are  independent over the functions in $x$, the previous equality implies that we have \mbox{$\partial_x C_1(x)=0$}, $\partial_x C_2(x)=0$, and so $C_1,C_2$ are constants. Thus the dimension over $\KK$ of the vector space of solutions of \eqref{eqric0} is exactly $2$. So the system \eqref{eqric0} defines $\mathcal{F}_1,\mathcal{F}_2$ up to a change of basis. This change of basis acts on the quotient $\mathcal{F}_1/\mathcal{F}_2$ as a homography.
\end{proof}

The canonical equations of the output of our algorithm thus define the first integral up to a composition by a simple single variable function.\\

Below, we give a necessary and sufficient criterion for ensuring that an equation \eqref{eqtype0}, \eqref{eqtype1}, \eqref{eqtype2}, \eqref{eqtype3} leads to a first integral.\\

\begin{prop}\label{proprepresent}$\,$
\begin{itemize}
\item A solution of equation \eqref{eqtype0} is a rational first integral if and only if
$$D_0(F)=0,\;F \in \overline{\KK}(x,y) \setminus \overline{\KK}.$$
\item A solution of equation \eqref{eqtype1} is a $k$-Darbouxian first integral if and only if 
$$D_0(F)=-AF\partial_y(B/A), \; F^k \in \overline{\KK}(x,y)\setminus \{0\}.$$
\item 
A solution of equation \eqref{eqtype2} is a Liouvillian first integral if and only if 
$$D_0(F)=-A\partial_y(B/A)F-A\partial_y^2(B/A), \; F \in \overline{\KK}(x,y).$$
\item
A solution of equation \eqref{eqtype3} is a Riccati first integral if and only if
$$D_0(F)=-2A\partial_y(B/A)F+\frac{1}{2}A\partial_y^3(B/A), \; F \in \overline{\KK}(x,y).$$
\end{itemize}
\end{prop}

\begin{proof}\underline{\emph{The rational case}}:\\
This is  the definition of rational first integrals.\\

\underline{\emph{The $k$-Darbouxian case}}:\\
 In the $k$-Darbouxian case, we have
$$\partial_y \mathcal{F}=F(x,y),\quad D_0(\mathcal{F})=0.$$
So this is equivalent to
$$\partial_y \mathcal{F}=F(x,y),\quad \partial_x \mathcal{F}=-\frac{B}{A}F.$$
So a necessary and sufficient condition for a Darbouxian first integral $\mathcal{F}$ to exist is the closed form condition,
\begin{eqnarray*}
\partial_x F=-\partial_y\left(\frac{B}{A}F \right)& \iff &\partial_xF=-\partial_y\Big(\dfrac{B}{A}\Big)F-\dfrac{B}{A}\partial_y F\\
&\iff & A \partial_x F + B \partial_y F=-A \partial_y\Big(\dfrac{B}{A}\Big)F\\& \iff & D_0(F)=-A \partial_y\Big(\dfrac{B}{A}\Big)F
\end{eqnarray*}
which gives the condition of the proposition.\\ 

\underline{\emph{The Liouvillian case}}:\\
In the Liouvillian case, the first integral $\mathcal{F}$ has to solve the PDE system
$$(\star) \quad\partial_y^2 \mathcal{F}-F\partial_y \mathcal{F}=0,\quad \quad (\star \star)\quad D_0(\mathcal{F})=0.$$
The derivative relatively to $x$ of $(\star)$ gives:
$$(\mathcal{A})\quad\partial_x \partial_y^2 \mathcal{F}- \partial_xF \partial_y \mathcal{F}-F \partial_x\partial_y \mathcal{F}=0.$$
The derivative relatively to $y$ of $(\star\star)$ divided by $A$ and then simplified thanks to $(\star)$ gives:
$$\partial_y\partial_x \mathcal{F}+\partial_y\Big(\dfrac{B}{A} \Big) \partial_y \mathcal{F}+\dfrac{B}{A}F\partial_y\mathcal{F}=0.$$
The derivative relatively to $y$ of the previous equality gives:
$$(\mathcal{B})\quad \partial_y^2\partial_x \mathcal{F}+\partial_y^2\Big(\dfrac{B}{A}\Big) \partial_y \mathcal{F}+ \partial_y\Big(\dfrac{B}{A}\Big) \partial_y^2 \mathcal{F}+\partial_y\Big(\dfrac{B}{A}F\Big) \partial_y \mathcal{F}+ \dfrac{B}{A}F\partial_y^2\mathcal{F}=0.$$

The difference $(\mathcal{A})-(\mathcal{B})$ simplified thanks to $(\star)$ gives:
$$-\partial_x F \partial_y \mathcal{F} - F \partial_x \partial_y \mathcal{F} - \partial_y^2\Big(\dfrac{B}{A} \Big) \partial_y\mathcal{F}
-2\partial_y\Big(\dfrac{B}{A}\Big)F\partial_y \mathcal{F}
-\dfrac{B}{A} \partial_y F \partial_y \mathcal{F}-\dfrac{B}{A}F \partial_y^2 \mathcal{F}=0.$$
The equation $(\star \star)$ implies:
\begin{eqnarray*}
0&=& -\partial_x F \partial_y \mathcal{F} - F  \partial_y\Big(-\dfrac{B}{A}\partial_y \mathcal{F}\Big) - \partial_y^2\Big(\dfrac{B}{A} \Big) \partial_y\mathcal{F} -2 \partial_y\Big( \dfrac{B}{A} \Big) F \partial_y\mathcal{F}
-\dfrac{B}{A} \partial_y F \partial_y \mathcal{F}\\
&&-\dfrac{B}{A}F \partial_y^2 \mathcal{F}\\
&=& -\left(\partial_x F+\frac{B}{A}\partial_y F +\partial_y(B/A)F+\partial_y^2(B/A) \right) \partial_y \mathcal{F}. \\
\end{eqnarray*}
If $\partial_y \mathcal{F}=0$  then $\mathcal{F}$ only depend on $x$. This is impossible as this would imply $A=0$ or $\mathcal{F}$ constant. So the only possibility left is
$$\partial_x F+\frac{B}{A}\partial_y F +\partial_y(B/A)F+\partial_y^2(B/A)=0. $$
This is the condition of the proposition.\\

Conversely, we suppose that $D_0(F)=-A\partial_y(B/A)F-A\partial_y^2(B/A)$. We are going to prove that in this situation $D_0$ has a Liouvillian first integral with integrating factor $\mathfrak{R}=e^{\mathfrak{F}}$, where $\mathfrak{F}$ is the integral of a rational closed $1$-form. We set
$$\Omega_1=F-\dfrac{\partial_y A}{A}, \quad \Omega_2=\dfrac{-div-B\Omega_1}{A}.$$
We are going to show that $\partial_x (\Omega_1)= \partial_y(\Omega_2)$. Then this implies that there exists a Darbouxian function $\mathfrak{F}$ such that $\Omega_1=\partial_y\mathfrak{F}$ and $\Omega_2=\partial_x \mathfrak{F}$. By construction we have $A \Omega_2+B\Omega_1=-div$, and so $B\mathfrak{R}dx-A\mathfrak{R}dy$ is closed, where $\mathfrak{R}=e^{\mathfrak{F}}$ . Therefore $\mathfrak{R}$ is the integrating factor of a Liouvillian first integral, and we get the desired result.\\
Thus, now we are going to prove that $\partial_x (\Omega_1)= \partial_y(\Omega_2)$. We have:
\begin{eqnarray*}
\partial_x(\Omega_1)- \partial_y (\Omega_2)&=& \partial_xF- \partial_x\Big(\dfrac{\partial_y A}{A}\Big)\\
&&+\partial_y\Big(\dfrac{div}{A}\Big)+ \partial_y\Big(\dfrac{B}{A}F\Big) -\partial_y\Big(\dfrac{B}{A}\dfrac{\partial_y A}{A} \Big).\\
\end{eqnarray*}
As $$-\partial_x \Big( \dfrac{\partial_y A}{A} \Big) +\partial_y \Big( \dfrac{ div}{A} \Big)= \dfrac{\partial_y^2 B}{A}- \dfrac{\partial_y B \partial_y A}{A^2},$$
we get
\begin{eqnarray*}
\partial_x(\Omega_1)- \partial_y (\Omega_2)&=&\partial_xF+\dfrac{\partial_y^2 B}{A}- \dfrac{\partial_y B \partial_y A}{A^2} + \partial_y\Big(\dfrac{B}{A}\Big)F + \dfrac{B}{A}\partial_yF\\
&&-\partial_y\Big(\dfrac{B}{A}\Big)\dfrac{\partial_y A}{A}
- \dfrac{B}{A} \partial_y\Big(\dfrac{\partial_y A}{A}\Big).\\
\end{eqnarray*}

By rearranging the terms of the equation, we get
\begin{eqnarray*}
\partial_x(\Omega_1)- \partial_y (\Omega_2)&=& \partial_x F +\Big(\dfrac{B}{A}\Big)\partial_yF+\partial_y\Big(\dfrac{B}{A}\Big)F\\
&&+\dfrac{\partial_y^2 B}{A}-\dfrac{\partial_y B \partial_y A}{A^2}- \partial_y\Big(\dfrac{B}{A}\Big) \dfrac{\partial_y A}{A}- \dfrac{B}{A}\partial_y\Big(\dfrac{\partial_y A}{A}\Big).
\end{eqnarray*}
Now, we note that
$$ \partial_y^2 \Big(\dfrac{B}{A}\Big)= \dfrac{\partial_y^2 B}{A}-\dfrac{\partial_y B \partial_y A}{A^2}- \partial_y\Big(\dfrac{B}{A}\Big) \dfrac{\partial_y A}{A}- \dfrac{B}{A}\partial_y\Big(\dfrac{\partial_y A}{A}\Big).$$
Thus 
$$\partial_x(\Omega_1)- \partial_y (\Omega_2)= \partial_x F +\Big(\dfrac{B}{A}\Big)\partial_yF+\partial_y\Big(\dfrac{B}{A}\Big)F +\partial_y^2 \Big(\dfrac{B}{A}\Big).$$
By hypothesis the right hand side of this equation is equal to zero. This gives the desired result.\\
 
\underline{\emph{The Riccati case}}:\\
In the Riccati case, a first integral is a quotient of two functions $\mathcal{F}_1,\mathcal{F}_2$, solutions of a PDE system of the form
$$(\Diamond)\quad \partial_y^2 \mathcal{F}_i-F \mathcal{F}_i=0,\quad \quad (\Diamond \Diamond) \quad D_0(\mathcal{F}_i)-\Omega\mathcal{F}_i=0.$$
The derivative relatively to $x$ of $(\Diamond)$ and then simplified by $(\Diamond \Diamond)$ gives:
$$(\mathcal{C}) \quad \partial_x \partial_y^2 \mathcal{F}_i - \partial_xF .\mathcal{F}_i-F\Big( \dfrac{\Omega}{A} \mathcal{F}_i-\dfrac{B}{A} \partial_y\mathcal{F}_i\Big)=0.$$
The derivative relatively to $y$ of $(\Diamond \Diamond)$ divided by $A$ and then simplified thanks to $(\Diamond)$ gives:
$$\partial_y \partial_x \mathcal{F}_i + \partial_y\Big(\dfrac{B}{A}\Big) \partial_y \mathcal{F}_i+\dfrac{B}{A}F \mathcal{F}_i-\partial_y\Big(\dfrac{\Omega}{A}\Big) \mathcal{F}_i-\dfrac{\Omega}{A}\partial_y \mathcal{F}_i=0.$$
The derivative relatively to $y$ of the previous equality and simplified thanks to $(\Diamond)$ gives:
\begin{eqnarray*}
(\mathcal{D})\quad 0 &=&\partial_y^2 \partial_x \mathcal{F}_i + \partial_y^2\Big( \dfrac{B}{A} \Big) \partial_y \mathcal{F}_i + \partial_y\Big( \dfrac{B}{A} \Big) F \mathcal{F}_i+ \partial_y\Big(\dfrac{B}{A}F\Big)\mathcal{F}_i+\dfrac{B}{A}F\partial_y \mathcal{F}_i\\
&& -\partial_y^2\Big(\dfrac{\Omega}{A}\Big) \mathcal{F}_i 
-\partial_y\Big(\dfrac{\Omega}{A} \Big) \partial_y\mathcal{F}_i - \partial_y\Big( \dfrac{\Omega}{A}\Big) \partial_y \mathcal{F}_i- \dfrac{\Omega}{A}F \mathcal{F}_i.\\
\end{eqnarray*}
The difference $(\mathcal{C})-(\mathcal{D})$ gives:
$$\left(-\partial_x F - \partial_y\Big(\dfrac{B}{A}\Big)F- \partial_y\Big( \dfrac{B}{A}F\Big)+\partial_y^2\Big( \dfrac{\Omega}{A}\Big) \right) \mathcal{F}_i+ \left( 2\partial_y\Big( \dfrac{\Omega}{A}\Big) - \partial_y^2\Big(\dfrac{B}{A}\Big) \right) \partial_y \mathcal{F}_i=0.$$

The functions $\mathcal{F}_1,\mathcal{F}_2$ are independent. Thus the Wronskian of $\mathcal{F}_1,\mathcal{F}_2$ in $y$ is not $0$. Therefore, $(\mathcal{F}_1,\partial_y \mathcal{F}_1)$ and $(\mathcal{F}_2, \partial_y \mathcal{F}_2)$ are linearly independent. The above linear form thus vanishes on them and is identically $0$, then
$$2\partial_y \left(\frac{\Omega}{A}\right)-\partial_y^2 \left(\frac{B}{A}\right)=0,$$
$$-\partial_x F-\frac{B}{A}\partial_y F-2\partial_y\Big(\dfrac{B}{A}\Big)F+\partial_y^2\Big(\dfrac{\Omega}{A}\Big)=0. $$
Therefore, by substituting the first  equation into the second, we get:
$$\partial_x F+\frac{B}{A}\partial_y F=-2\partial_y\Big(\dfrac{B}{A}\Big)F+\dfrac{1}{2}\partial_y^3\Big(\dfrac{B}{A}\Big).$$
This is the condition given by the proposition.\\

Conversely, let us prove that if the equation 
$$D_0(F)=-2A \partial_y\Big(\dfrac{B}{A} \Big) F +\dfrac{1}{2} A \partial_y^3\Big(\dfrac{B}{A} \Big)$$
 is satisfied, then \eqref{eqtype3} leads to a Riccati first integral.\\
 
  Let us choose
$$\Omega=\frac{1}{2}A\partial_y \left(\frac{B}{A}\right)$$
and prove that the system
\begin{equation}\label{eqricsys}
\partial_y^2 \mathcal{F}-F \mathcal{F}=0,\quad D_0(\mathcal{F})-\Omega\mathcal{F}=0
\end{equation}
has two independent solutions. The quotient of these two solutions gives then a Riccati first integral.\\

The strategy used to prove the existence of these two independent solutions is the following:\\
We consider $\mathcal{F}_1$, $\mathcal{F}_2$ a basis of solutions of $\partial_y^2 \mathcal{F}-F(x,y) \mathcal{F}=0$, i.e. two solutions independent over \emph{the constant field of functions in $x$}. Then, we are going to prove that there exists two independent couples $(E_1,E_2)$ of functions in $x$ such that $E_1 \mathcal{F}_1+E_2 \mathcal{F}_2$ satisfies the second equation of (\ref{eqricsys}). This gives then two independent solutions of equation  (\ref{eqricsys}).\\

In order to follow our strategy, we  need an explicit  expression of $\partial_x \mathcal{F}_1$ and $\partial_x \mathcal{F}_2$. These expressions will be useful when we will consider the equation $D_0(\mathcal{F})- \Omega \mathcal{F}=0$.\\
Differentiating in $x$ the equation $\partial_y^2 \mathcal{F} -F \mathcal{F}=0$ gives
$$\partial_y^2 \partial_x \mathcal{F}-\big(\partial_xF\big) \mathcal{F}-F\partial_x \mathcal{F}=0.$$
Setting $\mathcal{G}=\partial_x \mathcal{F}$, the previous equation can be rewritten
$$(\sharp)\quad \partial_y^2 \mathcal{G}-F\mathcal{G}=\big(\partial_xF\big) \mathcal{F}.$$
This is a linear differential equation with a non homogeneous term $\big(\partial_xF\big) \mathcal{F}$. %Let us now consider $\mathcal{F}_1,\mathcal{F}_2$ a basis of solutions of $\partial_y^2 \mathcal{F}-F(x,y) \mathcal{F}=0$, i.e. two solutions independent over \emph{the constant field of functions in $x$}.

Now, we solve equation $(\sharp)$ with $\mathcal{F}=\mathcal{F}_i,\;i=1,2$, this gives:
$$(\sharp\sharp)\quad \partial_y^2 \mathcal{G}_i-F\mathcal{G}_i=\big(\partial_xF) \mathcal{F}_i(x,y).$$
We already know a basis of solutions of the homogeneous part, and we guess as particular solution
$$\mathcal{G}_i=-\frac{B}{A} \partial_y \mathcal{F}_i+\frac{1}{2}\mathcal{F}_i\partial_y \left(\frac{B}{A}\right).$$

Indeed, thanks to the relation $\partial_y^2 \mathcal{F}-F\mathcal{F}=0$, we get:
$$\partial_y \mathcal{G}_i=-\dfrac{1}{2}\partial_y\Big(\dfrac{B}{A} \Big) \partial_y \mathcal{F}_i - \dfrac{B F}{A} \mathcal{F}_i+\dfrac{1}{2}\mathcal{F}_i \partial_y^2\Big(\dfrac{B}{A} \Big).$$
The derivation relatively to $y$ of the previous equation gives:\\
\begin{eqnarray*}
\partial_y^2 \mathcal{G}_i& =& - \dfrac{1}{2} \partial_y\Big(\dfrac{B}{A} \Big)F \mathcal{F}_i- \partial_y\Big(\dfrac{BF}{A} \Big) \mathcal{F}_i -\dfrac{BF}{A}\partial_y\mathcal{F}_i+\dfrac{1}{2} \mathcal{F}_i \partial_y^3\Big(\dfrac{B}{A}\Big).\\
\end{eqnarray*}
By rearranging the terms we get
\begin{eqnarray*}
\partial_y^2 \mathcal{G}_i& =&\dfrac{1}{2} \partial_y\Big(\dfrac{B}{A} \Big)F\mathcal{F}_i - \dfrac{BF}{A} \partial_y \mathcal{F}_i  -\partial_y\Big(\dfrac{B}{A} \Big)F\mathcal{F}_i  -\partial_y \Big(\dfrac{BF}{A}\Big)\partial_y\mathcal{F}_i+\dfrac{1}{2} \mathcal{F}_i \partial_y^3\Big(\dfrac{B}{A}\Big).
\end{eqnarray*}
Thanks to the definition of $\mathcal{G}_i$ we get:

\begin{eqnarray*}
\partial_y^2 \mathcal{G}_i&=&F\mathcal{G}_i+\left( -F\partial_y \Big( \dfrac{B}{A} \Big) - \partial_y\Big( \dfrac{BF}{A} \Big) +\dfrac{1}{2} \partial_y^3\Big(\dfrac{B}{A} \Big)\right)\mathcal{F}_i\\
&=&F\mathcal{G}_i+\left( -2\partial_y \Big( \dfrac{B}{A} \Big)F - \dfrac{B}{A}\partial_y F +\dfrac{1}{2} \partial_y^3\Big(\dfrac{B}{A} \Big)\right)\mathcal{F}_i.
\end{eqnarray*}

As by hypothesis, we have $D_0(F)=-2A \partial_y\Big(\dfrac{B}{A} \Big) F +\dfrac{1}{2} A \partial_y^3\Big(\dfrac{B}{A} \Big)$, this implies
$$ -2\partial_y \Big( \dfrac{B}{A} \Big)F - \dfrac{B}{A}\partial_y F +\dfrac{1}{2} \partial_y^3\Big(\dfrac{B}{A}\Big)=\partial_x F.$$
This proves that $\mathcal{G}_i$ is a particular solution of ($\sharp \sharp$).

So the solutions of $(\sharp\sharp)$ are respectively for $i=1,2$ of the form
$$\mathcal{G}_1=C_1\mathcal{F}_1+C_2\mathcal{F}_2-\frac{B}{A} \partial_y \mathcal{F}_1+\frac{1}{2}\mathcal{F}_1\partial_y \left(\frac{B}{A}\right)$$
$$\mathcal{G}_2=C_3\mathcal{F}_1+C_4\mathcal{F}_2-\frac{B}{A} \partial_y \mathcal{F}_2+\frac{1}{2}\mathcal{F}_2\partial_y \left(\frac{B}{A}\right)$$
where the $C_i$ \emph{depend on $x$ only}. So we deduce that there exists $C_1,C_2,C_3,C_4$ functions of $x$ only such that
$$\partial_x \mathcal{F}_1=C_1\mathcal{F}_1+C_2\mathcal{F}_2-\frac{B}{A} \partial_y \mathcal{F}_1+\frac{1}{2}\mathcal{F}_1\partial_y \left(\frac{B}{A}\right),$$
$$\partial_x \mathcal{F}_2=C_3\mathcal{F}_1+C_4\mathcal{F}_2-\frac{B}{A} \partial_y \mathcal{F}_2+\frac{1}{2}\mathcal{F}_2\partial_y \left(\frac{B}{A}\right).$$
We now search solutions of equations \eqref{eqricsys} of the form
$$E_1\mathcal{F}_1+E_2\mathcal{F}_2$$
with $E_1,E_2$ functions of $x$ only. Substituting it in equations \eqref{eqricsys}, we obtain $0$ for the first, and for the second
$$(E_1C_1+E_2C_3+\partial_x E_1)\mathcal{F}_1+(E_1C_2+E_2C_4+\partial_x E_2)\mathcal{F}_2=0.$$
As $\mathcal{F}_1,\mathcal{F}_2$ are independent over functions in $x$, this is equivalent to the system 
$$\partial_x E_1=-E_1C_1-E_2C_3,\quad \partial_x E_2=-E_1C_2-E_2C_4.$$
This is a $2\times 2$ linear differential system, and so admits two independent solutions. Then these two solutions $E_1$, $E_2$ give two independent solutions of equations \eqref{eqricsys} of the form $E_1\mathcal{F}_1+E_2\mathcal{F}_2$. Their quotient is then a first integral.
\end{proof}
 During the previous proof we have shown the following:
 \begin{cor}\label{cor:cofacteurRiccati}
 If $D_0$ has a Riccati first integral $\mathcal{F}_1/\mathcal{F}_2$ then we can suppose that 
 $$D_0(\mathcal{F}_i)=\dfrac{1}{2}A\partial_y\Big( \dfrac{B}{A}\Big) \mathcal{F}_i.$$
 \end{cor}
 
%%%%%%%%%%%%%%%%%%%%%%%%%%%%%%%%%%%%%%%%%%%
\subsection{Casale's Theorem}

We have defined $4$ types of first integrals. We are going to prove that there are no other types of first integrals with algebraic-differential properties. Recall that the flow is defined by
$$\quad \partial_x y(x_0,y_0;x)=\frac{B\big(x,y(x_0,y_0;x)\big)}{A\big(x,y(x_0,y_0;x)\big)},  \textrm{ and }  y(x_0,y_0;x_0)=y_0. $$
We are also interested in $y_r(x)=\partial_{y_0}^ry(x_0,y_0;x),\; r=1,2, 3$. These functions belong to $\KK(x_0,y_0)[[x-x_0]]$.
%and we note them $y(x_0,y_0;x),y_1(x_0,y_0;x),y_2(x_0,y_0;x)$, $y_3(x_0,y_0;x)$ respectively. Sometimes we will not write the dependence on $x_0,y_0$.
 As explained in the introduction,  $y,y_1,y_2,y_3$ are  seen as functions in $x$, solutions of some differential systems $(S'_r)$,
%$$(S'_3)\,
% \begin{cases}
%(S'_2) 
%\begin{cases}
%(S'_1) \,
%\begin{cases}
%(S'_0) \, \quad \partial_x y=\dfrac{B}{A}\\
%\,\\
%\partial_xy_1= y_1\partial_y \left(\dfrac{B}{A}\right)\\
% \end{cases}\\
%\, \\
%\partial_x y_2=y_2\partial_y \left(\dfrac{B}{A}\right)+y_1^2\partial_y^2 \left(\dfrac{B}{A}\right)\\
%
%\end{cases}\\
%\, \\
%\partial_x y_3=y_3\partial_y \left(\dfrac{B}{A}\right)+3y_2y_1\partial_y^2 \left(\dfrac{B}{A}\right)+y_1^3\partial_y^3 \left(\dfrac{B}{A}\right)\\
%\end{cases}\\
%$$
%and by construction,
 and their initial conditions are
$$y(x_0)=y_0,\, y_1(x_0)=1,\, y_2(x_0)=0,\,  y_3(x_0)=0.$$
%The system $(S')$ gives a method to compute the flow $y(x_0,y_0;x)$ and finitely many of its derivatives as series in $x$: we only have to solve $(S')$ using the Newton method for initial condition $x=x_0,y=y_0,y_1=1,y_2=0,y_3=0$. \\

If $\mathcal{F}$ is a first integral of $D_0$ then we have $\mathcal{F}\big(x,y(x_0,y_0;x)\big)=\mathcal{F}(x_0,y_0)$. As mentioned in the introduction, the derivation relatively to $y_0$ of this relation gives with our notations:
$$(\star)\quad \partial_y \mathcal{F}\big(x,y(x)\big)y_1(x)=\partial_{y_0}\mathcal{F}(x_0,y_0).$$
Therefore if $\mathcal{F}$ is a Darbouxian first integral, we have $\partial_y \mathcal{F}=F \in \overline{\KK}(x,y)$ and then
$$F\big(x,y(x)\big)y_1(x)=F(x_0,y_0).$$

Thus the rational function $F(x,y)y_1 \in \KK(x,y,y_1)$ is constant on $\big(x,y(x),y_1(x)\big)$, where the initial condition is $y(x_0)=y_0$ and $y_1(x_0)=1$. In Proposition \ref{propinv}, we prove that $F(x,y)y_1$ is also a rational first integral for $(S'_1)$.\\
 
 The derivation relatively to $y_0$ of equation $(\star)$ gives
% $$\partial_{y}^2 \mathcal{F}\big(x,y(x_0,y_0;x)\big)\Big(\partial_{y_0} y(x_0,y_0;x)\Big)^2+ \partial_{y} \mathcal{F}\big(x,y(x_0,y_0;x)\big) \partial_{y_0}^2 y(x_0,y_0;x)=\partial_{y_0}^2 \mathcal{F}(x_0,y_0).$$
$$\partial_{y}^2 \mathcal{F}\big(x,y(x)\big)\Big(y_1(x)\Big)^2+ \partial_{y} \mathcal{F}\big(x,y(x)\big)  y_2(x)=\partial_{y_0}^2 \mathcal{F}(x_0,y_0).$$
  If $ \mathcal{F}$ is Liouvillian then $\partial_y^2 \mathcal{F}/\partial_y \mathcal{F} = F$, and we get
 $$F\big(x,y(x)\big)y_1(x)+\dfrac{y_2(x)}{y_1(x)}=F(x_0,y_0).$$
We are also going to prove in Proposition \ref{propinv} that $F(x,y)y_1+y_2/y_1$ is a rational first integral of $(S'_2)$. \\
 
At last, for a Riccati first integral similar computations give a rational expression in $x,y,y_1,\dots, y_3$ which happens to be a rational first integral for $(S'_3)$.\\

The reason why we stop at $r=3$ is the following.

\begin{thm}[Casale]\label{thmcasale}
If there exists $J(x,y,y_1,\dots,y_n) \in \KK(x,y,y_1,\dots,y_n)$ such that
$$J\big(x_0,y(x_0),y_1(x_0),\dots, y_n(x_0)\big)=J\big(x,y(x),y_1(x),\dots,y_n(x)\big),$$
where $y_i(x)=\partial_{y_0}^iy(x_0,y_0;x)$,
then there exists a rational function \mbox{$h(x,y) \in \KK(x,y)$}  satisfying one of the following equalities:
\begin{itemize}
\item[$\bullet$] $h\big(x_0,y(x_0)\big)=h\big(x,y(x)\big)$,\\
\item[$\bullet$] $h\big(x_0,y(x_0)\big)=h\big(x,y(x)\big)y_1(x)^k$, with $k\in\mathbb{N}^*$,\\
\item[$\bullet$] $h\big(x_0,y(x_0)\big)=h\big(x,y(x)\big)y_1(x)+y_2(x)/y_1(x)$,\\
\item[$\bullet$] $
h\big(x_0,y(x_0)\big)=h\big(x,y(x)\big)y_1^2(x)+  2y_3(x)/y_1(x)
-3y_2^2(x)/y_1^2(x) 
$.
\end{itemize}
\end{thm}

This result is Proposition 1.18 and Theorem 1.19 of Casale in \cite{Casale} applied to the map $y_0 \mapsto \varphi(x_0,y_0;x)$, and restricted to the case of rational  instead of meromorphic invariants. Now,  Casale's invariants can be seen as first integrals of the systems $(S_r')$ and satisfy the equations $D_r(J)=0$.\\
 We will associate to each class of first integral a Casale invariant for the flow. This gives the following proposition stated in the introduction:\\

\textbf{Proposition \ref{propinv}.} %\label{propinv}
\begin{itemize}
\item
\emph{ The system $(S)$ admits a rational first integral associated to \eqref{eqtype0} if and only if $F(x,y)$ is a first integral of $(S'_0)$, where $F \in \overline{\KK}(x,y) \setminus \KK$.\\}

\item \emph{The system $(S)$ admits a $k$-Darbouxian first integral associated to \eqref{eqtype1} if and only if $y_1F(x,y)$ is a first integral of $(S'_1)$, where $F^k \in \overline{\KK}(x,y) \setminus \{0\}$.}\\

\item \emph{The system $(S)$ admits a Liouvillian first integral associated to \eqref{eqtype2} if and only if $y_1F(x,y)+y_2/y_1$ is a first integral of $(S'_2)$, where $F \in \overline{\KK}(x,y)$.}\\

\item\emph{The system $(S)$ admits a Riccati first integral associated to \eqref{eqtype3} if and only if $4y_1^2F(x,y)-2y_3/y_1+3y_2^2/y_1^2$ is a first integral of $(S'_3)$, where $F \in \overline{\KK}(x,y)$.}\\
\end{itemize}

In the Riccati case, we write our invariant in the form 
$$4y_1^2F(x,y)-2y_3/y_1+3y_2^2/y_1^2$$ and not in the form 
$$y_1^2F(x,y)-2y_3/y_1+3y_2^2/y_1^2$$
 as in Casale's theorem. We have chosen this expression because it leads to the Riccati equation (Ric): $\partial_y^2 \mathcal{F}- F \mathcal{F}=0$, whereas the expression used in Casale's theorem leads to the equation $\partial_y^2 \mathcal{F}- \frac{F}{4} \mathcal{F}=0$.

\begin{proof}
\underline{\emph{The rational case.}}\\
By definition, a rational first integral of $(S'_0)$ is a rational function $F(x,y)$ such that $F(x,y(x_0,y_0;x))=F(x_0,y_0)$. We deduce then 
\begin{eqnarray*}
&&\partial_x F\big(x, y(x_0,y_0;x)\big)+\partial_y  F\big(x, y(x_0,y_0;x)\big) \dfrac{B}{A}\big(x, y(x_0,y_0;x)\big)=0\\
\iff&& \partial_x F(x_0,y_0)+\partial_y  F\big(x_0,y_0) \dfrac{B}{A}(x_0,y_0)=0\\
\iff&& D_0(F)=0.
\end{eqnarray*}
This gives the desired conclusion in the rational case.\\

\underline{\emph{The $k$-Darbouxian case.}}\\
We have: $y_1F(x,y)$ is a first integral of $(S'_1)$ if and only if $D_1\big(y_1F\big)=0$, where $D_1=A^2\partial_x+AB \partial_y+y_1A^2 \partial_y\Big(\dfrac{B}{A}\Big) \partial_{y_1}$.\\
 This gives: 
\begin{eqnarray*}
D_1\big(y_1F\big)=0&\iff & y_1\left( \partial_xF+\dfrac{B}{A}\partial_yF+\partial_y\Big(\dfrac{B}{A}\Big)F\right)=0\\
&\iff& D_0(F)=-A\partial_y\Big(\dfrac{B}{A}\Big)F.
\end{eqnarray*}
Then, Proposition \ref{proprepresent} gives the desired conclusion.\\

\underline{\emph{The Liouvillian case.}}\\
We have:\\
 $y_1F(x,y)+y_2/y_1$ is a first integral of $(S'_2)$ if and only if $D_2(y_1F+y_2/y_1)=0$, and we recall that\\  \mbox{$D_2= A^3\partial_x+A^2B \partial_y+y_1A^3 \partial_y\Big(\dfrac{B}{A}\Big) \partial_{y_1} +A^3\Big( y_2\partial_y\Big( \dfrac{B}{A} \Big)+y_1^2 \partial^2_{y} \Big( \dfrac{B}{A} \Big)\Big) \partial_{y_2}$}.\\
 This gives:
\begin{eqnarray*}
D_2(y_1F+y_2/y_1)=0&\iff& y_1\left( \partial_xF+\dfrac{B}{A}\partial_yF+\partial_y\Big(\dfrac{B}{A}\Big)F+\partial_y^2\Big(\dfrac{B}{A}\Big)\right)=0\\
&\iff& D_0(F)=-A\partial_y\Big(\dfrac{B}{A}\Big)F -A \partial_y^2\Big(\dfrac{B}{A}\Big).
\end{eqnarray*}
As before we get the desired conclusion thanks to Proposition \ref{proprepresent}.\\

\underline{\emph{The Riccati case.}}\\
We consider the derivation
\begin{eqnarray*}
 D_3&=&A^4\partial_x + A^3B\partial_y+ y_1A^4\partial_y \left(\dfrac{B}{A}\right) \partial_{y_1}+A^4\Big(y_2\partial_y \left(\dfrac{B}{A}\right)+y_1^2\partial_y^2 \left(\dfrac{B}{A}\right)\Big) \partial_{y_2}\\
 &&+A^4\Big(y_3\partial_y \left(\frac{B}{A}\right)+3y_2y_1\partial_y^2 \left(\dfrac{B}{A}\right)+y_1^3\partial_y^3 \left(\dfrac{B}{A}\right)\Big)\partial_{y_3},\end{eqnarray*}
and we have: $4y_1^2F(x,y)-2y_3/y_1+3y_2^2/y_1^2$ is a first integral of $(S'_3)$ if and only if $D_3(4y_1^2F(x,y)-2y_3/y_1+3y_2^2/y_1^2)=0$.\\
 
 As $D_3(-2y_3/y_1+3y_2^2/y_1^2)= -2y_1^2 \partial_y^3(B/A)$, we get:
\begin{eqnarray*}
& &D_3(4y_1^2F(x,y)-2y_3/y_1+3y_2^2/y_1^2)=0\\
&\iff& y_1^2\left( \partial_xF+\dfrac{B}{A}\partial_yF+2A\partial_y\Big(\dfrac{B}{A}\Big)F-\frac{1}{2}A\partial_y^3\Big(\dfrac{B}{A}\Big)  \right)=0\\
&\iff& D_0(F)=-2A\partial_y\Big(\dfrac{B}{A}\Big)F+\frac{1}{2}A\partial_y^3\Big(\dfrac{B}{A}\Big).
\end{eqnarray*}
We conclude thanks to Proposition \ref{proprepresent}.
\end{proof}
Now, we recall some definitions, see \cite[Definition 4.1]{Casale}.
\begin{defi}\label{def:ric_extension}
Let $(K; \partial_x, \partial_y)$ be a differential field. An algebraic extension $L \supset K$ is a differential field such that $L=K(f)$ with $f$ algebraic over $K$. An exponential extension $L \supset K$ is a differential field such that $L=K(\exp f)$ with $f\in K$.\\
 A primitive extension $L \supset K$ is a differential field such that $L=K(f)$ with 
$$df=\partial_x f dx+\partial_yf dy$$ a $1$-form with coefficients in $K$, i.e. $\partial_x f$ and $\partial_y f$ belong to $K$.\\
A Riccati extension $L \supset K$ is a differential field such that $L=K(f)$ with
$df$ a $1$-form with coefficients in $K[f]_{\leq 2}$.
\end{defi}

\begin{prop}\label{prop:structure_successive_extension}$\,$
\begin{itemize}
\item The system $(S)$ admits a first integral in a field built by successive algebraic and primitive extensions  over $\KK(x,y)$ if and only if it admits a $k$-Darbouxian first integral.
\item The system $(S)$ admits a first integral in a field built by successive algebraic, exponential, primitive extensions,  over $\KK(x,y)$ if and only if it admits a Liouvillian first integral.
\item The system $(S)$ admits a first integral in a field built by successive algebraic extensions, exponential, primitive and Riccati extensions over $\KK(x,y)$ if and only if it admits a Riccati first integral.
\end{itemize}
\end{prop}

\begin{proof}
\underline{\emph{The Darbouxian case.}}\\
If $(S)$ admits a first integral built by successive algebraic and primitive extensions  over $\KK(x,y)$, then by Theorem~4.2, Theorem~1.19 and Proposition~1.18  of Casale \cite{Casale} there exists $k\in\mathbb{N}^*, F\in\overline{\KK}(x,y)$ such that $y_1^kF(x,y)$ is a first integral of $(S_1')$. By Proposition \ref{propinv}, then $(S)$ admits a $k$-Darbouxian first integral. The converse is immediate as a $k$-Darbouxian first integral is the integral of an algebraic $1$-form.\\

\underline{\emph{The Liouvillian case.}}\\
This case corresponds to Singer's result, see \cite{Singer}.\\

\underline{\emph{The Riccati case.}}\\
Thanks to  Theorem~4.2, Theorem~1.19 and Proposition~1.18  of Casale \cite{Casale}, if $(S)$ admits a first integral in a field built by successive algebraic  exponential,  primitive and Riccati extensions over $\KK(x,y)$, then  
$$4y_1^2F(x,y)-2y_3/y_1+3y_2^2/y_1^2$$
is a first integral of $(S'_3)$. Then, Proposition \ref{propinv} implies that $(S)$ admits a Riccati first integral.\\
 For the converse, we have a first integral which is the quotient of two solutions $\mathcal{F}_1,\mathcal{F}_2$ of a linear second order differential equation in $y$. We also know, thanks to Corollary \ref{cor:cofacteurRiccati}, that $\mathcal{F}_1,\mathcal{F}_2$ satisfy the equation
$$D_0(\mathcal{F})=\frac{1}{2}A\partial_y \left(\frac{B}{A}\right)\mathcal{F},$$
and thus writing $\partial_y\mathcal{F}$ as a function of $\partial_x \mathcal{F}$, we obtain another linear second order differential equation in $x$ of the following kind: $\partial_x^2 \mathcal{F}=R\partial_x\mathcal{F}+S\mathcal{F}$, where $R,S$ belong to $\KK(x,y)$.\\
Therefore,   $f_{1,i}=\partial_x \mathcal{F}_i/\mathcal{F}_i$ is a solution of the following Riccati associated equation: 
$$\partial_x f_{1,i}=Rf_{1,i}+S-f_{1,i}^2.$$
Thus $\partial_x f_{1,i} \in \KK(x,y)[f_{1,i}]_{\leq 2}$.\\
 Furthermore, $f_{2,i}=\partial_y \mathcal{F}_i/\mathcal{F}_i$ is a solution of the Riccati equation: 
 $$(\star) \quad \partial_y f_{2,i}=F-f_{2,i}^2.$$
 Thus $\partial_y f_{2,i} \in \KK(x,y)[f_{2,i}]_{\leq 2}$.\\
Now we are going to prove that $\partial_ y f_{i,i} \in \KK(x,y)[f_{1,i}]_{\leq 2}$. We have
$$(\star \star) \quad Af_{1,i}+Bf_{2,i}= \frac{1}{2}A\partial_y \left(\frac{B}{A}\right)$$
and thus, by dividing by $A$ and derivating relatively to $y$ the equation $(\star \star)$, we get:
\begin{eqnarray*}
\partial_y f_{1,i}&= & -\partial_y \left(\dfrac{B}{A}\right) f_{2,i}-\dfrac{B}{A} \partial_y f_{2,i}+\partial_y\left( \dfrac{1}{2}\partial_y \left(\dfrac{B}{A}\right)\right).\\
\end{eqnarray*}
The equation $(\star)$ gives
\begin{eqnarray*}
 \partial_y f_{1,i}&= & -\partial_y \left(\frac{B}{A}\right) f_{2,i}-\frac{B}{A} (F-f_{2,i}^2)+\partial_y\left( \frac{1}{2}\partial_y \left(\frac{B}{A}\right)\right).\\
 \end{eqnarray*}

Now, thanks to $(\star \star)$ we can write $f_{2,i}$ in terms of $f_{1,i}$ and this gives

\begin{eqnarray*}
\partial_y f_{1,i}&= & -\partial_y \left(\frac{B}{A}\right) \left( \frac{1}{2}\frac{A}{B}\partial_y \left(\frac{B}{A}\right)-\frac{A}{B} f_{1,i} \right) -\frac{B}{A} \Big[F-\Big( \frac{1}{2}\frac{A}{B}\partial_y \left(\frac{B}{A}\right)-\frac{A}{B} f_{1,i} \Big)^2\Big]\\
 &&+\partial_y\left[ \frac{1}{2}\partial_y \left(\frac{B}{A}\right)\right]
 \in \KK(x,y)[f_{1,i}]_{\leq 2}.
\end{eqnarray*}
Symmetrically, we also obtain that $\partial_x f_{2,i} \in\mathbb{K}(x,y)[f_{2,i}]_{\leq 2}$. Then, we can construct a (four successive) Riccati extension $L_1$ of $\KK(x,y)$ containing
$$\partial_y \mathcal{F}_1/\mathcal{F}_1,\partial_y \mathcal{F}_2/\mathcal{F}_2,\partial_x \mathcal{F}_1/\mathcal{F}_1,\partial_x \mathcal{F}_2/\mathcal{F}_2.$$
Now taking a primitive extension and then an exponential extension for each $\mathcal{F}_i$, we obtain a field $L_2$ containing $\mathcal{F}_1,\mathcal{F}_2$, and thus the first integral $\mathcal{F}_1/\mathcal{F}_2$.
\end{proof}

%
%%%%%%%%%%%%%%%%%%%%%%%%%%%%%%%%%%%%%%%%%%%%%%%%%%%%%%%%%%%%%%%%%%%%%%%%%%%%%%%%

%%%%%%%%%%%%%%%%%%%%%%%%%%%%%%%%%%%%

\section{Extactic hypersurfaces}\label{sec:exta}

As already remarked in Proposition \ref{propinv}, the existence of a Darbouxian, or Liouvillian or Riccati first integral is equivalent to the existence of a rational first integral with a special structure for a derivation associated to the problem. Furthermore, in this situation we have new variables $y_1, y_2, y_3$. In the following we will need to study rational first integral for derivations with several variables $x,y,y_1, y_2, y_3$. Thus, in this section we study  the characterization of  rational first integrals for a derivation with variables $x,y,y_1,\ldots,y_n$. Our main tool will be the extactic curve. This curve has been discovered independently by Lagutinski and Pereira, see \cite{Pereira}. It allows one to characterize the situation where a derivation has a rational first integral with bounded degree. Here, we define and prove the main property of this object for a derivation in $\KK(x,y,y_1,\ldots,y_n)$ and we will get extactic hypersurfaces. A similar study in the bivariate case has already be done in \cite{ChezeHDR}.

We consider a derivation
$$D=f \partial_x+ f_0 \partial_y +\sum_{j=1}^n f_j \partial_{y_j}, \textrm{ where } f_j \in \KK[x,y,y_1,\ldots,y_n]$$
with $f \neq 0$ and we consider the associated differential system:
$$(S_n) \quad 
\begin{cases}
\partial_x y(x)=\dfrac{f_0\big(x,y_1(x),\ldots,y_n(x)\big)}{f\big(x,y_1(x),\ldots,y_n(x)\big)},\\
\partial_{x} y_j(x)=\dfrac{f_j\big(x,y_1(x),\ldots,y_n(x)\big)}{f\big(x,y_1(x),\ldots,y_n(x)\big)}, \textrm{ for } j=1,\ldots,n.
\end{cases}$$

We want to characterize the existence of a rational first integral with degree smaller than $N$ for this kind of differential system. The idea is to study the order of contact between a solution of $(S_n)$ and a polynomial.

\begin{defi}
We set $\underline{y}(x)=\big(y(x),y_1(x) \ldots, y_n(x)\big)$.\\
A parametrized curve  $\big(x,\underline{y}(x)\big)$ and an  hypersurface defined by the zero locus of $g(x,y,y_1,\ldots,y_n) \in\KK[x,y,y_1,\ldots,y_n]$ have  a contact of order $\nu$  at $(x_0,y_0,y_{1,0},\ldots,y_{n,0})$ $=\big(x_0,\underline{y}(x_0)\big)$, when $\nu$ is the biggest integer  such that:
$$g\big(x,\underline{y}(x)\big)=0 \mod (x-x_0)^{\nu}.$$
\end{defi}
 
%When we consider a planar vector field, the idea  to discover a rational first integral is to compute algebraic curves with a ``high" order of contact with a generic solutions of $(S'_0)$. 
In order to find a rational first integral for a plane vector field, it is sufficient to  compute an algebraic curve with a ``high enough" order of contact with a \emph{generic solution of }$(S'_0)$, see \cite{BCCW}.
Actually, here ``high enough" means infinite. Indeed, we will see that if the order of contact is large enough (a bound will be given later) then  the order of contact will be infinite. We are going to use the same approach  in the multivariate case.\\
 
In control theory this kind of problem is classical. Risler, in \cite{Risler}, and Gabrielov \cite{Gabrielov} have  shown that if the order of contact is big enough then it is infinite. More precisely,  in \cite{Gabrielov} the author shows that if $y(x)$ is a solution of a differential system with degree $d$ and $g$ is a polynomial with degree $k$ such that $g(x,\underline{y}(x))\neq 0$ then the order of contact between $g=0$ and $(x,\underline{y}(x))$ is smaller than $2^{2n+3}\sum_{j=1}^{n+2}[k+(j-1)(d-1)]^{2n+4}$. \\
Here, we will get a better bound because we are going to consider a solution $y(x)$ with a \emph{generic initial condition}.\\
 In the following $x_0,y_0, y_{1,0}, \ldots, y_{n,0}$ \emph{are new variables}. They correspond to a generic initial condition.\\

In order to compute the  order of contact  between $g$ and a solution 
$$\big(x,\underline{y}(x)\big)=\big(x,y(x),y_1(x), \ldots,y_n(x)\big)$$ 
we have to compute the Taylor expansion of $g(x,\underline{y}(x))$ at $x_0$. As  
$$f\partial_x\Big(g\big(x,\underline{y}(x)\big)\Big)=D(g)\big(x,\underline{y}(x)\big),$$ and $f \neq 0$ we deduce easily that for all $l \in \NN$
$$\partial_{x}^i g\big(x,\underline{y}(x)\big) =0,\textrm{ for } i=1, \ldots,l \iff D^i(g)(x_0,y_0,y_{1,0},\ldots,y_{n,0})=0, \textrm{ for } i=1,\ldots,l.$$
where $D^0(g)=g$ and $D^{i}(g)=D\big(D^{i-1}(g)\big)$.\\
 
The study of the order of contact at a generic point $(x_{0},y_0,y_{1,0},\ldots,y_{n,0})$ leads us to consider the following map:
 
\begin{defi}
Let $D$ be a derivation and let $V$ be a linear subspace of dimension $l$ of $\KK[x,\underline{y}]$, where $\underline{y}=(y,y_1, \ldots,y_n)$.\\ Let $x_0$ and $\underline{y}_0=(y_0,y_{1,0}, \ldots,y_{n,0})$ be new variables and set $\LL=\KK(x_0,\underline{y}_0)$.
%Let $D$ be a derivation on $\KK[x,y_1\ldots, y_n]$, $V$ be a finite dimensional linear subspace of $\KK[x,y_1\ldots, y_n]$  and $x_0$, $\underline{y}_0=(y_{1,0},\ldots,y_{n,0})$ be new variables. We set $\LL=\KK(x_0,\underline{y}_0)$.
 We consider the linear $\LL$-morphism:
\begin{eqnarray*}
\mathcal{E}_{D}^{V}: \LL\otimes_{\KK} V& \longrightarrow & \LL^l\\
g(x,\underline{y})& \longmapsto &\big(g(x_0,\underline{y}_0),D(g)(x_0,\underline{y}_0),D^2(g)(x_0,\underline{y}_0), \ldots, D^{l-1}(g)(x_0,\underline{y}_0)\big)
\end{eqnarray*}
where %$l=\dim_{\KK} V$, 
$D^{k}(g)=D\big(D^{k-1}(g)\big)$ and $D$ is, by abuse of notation, the extension of the derivation $D$ to $\LL[x,\underline{y}]$, i.e. $$D\Big(\sum_{\alpha} c_{\alpha}(x_0,\underline{y}_0)x^{\alpha_1}\underline{y}^{\alpha_2}\Big)=\sum_{\alpha}c_{\alpha}(x_0,\underline{y}_0)D(x^{\alpha_1}\underline{y}^{\alpha_2}).$$ 
The determinant of this linear map is denoted by $E_{D}^{V}(x_0,\underline{y}_0)$.\\ 
Moreover, we call the hypersurface $E_{D}^{\KK[x,\underline{y}]_{\leq N}}(x_0,\underline{y}_0)$  the $N$-th extactic hypersurface.\\
\end{defi}

\begin{rem}
With this definition we have:
$$g(x,\underline{y}) \in \ker \mathcal{E}_D^V \iff g\big(x, \underline{y}(x)\big)=0 \mod (x-x_0)^l.$$
\end{rem}
If $\{g_1,\ldots,g_l\}$ is a basis of  $V$ then the  associated extactic hypersurface is given by 
$$E^V_D(x_0,\underline{y}_0)=\begin{vmatrix}
g_1(x_0,\underline{y}_0)&g_2(x_0,\underline{y}_0) &\ldots & g_l(x_0,\underline{y}_0) \\
D(g_1)(x_0,\underline{y}_0) & D(g_2)(x_0,\underline{y}_0) & \ldots & D(g_l)(x_0,\underline{y}_0)\\
\vdots& \vdots &\vdots &\vdots\\
D^{l-1}(g_1)(x_0,\underline{y}_0) & D^{l-1}(g_2)(x_0,\underline{y}_0) & \ldots & D^{l-1}(g_l)(x_0,\underline{y}_0)
\end{vmatrix}.$$
The extactic hypersurface is related to invariant algebraic hypersurfaces. We call these kinds of hypersurfaces Darboux polynomials.

\begin{defi} \label{defdarboux}
A non constant polynomial $M \in \overline{\KK}[x,\underline{y}]$ is a {\em Darboux polynomial for $D$} if $M$ divides $D(M)$ in $\overline{\KK}[x,\underline{y}]$. We call the polynomial \mbox{$\Lambda=D(M)/M$} the {\em cofactor associated with the Darboux polynomial $M$}.
\end{defi}

\begin{prop}\label{PropDarbouxfacteurext}
If $g \in V$ is a Darboux polynomial of a derivation $D$ then $g(x_0,\underline{y}_0)$ is a factor of $E_{D}^{V}(x_0,\underline{y}_0)$.
\end{prop}
\begin{proof}
Let $\mathcal{B}=\{g,g_2,\ldots,g_l\}$ be a basis of $V$.\\
As $g$ is a Darboux polynomial we have $D(g)=\Lambda. g$ where $\Lambda \in \KK[x,\underline{y}]$. Thus there exist polynomials $\Lambda_j$ such that $D^j(g)=\Lambda_j.g$. Thus, we have
$$\mathcal{E}_{D}^{V}(g)=g(x_0,\underline{y}_0).\big(1,\Lambda(x_0,\underline{y}_0), \ldots, \Lambda_l(x_0,\underline{y}_0)\big)$$
and $g(x_0,\underline{y}_0)$ is a factor of the first column of a matrix representation of $\mathcal{E}_{D}^{V}$ in the basis $\mathcal{B}$. It follows that $g$ is factor of $E_{D}^{V}(x_0,\underline{y}_0)$.
\end{proof}

We remark that by definition the  factors of $E_{D}^V(x_0,\underline{y}_0)$ which are not Darboux polynomials correspond to algebraic hypersurfaces with order of contact with a solution $(x,y(x),y_1(x),\ldots,y_n(x))$  bigger than $l=\dim_{\KK} V$.\\

We also remark that the determinant $E^V_D(x_0,\underline{y}_0)$ corresponds to a Wronskian and we recall the following classical lemma, see \cite[Lemma 3.3.5]{bronstein}:

\begin{lem}\label{Wronskian}
Let $(F,D)$ be a differential field. Then $g_1,\ldots, g_l \in F$ are linearly dependent over $\ker D$ if and only if $W_D(g_1,\ldots,g_l)=0$, where
$$W_D(g_1,\ldots,g_l)=
\begin{vmatrix}
g_1&g_2 &\ldots & g_l \\
D(g_1) & D(g_2) & \ldots & D(g_l)\\
\vdots& \vdots &\vdots &\vdots\\
D^{l-1}(g_1) & D^{l-1}(g_2) & \ldots & D^{l-1}(g_l)
\end{vmatrix}
$$
is the Wronskian of $g_1,\ldots, g_l$ relatively to $D$.
\end{lem}

 This leads to the following proposition:

\begin{prop}\label{Propextwronskian}
%We have the following equivalence:\\
There exists a $\KK$ vector space $V$  such that  $E^V_D(x_0,\underline{y}_0) =0$  if and only if $D$ has a  rational first integral.
\end{prop}

\begin{proof}
We denote by $\{g_1, \ldots, g_l\}$ a basis of $V$. This basis gives a basis of the $\LL$ vector space $\LL \otimes_{\KK} V$, where $\LL=\KK(x_0,\underline{y}_0)$.\\
We consider the $\KK[x_0,\underline{y}_0]$ derivation
$$\tilde{D}=  f(x_0,\underline{y}_0) \partial_{x_0}+f_0(x_0,\underline{y}_0)\partial_{y_0}+ \sum_i f_i(x_0,\underline{y}_0) \partial_{y_{i,0}}.$$
We remark that there exists an isomorphism between $\ker D \subset \KK(x,\underline{y})$ and $\ker \tilde{D} \subset \KK(x_0, \underline{y}_0)$.\\

We have then $E^V_D(x_0,\underline{y}_0)=W_{\tilde{D}}(g_1(x_0,\underline{y}_0),\ldots, g_l(x_0,\underline{y}_0))$.\\
By Lemma \ref{Wronskian} applied with $F=\LL$  and $\tilde{D}$, we deduce:
$$E^V_D(x_0,\underline{y}_0) =0 \iff g_1(x_0,\underline{y}_0),\ldots,g_l(x_0,\underline{y}_0) \textrm{ are linearly dependent over } \ker \tilde{D}.$$
As $g_1(x_0,\underline{y}_0),\ldots,g_l(x_0,\underline{y}_0)$ are linearly independent over $\KK$, this means that \mbox{$\ker \tilde{D} \neq \KK$}. Thus $\tilde{D}$ and then $D$ has  a rational first integral. \\

Conversely, if $D$ has a rational first integral $G_1/G_2$ with degree $N$ then we set $V= \KK[x,\underline{y}]_{\leq N}$ and we can consider  a basis $\{g_1,g_2, \ldots,g_l\}$ of $V$ where $g_1=G_1$ and $g_2=G_2$. We have $g_1(x_0,\underline{y}_0)/g_2(x_0,\underline{y}_0) \in \ker \tilde{D}$ and 
$$g_1(x_0,\underline{y}_0) - \dfrac{g_1(x_0,\underline{y}_0)}{g_2(x_0,\underline{y}_0)}g_2(x_0,\underline{y}_0)=0.$$
Thus we have a non-trivial relation over $\ker \tilde{D}$ then by Lemma~\ref{Wronskian}
$$W_{\tilde{D}}(g_1(x_0,\underline{y}_0),\ldots, g_l(x_0,\underline{y}_0))=E^V_D(x_0,\underline{y}_0)=0.$$
\end{proof}

\begin{rem}\label{remdegiprexta}
The previous proof shows that if $D$ has a rational first integral with degree $N$ then $E^{\KK[x,y]_{\leq N}}_D(x_0,\underline{y}_0)=0$.
\end{rem}

This kind of result is not new, see \cite{Pereira}. We have given here a proof in order to emphasize the relation between the extactic curve and the Wronskian. The following example shows however that it is possible to have $E^{\KK[x,\underline{y}]_{\leq N}}_D(x_0,\underline{y}_0)$ equals to zero and no rational first integral with degree $N$.

\begin{xx}
Consider the following derivation $$D=x\partial_{x}+(3x-2y)\partial_{y}-3x^3\partial_{y_1}.$$
This derivation has two polynomial first integrals with degree 3: $$P_1(x,y,y_1)=x^2y+y_1,\quad P_2(x,y,y_1)=x^3+y_1.$$
For this derivation we have $E_D^{\KK[x,y,y_1]_{\leq 2}}(x_0,\underline{y}_0)=0$, but $D$ has no rational first integral with degree $2$. Indeed, a direct computation with a computer algebra system shows that the only Darboux polynomials for this derivation with degree smaller than 2 are: $x$, $y-x$ and their products.\\
Now, we explain why the second extactic curve is equal to zero. As $P_1$ and $P_2$ are first integrals we have: $\Delta P_i(x_0,\underline{y}_0;x,y(x))=0$, where 
$$\Delta P_i(x_0,\underline{y}_0;x,\underline{y})=P_i(x,y,y_1)-P_i(x_0,y_{0},y_{1,0}).$$ Thus 
\begin{eqnarray*}
P(x_0,\underline{y}_0;x,\underline{y})&=&x.\Delta P_1(x_0,\underline{y}_0;x,\underline{y})-y.\Delta P_2(x_0,\underline{y}_0;x,\underline{y})\\
&=&xy_1-yy_1-    x(x_{0}^2y_{0}+y_{1,0})+y(x_{0}^3+y_{1,0}) 
\end{eqnarray*}
 has degree 2 in $\KK(x_0,\underline{y}_0)[x,y,y_1]$ and satisfies 
 $$ P(x_0,\underline{y}_0;x,y(x),y_1(x))=0.$$
Thus $P \in \ker \mathcal{E}_{D}^{\KK[x,y,y_1]_{\leq 2}}$ and $E^{\KK[x,y,y_1]_{\leq 2}}_D(x_0,\underline{y}_0)=0$.
\end{xx}

Now, we show that the computation of $\ker \mathcal{E}^N_D$ gives rational first integrals. More precisely, we are going to exhibit a structure for the elements in  $\ker \mathcal{E}^V_D$.

\begin{prop}\label{rrefker}
Let $V$ be a linear subspace of $\KK[x,\underline{y}]$ of dimension $l$.\\
Let $\{b_i(x,\underline{y})\}$ be a basis of $V$, and let $\tilde{D}$ be the following $\KK[x_0,\underline{y}_0]$ derivation:
$\tilde{D}= f(x_0,\underline{y}_0) \partial_{x_0} +f_0(x_0,\underline{y}_0)\partial_{y_0}+ \sum_i f_i(x_0,\underline{y}_0) \partial_{y_{i,0}}$.\\
Consider $g_1(x_0,\underline{y}_0;x,\underline{y})$, \ldots, $g_l(x_0,\underline{y}_0;x,\underline{y})$ a basis of $\ker \mathcal{E}^V_D$ in reduced row echelon form. Then we can write each $g_i$ in the following form:
$$g_i(x_0,\underline{y}_0;x,\underline{y})=\sum_{j \in J_i} c_{j}(x_0,\underline{y}_0)b_j(x,y),$$
where  $J_i$ is a finite set and  $c_{j}(x_0,\underline{y}_0) \in \ker \tilde{D}$.\\
 Furthermore, for all $g_i$  there exists $j_i$ such that $c_{j_i}(x_0,\underline{y}_0) \not \in \KK$.
\end{prop}

As mentioned before, there exists a natural isomorphism between the two kernels, \mbox{$\ker D \subset \KK(x,\underline{y})$} and $\ker \tilde{D} \subset \KK(x_0,\underline{y}_0)$. Thus, this proposition says that the computation of a reduced row echelon basis of $\ker \mathcal{E}^V_D$ gives  a non-trivial rational first integral: $c_{j_i}(x,\underline{y})$.\\
For the definition of the reduced row echelon form we can consult \cite[Chapter~1]{Lay}.

\begin{proof}
Consider $\{g_1,\ldots,g_l\}$ a basis of $\ker \mathcal{E}^V_D$ in reduced row echelon form and we set $g_i(x_0,\underline{y}_0;x,\underline{y})=\sum_{j \in J_i}p_{j}(x_0,\underline{y}_0)b_j(x,y)$.

As $g_i \in \ker \mathcal{E}^V_D$, we have for all $k\leq l-1$, 
$$\sum_{j \in J_i} p_j(x_0,\underline{y}_0)D^k(b_j)(x_0,\underline{y}_0)=0.$$
We set $J_i=\{j_0,j_1\ldots,j_{k_{i}}\}$, then the previous equalities give:
$$\begin{pmatrix}
b_{j_0}(x_0,\underline{y}_0) & \ldots & b_{j_{k_i}}(x_0,\underline{y}_0)\\
D(b_{j_0})(x_0,\underline{y}_0) & \ldots & D(b_{j_{k_i}})(x_0,\underline{y}_0)\\
\vdots & \ldots &\vdots\\
D^{k_i}(b_{j_0})(x_0,\underline{y}_0) & \ldots & D^{k_i}(b_{j_{k_i}})(x_0,\underline{y}_0)\\
\end{pmatrix}
\cdot
\begin{pmatrix}
p_{j_0}(x_0,\underline{y}_0)\\
p_{j_1}(x_0,\underline{y}_0)\\
\vdots\\
p_{j_{k_i}}(x_0,\underline{y}_0)
\end{pmatrix}
=
\begin{pmatrix}
0\\
0\\
\vdots\\
0\\
\end{pmatrix}.
$$
As $(p_{j_0}(x_0,\underline{y}_0),\ldots, p_{j_{k_i}}(x_0,\underline{y}_0)) \neq 0$, this implies $W_{\tilde{D}}(b_j(x_0,\underline{y}_0);j \in J_i)=0$. Then, we deduce thanks to Lemma~\ref{Wronskian}, that the $b_j(x_0,\underline{y}_0)$ with $j \in J_i$   are linearly related over $\ker \tilde{D}$.\\
Then there exists $c_{j}(x_0,\underline{y}_0) \in \ker \tilde{D}$ such that $\sum_{j \in J_i}c_{j}(x_0,\underline{y}_0)b_j(x_0,\underline{y}_0)=0$. As $\{b_j(x_0,\underline{y}_0) \, |  \, j\in J_i\}$ is a family of  linearly independent elements over $\KK$, we can suppose without loss of generality that $c_{j_0}(x_0,\underline{y}_0)=1$ and that there exists an index $j_i$ such that $c_{j_i}(x_0,\underline{y}_0) \not \in\KK$.
Furthermore, we have:
\begin{eqnarray*}
0&=&\tilde{D}\big(\sum_{j \in J_i}c_{j}(x_0,\underline{y}_0)b_j(x_0,\underline{y}_0)\big)\\
&=&\sum_{j \in J_i}\tilde{D}(c_{j})(x_0,\underline{y}_0)b_j(x_0,\underline{y}_0)+\sum_{j \in J_i}c_{j}(x_0,\underline{y}_0)\tilde{D}\big(b_j(x_0,\underline{y}_0)\big)\\
\end{eqnarray*}
As $c_j(x_0,\underline{y}_0) \in \ker \tilde{D}$, this implies
\begin{eqnarray*}
0&=&\sum_{j \in J_i}c_{j}(x_0,\underline{y}_0)\tilde{D}\big(b_j(x_0,\underline{y}_0)\big).
\end{eqnarray*}
In the same way, we get $\sum_{j \in J_i}c_{j}(x_0,\underline{y}_0)\tilde{D}^j\big( b_j(x_0,\underline{y}_0)\big)=0.$ \\
It follows that $\sum_{j \in J_i}c_{j}(x_0,\underline{y}_0)b_j(x,y) \in \ker \mathcal{E}_D^V$. This polynomial has the same support as the polynomial $g_i$ and the basis $\{g_1,\ldots,g_l\}$ is in reduced row echelon form, thus we get the desired result.
\end{proof}

Now, we are going to give an explicit statement for ``if the contact between an hypersurface and an orbit is big enough then the orbit is included in the hypersurface."

\begin{thm}\label{prec_gen}
Let $\underline{y}(x)=\big(y(x),y_1(x),\ldots,y_n(x)\big)$ be a solution of  $(S_n)$ satisfying the initial condition  $\underline{y}(x_0)=\underline{y}_0$ and let $V$ be a linear subspace of $\KK[x,\underline{y}]$ of dimension $l$.\\
If $P \in \LL \otimes_{\KK} V$ and $P(x_0,\underline{y}_0;x,\underline{y}(x))=0 \mod(x-x_0)^{l}$ then \mbox{$P(x_0,\underline{y}_0;x,\underline{y}(x))=0$}.
\end{thm}

\begin{proof}
Consider $P(x_0,\underline{y}_0;x,\underline{y})$ such that $P(x_0,\underline{y}_0;x,\underline{y}(x))=0 \mod (x-x_0)^{l}$. The Taylor expansion of $P$ shows that $P(x_0,\underline{y}_0;x,\underline{y}) \in \ker \mathcal{E}^V_D$.\\
We can write $P(x_0,\underline{y}_0;x,\underline{y})$ in the following form:
$$P(x_0,\underline{y}_0;x,\underline{y})=\sum_i \lambda_i(x_0,\underline{y}_0)g_i(x_0,\underline{y}_0;x,\underline{y}),$$
where the polynomials $g_i(x_0,\underline{y}_0;x,\underline{y})$ satisfy Proposition \ref{rrefker}.

We have:
\begin{eqnarray*}
g_i(x,\underline{y}(x);x,\underline{y})&=&\sum_{j \in J_i} c_{j}(x,\underline{y}(x))b_j(x,\underline{y})\\
&=&\sum_{j \in J_i} c_{j}(x_0,\underline{y}_0)b_j(x, \underline{y}), \textrm{ because }  c_{j}(x_0,\underline{y}_0) \in \ker \tilde{D},\\
&=&g_i(x_0,\underline{y}_0;x,\underline{y}).
\end{eqnarray*}
Furthermore, $g_i(x_0,\underline{y}_0;x_0,\underline{y}_0)=0$ because $g_i(x_0,\underline{y}_0;x,\underline{y}) \in \ker \mathcal{E}^V_D$.\\ 
Thus $g_i(x,\underline{y}(x);x,\underline{y}(x))=0$.\\
 As $g_i(x,\underline{y}(x);x,\underline{y})=g_i(x_0,\underline{y}_0;x,\underline{y})$ we get
$$0=g_i(x,\underline{y}(x);x,\underline{y}(x))=g_i(x_0,\underline{y}_0;x,\underline{y}(x)).$$
Then $P(x_0,\underline{y}_0;x,\underline{y}(x))=\sum_i \lambda_i(x_0,\underline{y}_0)g_i(x_0,\underline{y}_0;x,\underline{y}(x))=0$.
\end{proof}
This result  means that for a generic point if the order of contact with a polynomial of degree $N$ is bigger than $\dim_{\KK}\KK[x,\underline{y}]_{\leq N}$ then this order of contact is infinite.
%%%%%%%%%%%%%%%%%%%%%%%%%%%%%%%%%%%%%%%%%%%%%%%%%%%%%%%%%%%%%%%%%%%%%%

\section{Extactic curves}\label{sec:extacurve}

\subsection{Rational extactic curve}%$\,$\\
In this subsection, we recall a classical result for the extactic curve in two variables. Let us note
$$\tilde{\mathcal{E}}_{D_0}^N(x_0,y_0)=\mathcal{E}_{D_0}^{\KK[x,y]_{\leq N}},\quad \tilde{E}_{D_0}^N(x_0,y_0)=E_{D_0}^{\KK[x,y]_{\leq N}},$$
and call a rational function $F\in\bar{\mathbb{K}}(x,y)$ \emph{indecomposable} when it cannot be written $f \circ g$, with $f \in \overline{\KK}(T)$, $g \in \overline{\KK}(x,y)$ and $\deg(f)\geq 2$.

\begin{thm}[Bivariate rational extactic curve theorem]\label{thm:exta+rfi}$\,$\\
The derivation $D_0$ has an indecomposable rational first integral with degree $N$ if and only if $\tilde{E}_{D_0}^N(x_0,y_0)=0$ and $\tilde{E}_{D_0}^{N-1}(x_0,y_0) \neq 0$. Moreover, this indecomposable first integral can always be assumed to have coefficients in $\mathbb{K}$.
\end{thm}

This theorem says that the minimal degree of a rational first integral corresponds to the minimal index where the extactic curve vanishes. This theorem is  not new, see e.g. \cite{Pereira,ChrLibPer}. We have recalled it because we are going to generalize this result for the study of Darbouxian, Liouvillian and Riccati first integrals.

\subsection{Darbouxian extactic curve}%$\,$\\
In this subsection we are going to apply the result of Section \ref{sec:exta} to the derivation $D_1$. Then in the following $\big(y(x),y_1(x)\big)$ is a solution of $(S'_1)$ satisfying the initial condition $y(x_0)=y_0,y_1(x_0)=y_{1,0}$, where $x_0$, $y_0$ and $y_{1,0}$ are variables.\\

Now, we are going to generalize Theorem \ref{thm:exta+rfi} to the Darbouxian case.

\begin{defi}\label{def:exta_Darb}
We set 
$$V_1:=\KK[x,y]_{\leq N}\oplus \KK[x,y]_{\leq N} y_1^k,\quad l_1=\hbox{dim}(V_1).$$
We denote by $\tilde{\mathcal{E}}^{N,k}_{D_1}(x_0,y_0)$ the linear map $\mathcal{E}^{V_1}_{D_1}$ after the specialization $y_{1,0}=1$.\\
The $N$-th $k$-Darbouxian extactic curve is defined by
$$ \tilde{E}^{N,k}_{D_1}(x_0,y_0)= \det\big( \tilde{\mathcal{E}}^{N,k}_{D_1}(x_0,y_0)\big).$$
\end{defi}

Let us begin with two lemmas about this Darbouxian extactic curve.

\begin{lem}\label{lem:eval_exta_darboux}
We have the following equivalence:
$$y_1^kP+Q \in \ker \tilde{\mathcal{E}}^{N,k}_{D_1}(x_0,y_0)$$
$$\Updownarrow$$
$$y_1^kP+y_{1,0}^k Q \in \ker \mathcal{E}^{V_1}_{D_1}.$$
\end{lem}

This lemma says that evaluating $y_{1,0}=1$ in the definition of the extactic curve does not lose much information.

\begin{proof}
We denote by $\big(\psi(x),\psi_1(x)\big)$ the solution of $(S'_1)$ such that $\psi(x_0)=y_0$, $\psi_1(x_0)=1$.\\
We consider the transformation 
$$T(y,y_1)=(y,\, y_{1,0}y_1).$$
We set 
$$\big(\psi_T(x),\psi_{T,1}(x)\big):=T(\psi(x),\psi_1(x)\big)=\big( \psi(x),\, y_{1,0}\psi_1(x)\big).$$
Now, we are going to show that $\big(\psi_T(x),\psi_{T,1}(x)\big)$ is a solution of $(S'_1)$ with initial conditions $\psi_T(x_0)=y_0$, $\psi_{T,1}(x_0)=y_{1,0}$. 
Indeed, 
$$\partial_x\psi_{T,1}(x)=y_{1,0}\partial_x\psi_1(x)=y_{1,0} \psi_1(x)\partial_y\Big( \dfrac{B}{A} \Big)\big(x,\psi(x) \big)= \psi_{T,1}(x)\partial_y\Big( \dfrac{B}{A} \Big)\big(x,\psi_T(x) \big).$$
This implies the equality  $\big(y(x),y_1(x)\big)=\big( \psi_T(x),\psi_{T,1}(x)\big)$.\\

Now suppose that $y_1^kP+Q \in \ker \tilde{\mathcal{E}}^{N,k}_{D_1}(x_0,y_0)$ then 
$$\psi_1^k(x)P\big(x, \psi(x)\big)+Q\big(x, \psi(x)\big)= 0  \mod (x-x_0)^{l_1}.$$
Thus this equality multiplied by $y_{1,0}^k$ gives
$$y_{1,0}^k\psi_1^k(x)P\big(x, \psi(x)\big)+y_{1,0}^kQ\big(x, \psi(x)\big)= 0  \mod (x-x_0)^{l_1}.$$
It follows
$$\psi_{T,1}^k(x)P\big(x, \psi_T(x)\big)+y_{1,0}^kQ\big(x, \psi_T(x)\big)= 0  \mod (x-x_0)^{l_1}.$$
Therefore, as $\big(y(x),y_1(x)\big)=\big( \psi_T(x),\psi_{T,1}(x)\big)$ we deduce that   $y_1^kP+y_{1,0}^k Q$ belongs to $\ker \mathcal{E}^{V_1}_{D_1}$.\\

The converse is straightforward.
\end{proof}

\begin{lem}\label{lem:ker_exta_darboux}
Consider  a non trivial element $y_1^kP+Q \in \ker \tilde{\mathcal{E}}^{N,k}_{D_1}(x_0,y_0)$, then:
\begin{itemize}
\item If $P=0$ then $Q \in \ker \tilde{\mathcal{E}}^N_{D_0}(x_0,y_0)$.
\item If $Q=0$ then $P \in \ker \tilde{\mathcal{E}}^N_{D_0}(x_0,y_0)$.
\item If $PQ \neq 0$ and $Q \not \in\ker \tilde{\mathcal{E}}^N_{D_0}(x_0,y_0)$ then:
$$\left(   D_0\left( \left(\dfrac{P}{Q}\right)^{1/k}\right) + A \left(\dfrac{P}{Q}\right)^{1/k}\partial_y\Big(\dfrac{B}{A}\Big)\right)\big(x,y(x)\big)=0.$$
\end{itemize}
\end{lem}

The first two  cases are pathological ones, i.e. we compute the Darbouxian extactic curve but it appears that a rational first integral exists.

\begin{proof}
As $y_1^kP+Q \in \ker \tilde{\mathcal{E}}^{N,k}_{D_1}(x_0,y_0)$ we get using Lemma \ref{lem:eval_exta_darboux}
$$ y_1^kP(x,y)+y_{1,0}^k Q(x,y) \in \ker \mathcal{E}^{V_1}_{D_1}.$$
%and with the initial conditions $y(x_0)=y_0,y_1(x_0)=y_{1,0}$

Thus
$$y_1(x)^kP(x,y(x)) + y_{1,0}^k Q(x,y(x)) =0\mod (x-x_0)^{l_1}.$$
By Theorem \ref{prec_gen}, we deduce that $y_1(x)^kP(x,y(x))+y_{1,0}^k Q(x,y(x)) =0$.\\
If $P=0$ then we get $Q(x,y(x))=0 \mod (x-x_0)^{l_1}$, then $Q \in \ker \tilde{\mathcal{E}}^N_{D_0}(x_0,y_0)$.\\ 
If $Q=0$ then we get $y_1(x)^kP(x,y(x))=0 \mod (x-x_0)^{l_1}$. We have $y_1(x) \neq 0$ as $y_1(x_0)=y_{1,0}$ and thus $P \in \ker \tilde{\mathcal{E}}^N_{D_0}(x_0,y_0)$.

Now we suppose that $PQ \neq 0$, and $Q \not \in \ker \tilde{\mathcal{E}}^N_{D_0}(x_0,y_0)$, then  \mbox{$Q\big(x,y(x)\big) \neq 0$} and:
$$y_1(x)^k\dfrac{P}{Q}\big(x,y(x)\big)=-y_{1,0}^k,$$
and thus
$$y_1(x)\left(\dfrac{P}{Q}\right)^{1/k}\big(x,y(x)\big)=\xi y_{1,0}, \;\;\; \textrm{ where }\;\xi\hbox{ is a } k\hbox{-th root of }-1.$$
The derivation relative to $x$ of this relation and the fact that $y_1(x)$ is a solution of $(S'_1)$ gives:
$$y_1(x)\partial_y\Big(\dfrac{B}{A}\Big)\big(x,y(x)\big)\left(\dfrac{P}{Q}\right)^{1/k}\big(x,y(x)\big)$$
$$+y_1(x)\left( \partial_x\Big(\Big(\dfrac{P}{Q}\Big)^{1/k}\Big) \big(x,y(x)\big) +\partial_y\Big(\Big(\dfrac{P}{Q}\Big)^{1/k}\Big)\big(x,y(x)\big)\dfrac{B}{A}\big(x,y(x)\big)\right)=0.$$
As $y_1(x) \neq 0$ we get:
$$\left( A \left(\dfrac{P}{Q}\right)^{1/k} \partial_y\Big( \dfrac{B}{A} \Big)+D_0\Big(\left(\dfrac{P}{Q}\Big)^{1/k}\right)\right)\big(x,y(x)\big)=0.$$
This gives the desired result.
\end{proof}

We now prove the main result about the Darbouxian extactic curve.\\

\begin{thm}[Darbouxian extactic curve theorem]\label{thm:exta+dfi}$\,$\\
\begin{enumerate}
\item If $\tilde{E}^{N,k}_{D_1}(x_0,y_0)=0$ then the derivation $D_0$ has a $k$-Darbouxian first integral with degree smaller than $N$ or a rational first integral with degree smaller than $2N+2d-1$. Moreover the defining equation of the first integral, (Rat) or (D), has coefficients in $\KK$.\\
\item If $D_0$ has a rational or a $k$-Darbouxian first integral with degree smaller than $N$ then $\tilde{E}^{N,k}_{D_1}(x_0,y_0)=0$.
\end{enumerate}
\end{thm}

Suppose that $D_0$ has no rational first integral and that $N$ is the smallest integer such that $\tilde{E}^{N,k}_{D_1}(x_0,y_0)=0$. Then, by the previous theorem, there exists a Darbouxian first integral with degree smaller than $N$. Moreover, this degree is minimal. Indeed, if there exists a Darbouxian first integral with degree  $M<N$ then, by the second part the theorem,  we have  $\tilde{E}^{M,k}_{D_1}(x_0,y_0)=0$ and this contradicts the minimality of $N$.
Therefore, as in the rational case, the smallest integer $N$ satisfying $\tilde{E}^{N,k}_{D_1}(x_0,y_0)=0$ is related to a first integral with minimal degree.\\

\begin{proof}%[Proof of Theorem \ref{thm:exta+dfi}]
First, consider a non trivial solution $y_1^kP+Q$ in $\ker \tilde{\mathcal{E}}^{N,k}_{D_1}(x_0,y_0)$.\\
If $P=0$ then by Lemma \ref{lem:ker_exta_darboux} we get $Q \in \ker \tilde{\mathcal{E}}^N_{D_0}(x_0,y_0)$.  Theorem \ref{thm:exta+rfi}  implies  that the derivation $D_0$ has a rational first integral with degree smaller than $N$ with coefficients in $\mathbb{K}$.\\
If $Q=0$, we deduce in the same way that $D_0$ has a rational first integral with degree smaller than $N$ with coefficients in $\mathbb{K}$.

Now we suppose that $PQ \neq 0$.\\
If  $Q \in\ker\tilde{\mathcal{E}}^N_{D_0}(x_0,y_0)$, then by Theorem \ref{thm:exta+rfi}, $D_0$ has a rational first integral with degree smaller than $N$ with coefficients in $\mathbb{K}$.\\
Now, we suppose that $Q$ does not belong to $\ker \mathcal{E}^N_{D_0}(x_0,y_0)$. By Lemma \ref{lem:ker_exta_darboux}, we have
$$\Big(D_0((P/Q)^{1/k}) + A (P/Q)^{1/k}\partial_y(B/A)\Big)\big(x,y(x)\big)=0.$$
We introduce the polynomial
$$G:=APQ(P/Q)^{-1/k}\Big(D_0((P/Q)^{1/k}) + A (P/Q)^{1/k}\partial_y(B/A)\Big)$$
we have $G\big(x,y(x)\big)=0$.\\
If $G=0$  then Proposition \ref{proprepresent} with $F=(P/Q)^{1/k}$ gives the existence of a $k$-Darbouxian first integral. Moreover, as $P, Q \in\KK(x_0,y_0)[x,y]$, the equation of type (D) giving the existence of a Darbouxian first integral has coefficients in $\mathbb{K}$.\\
If $G\neq 0$  then $G$ is a non-zero polynomial with degree smaller than $2N+2d-1$ such that $G \in \ker \tilde{\mathcal{E}}_{D_0}^{2N+2d-1}(x_0,y_0)$. By Theorem~\ref{thm:exta+rfi}, $D_0$ has a rational first integral with degree smaller than $2N+2d-1$ with coefficients in $\mathbb{K}$. This concludes the first part of the proof.\\

Now, we suppose that $D_0$ has a rational or a $k$-Darbouxian first integral with degree smaller than $N$.\\
If $D_0$ has a $k$-Darbouxian first integral $\mathcal{F}$ with degree smaller than $N$ then we have $\partial_y \mathcal{F}=(P/Q)^{1/k}$, with $P, Q \in \bar{\KK}[x,y]$, and $\deg P, \deg Q \leq N$. Proposition~\ref{propinv} implies that $y_1^kP/Q$ is a rational first integral of $(S'_1)$. Thus, when we consider a solution of $(S'_1)$ with initial condition $y(x_0)=y_0$, $y_1(x_0)=1$ we get
$$y_1(x)^k\dfrac{P\big(x,y(x)\big)}{Q\big(x,y(x)\big)}=c,$$
where $c \in \bar{\mathbb{K}}(x_0,y_0)$. Thus $y_1^kP(x,y)-cQ(x,y)$ belongs to $\bar{\KK} \otimes_{\KK} \ker \tilde{\mathcal{E}}^{N,k}_{D_1}(x_0,y_0)$. Now as $A,B \in \KK[x,y]$, the coefficients of $\tilde{\mathcal{E}}^{N,k}_{D_1}(x_0,y_0)$ are in $\mathbb{K}(x_0,y_0)$, and so we deduce the existence of $\tilde{P},\tilde{Q} \in \KK(x_0,y_0)[x,y]$ with degree $N$ such that $y_1^k\tilde{P}-\tilde{Q}$ belong to $\ker \tilde{\mathcal{E}}^{N,k}_{D_1}(x_0,y_0)$. Thus $\tilde{E}^{N,k}_{D_1}(x_0,y_0)=0$ and we get the desired conclusion.\\
If $D_0$ has a rational first integral $P/Q$ with degree smaller than $N$, then we can suppose $P/Q \in \KK(x,y)$ thanks to Theorem~\ref{thm:exta+rfi}. Furthermore,
$$P(x,y)Q(x_0,y_0)-Q(x,y)P(x_0,y_0) \in  \ker \tilde{\mathcal{E}}^{N,k}_{D_1}(x_0,y_0).$$
Thus $\tilde{E}^{N,k}_{D_1}(x_0,y_0)=0$.
\end{proof}

As a corollary, we obtain the following result about the coefficient field extensions of the possible first integrals.

\begin{cor}\label{cor:darbouxdansK}
Suppose that $D_0$ has a $k$-Darbouxian first integral and no rational first integral. We have:\\ 
There exists $F \in \KK(x,y)$ with $\deg F \leq N$ such that equation (D) gives a $k$-Darbouxian first integral if and only if 
there exists $\tilde{F} \in \overline{\KK}(x,y)$ with $\deg \tilde{F}\leq N$ such that equation (D) with $\tilde{F}$ gives a $k$-Darbouxian first integral.
\end{cor}

\begin{proof}
Suppose that there exists $\tilde{F} \in \overline{\KK}(x,y)$ with $\deg \tilde{F}=N$ such that equation (D) gives a $k$-Darbouxian first integral. Then by the second part of Theorem \ref{thm:exta+dfi}, we have $\tilde{E}^{N,k}_{D_1}(x_0,y_0)=0$. Applying now the first part, we obtain either a rational first integral (forbidden by our assumption) or a $k$-Darbouxian first integral given by an equation (D) with a rational function with degree smaller than $N$ and with coefficients in $\mathbb{K}$.
\end{proof}

%%%%%%%%%%%%%%%%%%%%%%%%%%%%%%%%%%%%%%%%%%%%%%%%%%%%%%
\subsection{Liouvillian Extactic curve}%$\,$\\
In this subsection we are going to apply the result of Section \ref{sec:exta} to the derivation $D_2$. Then in the following $\big(y(x),y_1(x), y_2(x)\big)$ is a solution of $(S'_2)$ satisfying the initial condition $y(x_0)=y_0$, $y_1(x_0)=y_{1,0}$,  $y_2(x_0)=y_{2,0}$ where $x_0$, $y_0$, $y_{1,0}$ and $y_{2,0}$ are variables.\\

We are going to generalize Theorem \ref{thm:exta+rfi} to the Liouvillian case. We need thus to define the extactic curve in the Liouvillian case.

\begin{defi}\label{def_exta_Liouv}
We set 
$$V_2:= \KK[x,y]_{\leq N} y_1^2 \oplus \KK[x,y]_{\leq N} y_2\oplus  \KK[x,y]_{\leq N}y_1,\quad l_2=\hbox{dim}(V_2).$$
We denote by $\tilde{\mathcal{E}}^{N}_{D_2}(x_0,y_0)$ the linear map $\mathcal{E}^{V_2}_{D_2}$ after the specialization $y_{1,0}=1$, $y_{2,0}=0$.\\
The N-th Liouvillian extactic curve is 
$$\tilde{E}^N_{D_2}(x_0,y_0)=\det \big( \tilde{\mathcal{E}}^{N}_{D_2}(x_0,y_0)\big).$$
\end{defi}

As in the Darbouxian case, we begin with two lemmas in order to show that the specialization $y_{1,0}=1$, $y_{2,0}=0$ does not lose information.

\begin{lem}\label{lem:eval_exta_liouv}
We have the following equivalence:
$$P(x,y)y_1^2+Q(x,y)y_2+R(x,y)y_1 \in \ker \tilde{\mathcal{E}}^{N}_{D_2}(x_0,y_0)$$
 $$\Updownarrow$$ 
$$P(x,y)y_1^2+Q(x,y)y_2+\Big(y_{1,0}R(x,y)-\dfrac{y_{2,0}}{y_{1,0}}Q(x,y)\Big)y_1 \in \ker \mathcal{E}^{V_2}_{D_2}.$$
\end{lem}

\begin{proof}
We are going to use the same strategy as the one used to prove Lemma~\ref{lem:eval_exta_darboux}. Here the transformation $T$ is $$T(y,y_1,y_2)=(y,\,\, y_{1,0}y_1,\,\, y_{1,0}^2y_2+y_{2,0}y_1).$$

We denote by $\big(\psi(x),\psi_1(x),\psi_2(x)\big)$ the solution of $(S'_2)$ such that \mbox{$\psi(x_0)=y_0$}, $\psi_1(x_0)=1$, $\psi_2(x_0)=0$.\\
We set 
\begin{eqnarray*}
\big(\psi_T(x),\psi_{T,1}(x),\psi_{T,2}(x)\big)&:=&T(\psi(x),\psi_1(x),\psi_2(x)\big)\\
&=&\big( \psi(x),\, \, y_{1,0}\psi_1(x),\, \, y_{1,0}^2\psi_2(x)+y_{2,0}\psi_1(x)\big).
\end{eqnarray*}
A direct computation shows that $\big(\psi_T(x),\psi_{T,1}(x),\psi_{T,2}(x)\big)$ is a solution of $(S'_2)$ with initial conditions $\psi_T(x_0)=y_0$, $\psi_{T,1}(x_0)=y_{1,0}$, $\psi_{T,2}(x_0)=y_{2,0}$. Thus we have $\big( y(x),y_1(x),y_2(x)\big)=\big(\psi_T(x),\psi_{T,1}(x),\psi_{T,2}(x)\big)$.\\
%Now, we are going to show that $\big(\psi_T(x),\psi_{1}_T(x),\bar{\psi_{1}}_T(x)\big)$ is a solution of $(S'_2)$ with initial conditions $\psi_T(x_0)=y_0$, $\psi_{1}_T(x_0)=y_{1,0}$, $\bar{\psi_{1}}_T(x_0)=y_{2,0}$. Indeed we have:
%$$\partial_x \psi_{1}_T(x)= y_{1,0} \partial_x \psi_{1}(x)= y_{1,0}\psi_{1}(x) \partial_y\Big( \dfrac{B}{A}\Big) \big( x,\psi(x)\big)=\psi_{1}_T(x)\partial_y\Big( \dfrac{B}{A}\Big) \big( x,\psi_T(x)\big).$$
%Furthermore,
%\begin{eqnarray*}
%\partial_x \bar{\psi_{1}}_T(x)&=& y_{1,0}^2 \partial_x \bar{\psi_{1}}(x)+y_{2,0}\partial_x \psi_{1}(x)\\
%&=& y_{1,0}^2 \bar{\psi_{1}}(x) \partial_y\Big( \dfrac{B}{A}\Big) \big( x,\psi(x)\big)+y_{1,0}^2\big( \psi_{1}(x)\big)^2\partial_y^2\Big( \dfrac{B}{A} \Big)\big( x,\psi(x)\big)+ y_{2,0}\psi_{1}(x)\partial_y\Big( \dfrac{B}{A}\Big) \big( x,\psi(x)\big)\\
%&=&\big(y_{1,0}^2 \bar{\psi_{1}}(x)+y_{2,0}\psi_{1}(x)\big)\partial_y\Big( \dfrac{B}{A}\Big) \big( x,\psi(x)\big)+\big(y_0\psi_{1}(x)\big)^2\partial_y^2\Big( \dfrac{B}{A} \Big)\big( x,\psi(x)\big)\\
%&=& \bar{\psi_{1}}_T(x)\partial_y\Big( \dfrac{B}{A}\Big) \big( x,\psi_T(x)\big)+\big(\psi_{1}_T(x)\big)^2\partial_y^2\Big( \dfrac{B}{A} \Big)\big( x,\psi_T(x)\big)
%\end{eqnarray*}

Now suppose that $P(x,y)y_1^2+Q(x,y)y_2+R(x,y)y_1 \in \ker \tilde{\mathcal{E}}^N_{D_2}(x_0,y_0)$  then
$$P(x,\psi(x))\psi_{1}^2(x)+Q(x,\psi(x))\psi_{2}(x)+R(x,\psi(x))\psi_{1}(x)= 0 \mod (x-x_0)^{l_2}.$$
Thus, the previous equality multiplied by $y_{1,0}^2$ gives 
\begin{eqnarray*}
y_{1,0}^2P(x,\psi(x))\psi_1^2(x)+y_{1,0}^2Q(x,\psi(x))\psi_2(x)+ y_{2,0} Q(x,\psi(x))\psi_1(x)&&\\ -    y_{2,0} Q(x,\psi(x))\psi_1(x) +y_{1,0}^2R(x,\psi(x))\psi_1(x) \quad = \quad 0 \mod (x-x_0)^{l_2}.&&
\end{eqnarray*}
This gives, thanks to the definition of $\psi_T$, $\psi_{T,1}$ and $\psi_{T,2}$,
\begin{eqnarray*}
P(x,\psi_T(x))\psi_{T,1}^2(x)+Q(x,\psi_T(x))\psi_{T,2}(x)\\
+\Big( -    \dfrac{y_{2,0}}{y_{1,0}} Q(x,\psi_T(x)) +y_{1,0}R(x,\psi_T(x)\Big)\psi_{T,1}(x) &=& 0 \mod (x-x_0)^{l_2}.
\end{eqnarray*}
Therefore, as $\big( y(x),y_1(x),y_2(x)\big)=\big(\psi_T(x),\psi_{T,1}(x),\psi_{T,2}(x)\big)$, we get
$$P(x,y)y_1^2+Q(x,y)y_2+\Big(y_{1,0}R(x,y)-\dfrac{y_{2,0}}{y_{1,0}}Q(x,y)\Big)y_1 \in \ker \mathcal{E}^{V_2}_{D_2}.$$
The converse is straightforward.
\end{proof}

\begin{lem}\label{lem:ker_exta_liouv}
Consider a non trivial element $P(x,y)y_1^2+Q(x,y)y_2+R(x,y)y_1$  in $\ker \tilde{\mathcal{E}}^{N}_{D_2}(x_0,y_0)$, then:
\begin{itemize}
\item If $Q=0$ then $Py_1+R \in \ker \tilde{\mathcal{E}}^{N}_{D_1}(x_0,y_0)$.
\item If $Q\neq 0$ and $Q \not \in \ker \tilde{\mathcal{E}}^N_{D_0}(x_0,y_0)$ then:
\begin{eqnarray*}
0&=&y_1(x) \Big( D_0(P/Q)+A(P/Q)\partial_y(B/A)+A\partial_y^2(B/A)  \Big) \big(x,y(x)\big)\\
&& + y_{1,0} D_0(R/Q)\big(x,y(x)\big)
\end{eqnarray*}
\end{itemize}
\end{lem}

\begin{proof}
We have $P(x,y)y_1^2+Q(x,y)y_2+R(x,y)y_1 \in \ker \tilde{\mathcal{E}}^{N}_{D_2}(x_0,y_0)$, then Lemma~\ref{lem:eval_exta_liouv} and Theorem~\ref{prec_gen} gives, as in the proof of Lemma~\ref{lem:ker_exta_darboux},
$$P(x,y(x))y_1(x)^2+Q(x,y(x))y_2(x)+\Big(y_{1,0}R(x,y(x))-\dfrac{y_{2,0}}{y_{1,0}}Q(x,y(x))\Big)y_1(x) =0 .$$

%As $P(x,y)y_1^2+Q(x,y)y_2+R(x,y)y_1 \in \ker \tilde{\mathcal{E}}^{N}_{D_2}(x_0,y_0)$, we get using Lemma~\ref{lem:eval_exta_liouv}
%$$P(x,y)y_1^2+Q(x,y)y_2+\Big(y_{1,0}R(x,y)-\dfrac{y_{2,0}}{y_{1,0}}Q(x,y)\Big)y_1 \in \ker \mathcal{E}^{V_2}_{D_2}(x_0,y_0,y_{1,0},y_{2,0})$$
% and thus
%$$P(x,y(x))y_1(x)^2+Q(x,y(x))y_2(x)+\Big(y_{1,0}R(x,y(x))-\dfrac{y_{2,0}}{y_{1,0}}Q(x,y(x))\Big)y_1(x) =0 \mod (x-x_0)^{l_2} $$
%Then Theorem \ref{prec_gen} applied with the four variables $x,y,y_1,y_2$ gives
%$$P(x,y(x))y_1(x)^2+Q(x,y(x))y_2(x)+\Big(y_{1,0}R(x,y(x))-\dfrac{y_{2,0}}{y_{1,0}}Q(x,y(x))\Big)y_1(x) =0 .$$
If $Q=0$, then we have 
$$P(x,y(x))y_1(x)+y_{1,0} R(x,y(x)) =0.$$
Thus $Py_1+y_{1,0}R \in \ker \mathcal{E}_{D_1}^{V_1}$ and by Lemma~\ref{lem:eval_exta_darboux} we get that $Py_1+R$ belongs to $\ker \tilde{\mathcal{E}}^{N}_{D_1}(x_0,y_0)$.\\
 
If $Q\neq 0$  and $Q \not \in \ker \tilde{\mathcal{E}}_{D_0}^N(x_0,y_0)$ then $Q\big(x,y(x)\big)\neq 0$. We set $F=P/Q$ and $G=R/Q$ then we have:
$$F\big(x,y(x)\big)y_1(x)+\dfrac{y_2(x)}{y_1(x)}-\dfrac{y_{2,0}}{y_{1,0}}+y_{1,0} G\big(x,y(x)\big)=0.$$
The derivation relatively to $x$ of the previous expression and the relation given by the differential system $(S'_2)$ gives:

\begin{eqnarray*}
0&=&\partial_x F\big(x,y(x)\big)y_1(x)+ \partial_y F \big(x,y(x)\big)\dfrac{B}{A}\big(x,y(x)\big)y_1(x)\\
&& + F\big(x,y(x)\big)y_1(x) \partial_y\Big(\dfrac{B}{A}\Big)\big(x,y(x)\big) 
+y_1(x)\partial_y^2\Big(\dfrac{B}{A}\Big)\big(x,y(x)\big)\\
&&+y_{1,0}\partial_xG\big(x,y(x)\big)+y_{1,0}\partial_yG\big(x,y(x)\big)\dfrac{B}{A}\big(x,y(x)\big).\\
\end{eqnarray*}

%\begin{eqnarray*}
%0&=& \partial_x F\big(x,y(x)\big)y_1(x)+ \partial_y F \big(x,y(x)\big)\dfrac{B}{A}\big(x,y(x)\big)y_1(x) + F\big(x,y(x)\big)y_1(x) \partial_y\Big(\dfrac{B}{A}\Big)\big(x,y(x)\big)\\
%&&+\dfrac{1}{y_1^2(x)}\left(
%y_2(x)y_1(x)\partial_y\Big(\dfrac{B}{A}\Big)\big(x,y(x)\big)+(y_1(x))^3\partial_y^2\Big(\dfrac{B}{A}\Big)\big(x,y(x)\big)
%-y_2(x)y_1(x) \partial_y\Big(\dfrac{B}{A}\Big)\big(x,y(x)\big) \right)\\
%&&+y_{1,0}\partial_xG\big(x,y(x)\big)+y_{1,0}\partial_yG\big(x,y(x)\big)\dfrac{B}{A}\big(x,y(x)\big)\\
%&=&\partial_x F\big(x,y(x)\big)y_1(x)+ \partial_y F \big(x,y(x)\big)\dfrac{B}{A}\big(x,y(x)\big)y_1(x) + F\big(x,y(x)\big)y_1(x) \partial_y\Big(\dfrac{B}{A}\Big)\big(x,y(x)\big)\\
%&&+y_1(x)\partial_y^2\Big(\dfrac{B}{A}\Big)\big(x,y(x)\big)+y_{1,0}\partial_xG\big(x,y(x)\big)+y_{1,0}\partial_yG\big(x,y(x)\big)\dfrac{B}{A}\big(x,y(x)\big)\\
%\end{eqnarray*}
As $y_1(x)$ is a common factor of the first terms and $y_{1,0}$ is a common factor of the last terms we get:
\begin{eqnarray*}
0&=&y_1(x) \left(
\partial_x F\big(x,y(x)\big)+ \partial_y F \big(x,y(x)\big)\dfrac{B}{A}\big(x,y(x)\big) + F\big(x,y(x)\big) \partial_y\Big(\dfrac{B}{A}\Big)\big(x,y(x)\big)   \right.\\
&&\left.+\partial_y^2\Big(\dfrac{B}{A}\Big)\big(x,y(x)\big)
\right)+y_{1,0}\Big( \partial_xG\big(x,y(x)\big)+\partial_yG\big(x,y(x)\big)\dfrac{B}{A}\big(x,y(x)\big) \Big).\\
\end{eqnarray*}
Thus
$$0=y_1(x) \Big( D_0(F)+(AF)\partial_y(B/A)+A\partial_y^2(B/A)\Big)\big(x,y(x)\big) +y_{1,0} D_0(G)\big(x,y(x)\big).$$
%\begin{eqnarray*}
%0&=&y_1(x) \Big( D_0(F)+(AF)\partial_y(B/A)+A\partial_y^2(B/A)\Big)\big(x,y(x)\big)\\
%&& +y_{1,0} D_0(G)\big(x,y(x)\big).
%\end{eqnarray*}
This gives the desired conclusion.
\end{proof}

Now we can state the generalization of Theorem \ref{thm:exta+rfi} for the Liouvillian case.\\
\vspace{0.4cm}

\begin{thm}[Liouvillian extactic curve Theorem]\label{thm:exta+lfi}$\,$\\
\begin{enumerate}
\item If $\tilde{E}^N_{D_2}(x_0,y_0)=0$ then the derivation $D_0$ has a Liouvillian first integral with degree smaller than $N$ or a Darbouxian first integral with degree smaller than $2N+3d-1$ or a rational first integral with degree smaller than \mbox{$4N+8d-3$}. Moreover the defining equation of the first integral, equation (Rat), (D) or (L), has coefficients in $\KK$.\\
\item If $D_0$ has a rational or a Darbouxian or a Liouvillian first integral with degree smaller than $N$ then $\tilde{E}^N_{D_2}(x_0,y_0)=0$.
\end{enumerate}
\end{thm}

\begin{proof}%[Proof of Theorem \ref{thm:exta+lfi}]
First suppose that $\tilde{E}^N_{D_2}(x_0,y_0) =0$, and consider a non trivial element
$$P(x,y)y_1^2+Q(x,y)y_2+R(x,y)y_1 \in \ker \tilde{\mathcal{E}}^{N}_{D_2}(x_0,y_0).$$

If $Q=0$ then by Lemma \ref{lem:ker_exta_liouv},  $Py_1+R \in \ker \tilde{\mathcal{E}}^{N}_{D_1}(x_0,y_0)$ and Theorem~\ref{thm:exta+dfi} gives the existence of a Darbouxian first integral with degree smaller than $N$ and defining equation (D) with coefficients in $\KK$ or a first integral with degree smaller than $2N+2d-1$ with coefficients in $\mathbb{K}$.\\

If $Q\neq0$ and $Q \in \ker \tilde{\mathcal{E}}^N_{D_0}(x_0,y_0)$ then by Theorem \ref{thm:exta+rfi} there exists a rational first integral with degree smaller than $N$ with coefficients in $\mathbb{K}$.\\

If $Q\neq 0$ and $Q  \not \in \ker \tilde{\mathcal{E}}^N_{D_0}(x_0,y_0)$ then  we have two situations:\\
In the first situation we have: 
$$D_0(P/Q)+A(P/Q)\partial_y(B/A)+A\partial_y^2(B/A)=0.$$
In this case, Proposition \ref{proprepresent} gives the existence of a Liouvillian first integral with degree smaller than $N$, since $\deg(P),\deg(Q) \leq N$. Moreover, the equation of type (L) giving the existence of a Liouvillian first integral has coefficients in $\KK$ since $P, Q \in \KK(x_0,y_0)[x,y]$. \\

In the second situation we have 
$$D_0(P/Q)+A(P/Q)\partial_y(B/A)+A\partial_y^2(B/A)\neq 0$$
 and thanks to Lemma \ref{lem:ker_exta_liouv}
 \begin{eqnarray*}
&&y_1(x) \Big( D_0(P/Q)+A(P/Q)\partial_y(B/A)+A\partial_y^2(B/A)  \Big) \big(x,y(x)\big)\\
 &&+ y_{1,0} D_0(R/Q)\big(x,y(x)\big)=0.
\end{eqnarray*}
In this case, we set $P_1=A^2Q^2\big(D_0(P/Q)+A(P/Q)\partial_y(B/A)+A\partial_y^2(B/A)\big)$,\\
 \mbox{$Q_1=A^2Q^2D_0(R/Q)$}, and $P_1$, $Q_1$ are polynomials with degree smaller than $2N+3d-1$. The previous equality gives $y_1 P_1 +y_{1,0} Q_1 \in \ker \mathcal{E}_{D_1}^{2N+3d-1,1}$.  Thus thanks to Lemma~\ref{lem:eval_exta_darboux} we get $y_1P_1+Q_1 \in \ker \tilde{\mathcal{E}}^{2N+3d-1,1}_{D_1}(x_0,y_0)$. Theorem \ref{thm:exta+dfi} gives the existence of a Darbouxian first integral with degree smaller than $2N+3d-1$ or the existence of a rational first integral with degree smaller than $4N+8d-3$ with coefficients in $\mathbb{K}$.\\

The second part of the theorem is straightforward and is proved with the strategy used in Theorem~\ref{thm:exta+dfi}.
%Now, we study the second part of the theorem.\\
%If $D_0$ has a Liouvillian first integral with degree smaller than $N$ then we have $\partial_{y}^2 \mathcal{F}/\partial_y \mathcal{F}=P/Q$, with $P,Q \in \bar{\KK}[x,y]$ and $\deg(P), \deg(Q) \leq N$. Proposition~\ref{propinv} implies that $y_1P/Q + y_2/y_1$ is a first integral of $(S'_2)$. This gives 
%$$y_1(x)P/Q\big(x,y(x)\big) + y_2(x)/y_1(x)= c,$$ where $c  \in \bar{\KK}(x_0,y_0)$ putting $y_{1,0}=1,y_{2,0}=0$. Thus 
%$$-cQ\big(x,y(x)\big)y_1(x) +P\big(x,y(x)\big) y_1^2(x) + Q\big(x,y(x)\big)y_2(x)=0,$$
%and then
%$$-cQy_1 +Py_1^2 + Qy_2 \in \bar{\mathbb{K}} \otimes_{\mathbb{K}} \ker \tilde{\mathcal{E}}^N_{D_2}(x_0,y_0).$$
%As the coefficients of $ \tilde{\mathcal{E}}^N_{D_2}(x_0,y_0)$ are in $\mathbb{K}(x_0,y_0)$, this implies that there exists an element with coefficients in $\mathbb{K}(x_0,y_0)$ in $ \ker \tilde{\mathcal{E}}^N_{D_2}(x_0,y_0)$ and thus $\tilde{E}^N_{D_2}(x_0,y_0)=0$. The situation where $D_0$ has a rational or a Darbouxian first integral with degree smaller than $N$ can be done in the same way.
\end{proof}

As before, we deduce that the computation of $F \in \KK(x,y)$ is not restrictive.
\begin{cor}\label{cor:liouvdansK}
Suppose that $D_0$ has a Liouvillian first integral and no Darbouxian nor rational first integral. There exists $F \in \KK(x,y)$ with $\deg F\leq N$ such that equation (L) gives a Liouvillian first integral if and only if 
there exists $\tilde{F} \in \overline{\KK}(x,y)$ with $\deg \tilde{F} \leq N$ such that equation (L) with $\tilde{F}$ gives a Liouvillian first integral.
\end{cor}

%\begin{proof}
%We just apply to a Liouvillian first integral with a defining equation~(L)  with coefficients in $\bar{\mathbb{K}}$ the second part of Theorem \ref{thm:exta+lfi}. This proves that  $\tilde{E}^N_{D_2}(x_0,y_0)$ equals $0$. Then we apply the first part proving there exists a Liouvillian first integral with degree smaller than $N$ or a Darbouxian first integral with degree smaller than $2N+3d-1$ or a rational first integral with degree smaller than $4N+8d-3$ with coefficients in $\mathbb{K}$. As the last two are forbidden by assumption, $D_0$ admits a Liouvillian first integral 
%given by an equation (L) with a rational function with degree smaller than $N$ and with coefficients in $\mathbb{K}$.
%\end{proof}

%%%%%%%%%%%%%%%%%%%%%%%%%%%%%%%%%%%%%%%%%%%%%%%%%%%%%%%%%%%%%%%%%%%%%%%
\subsection{Riccati extactic curve}%$\,$\\
In this subsection we are going to apply the result of Section \ref{sec:exta} to the derivation $D_3$. Then in the following $\big(y(x),y_1(x), y_2(x), y_3(x)\big)$ is a solution of $(S'_3)$ satisfying the initial condition $y(x_0)=y_0$, $y_1(x_0)=y_{1,0}$,  $y_2(x_0)=y_{2,0}, y_3(x_0)=y_{3,0}$ where $x_0$, $y_0$, $y_{1,0}$, $y_{2,0}$  and  $y_{3,0}$ are variables.\\

Now, we are going to generalize Theorem \ref{thm:exta+rfi} to the Riccati case and we follow the same strategy as before.

\begin{defi}\label{def_exta_Ric}
We set 
$$V_3:=\KK[x,y]_{\leq N}y_1^4\oplus \KK[x,y]_{\leq N} (3y_2^2-2y_3y_1)\oplus \KK[x,y]_{\leq N} y_1^2,\quad l_3=\hbox{dim}(V_3).$$
We denote by $\tilde{\mathcal{E}}^{N}_{D_3}(x_0,y_0)$ the linear map $\mathcal{E}^{V_3}_{D_3}$ after the specialization $y_{1,0}=1$, $y_{2,0}=0$, $y_{3,0}=0$.\\
The $N$-th Riccati extactic curve is defined by 
$$ \tilde{E}^N_{D_3}(x_0,y_0)= \det \big( \tilde{\mathcal{E}}^{N}_{D_3}(x_0,y_0) \big).$$
\end{defi}

As before, we begin by proving two Lemmas.

\begin{lem}\label{lem:eval_exta_ric}
We have the following equivalence:
$$4P(x,y)y_1^4+Q(x,y)(3y_2^2-2y_3y_1)+R(x,y)y_1^2 \in \ker \tilde{\mathcal{E}}^{N}_{D_3}(x_0,y_0)$$
 $$\Updownarrow$$ 
\begin{eqnarray*}
&4P(x,y)y_1^4+Q(x,y)(3y_2^2-2y_3y_1)&\\
&+\left(R(x,y)y_{1,0}^2-\left(3\frac{y_{2,0}^2}{y_{1,0}^2}-2\frac{y_{3,0}}{y_{1,0}} \right)Q(x,y)\right)y_1^2 &\in  \ker \mathcal{E}^{V_3}_{D_3}
\end{eqnarray*}
\end{lem}

\begin{proof}
With the same strategy as the one used to prove Lemma~\ref{lem:eval_exta_darboux}, we are going to obtain the desired equivalence.\\

We denote by $\psi(x)$ the solution of $(S'_3)$ such that $\psi(x_0)=y_0$, $\psi_{1}(x_0)=1$, $\psi_{2}(x_0)=0$, $\psi_{3}(x_0)=0$.\\
We consider the transformation 
$$T(y,y_1,y_2,y_3)=(y, \, y_{1,0}y_1, \, y_{1,0}^2y_2+y_{2,0}y_1, \, y_{1,0}^3y_3+3y_{1,0}y_{2,0} y_2+y_{3,0}y_1 ).$$
We set
\begin{eqnarray*}
&&\big(\psi_T(x),\psi_{T,1}(x),\psi_{T,2}(x),\psi_{T,3}(x)\big):=T(\psi(x),\psi_{1}(x),\psi_{2}(x),\psi_{3}(x)\big)\\
&=& \big( \psi(x), \, y_{1,0}\psi_{1}(x), \, y_{1,0}^2\psi_{2}(x)+y_{2,0}\psi_{1}(x), \,  y_{1,0}^3\psi_{3}(x)+3y_{1,0}y_{2,0} \psi_{2}(x)+y_{3,0}\psi_{1}(x) \big).
\end{eqnarray*}

A direct computation shows that $\big(\psi_T(x),\psi_{T,1}(x),\psi_{T,2}(x),\psi_{T,3}(x)\big)$ is a solution of $(S'_3)$ with initial conditions 
$$\psi_T(x_0)=y_0, \psi_{T,1}(x_0)=y_{1,0}, \psi_{T,2}(x_0)=y_{2,0}, \psi_{T,3}(x_0)=y_{3,0}.$$
Then we have the equality $\big(y(x),y_1(x),y_2(x),y_3(x)\big)=\big(\psi_T(x),\psi_{T,1}(x),\psi_{T,2}(x),\psi_{T,3}(x)\big)$.\\

Now suppose that $4P(x,y)y_1^4+Q(x,y)(3y_2^2-2y_3y_1)+R(x,y)y_1^2 \in \ker \tilde{\mathcal{E}}^N_{D_3}(x_0,y_0)$  then
$$4P(x,\psi(x))\psi_{1}^4(x)+Q(x,\psi(x))(3\psi_{2}^2(x)-2\psi_{3}(x)\psi_{1}(x))+R(x,\psi(x))\psi_{1}^2(x)= 0 \mod (x-x_0)^{l_3}.$$
We multiply by $y_{1,0}^4$ the previous equality. Then, in order to bring $\psi_{T,3}$ out, we use the following equality:
$$y_{0}^4 \big( 3\psi_{2}^2 - 2 \psi_{3}\psi_1\big)= 3y_0^4 \psi^2_2 - 2 \big( y_{1,0}^3\psi_{3}+3y_{1,0}y_{2,0} \psi_2+y_{3,0}\psi_1\big) y_{1,0} \psi_1 +6 y_{1,0}^2y_{2,0}\psi_2\psi_1+2 y_{3,0}y_{1,0}\psi_1^2.$$
 This gives: 
 
\begin{tabular}{l}
$4P(x,\psi(x))y_{1,0}^4\psi_{1}^4(x)$\\
$+Q(x,\psi(x))\Big(3y_{1,0}^4\psi_{2}^2(x)  -2\big(y_{1,0}^3\psi_{3}(x)+3y_{1,0}y_{2,0} \psi_{2}(x)+y_{3,0}\psi_{1}(x)    \big) y_{1,0}\psi_{1}(x)$\\
$+6y_{1,0}^2y_{2,0} \psi_{2}(x)\psi_{1}(x)+2y_{3,0}y_{1,0}\psi_{1}^2(x)\Big) +R(x,\psi(x))y_{1,0}^4\psi_{1}^2(x)$\\
$= 0 \mod (x-x_0)^{l_3}$.\\
\end{tabular}\\

By definition of $\psi_1$ and $\psi_3$, we get:\\

\begin{tabular}{l}
$4P(x,\psi(x))\psi_{T,1}^4(x)$\\
$+Q(x,\psi(x))\Big(3y_{1,0}^4\psi_{2}^2(x)  -2\psi_{T,3}(x)\psi_{T,1}(x)
+6y_{1,0}^2y_{2,0} \psi_{2}(x)\psi_{1}(x)+2y_{3,0}y_{1,0}\psi_{1}^2(x)\Big)$\\
$+R(x,\psi(x))y_{1,0}^2\psi_{T,1}^2(x)= 0 \mod (x-x_0)^{l_3}.$
\end{tabular}\\

Now, in order to bring $\psi_{T,2}$ out, we use the equality
$$ 3 y_{1,0}^4 \psi_{2}^2 = 3( y_{1,0}^2 \psi_2 +y_{2,0}\psi_1)^2-6y_{1,0}^2y_{2,0}\psi_1 \psi_2 - 3y_{2,0}^2 \psi_1^2.$$

This gives:\\

\begin{tabular}{l}
$4P(x,\psi(x))\psi_{T,1}^4(x)$\\
$+Q(x,\psi(x))\Big(3(y_{1,0}^2\psi_{2}(x)+y_{2,0}\psi_{1}(x) )^2-6y_{1,0}^2y_{2,0}\psi_{2}(x)\psi_{1}(x)-3y_{2,0}^2\psi_{1}^2(x)$\\
$-2\psi_{T,3}(x)\psi_{T,1}(x)+6y_{1,0}^2y_{2,0} \psi_{2}(x)\psi_{1}(x)+2y_{3,0}y_{1,0}\psi_{1}^2(x)\Big)$\\
$+R(x,\psi(x))y_{1,0}^2\psi_{T,1}^2(x)= 0 \mod (x-x_0)^{l_3}$.
\end{tabular}\\

We can simplify by $6 y_{1,0}^2y_{2,0}\psi_2 \psi_1$ the previous equality. Then by definition of $\psi_{T,2}$ we obtain:\\

\begin{tabular}{l}
$4P(x,\psi(x))\psi_{T,1}^4(x)$\\
$+Q(x,\psi(x))\Big(3\psi_{T,2}^2(x)-3y_{2,0}^2\psi_{1}^2(x)-2\psi_{T,3}(x)\psi_{T,1}(x)+2y_{3,0}y_{1,0}\psi_{1}^2(x)\Big)$\\
$+R(x,\psi(x))y_{1,0}^2\psi_{T,1}^2(x)= 0 \mod (x-x_0)^{l_3}$.
\end{tabular}\\

At last, we factorize by $\psi_1^2$, and we get:\\

\begin{tabular}{l}
$4P(x,\psi_T(x))\psi_{T,1}^4(x)$\\
$+Q(x,\psi_T(x))(3\psi_{T,2}(x)^2-2\psi_{T,3}(x)\psi_{T,1}(x))$\\
$\left(\left(-3\frac{y_{2,0}^2}{y_{1,0}^2}+2\frac{y_{3,0}}{y_{1,0}} \right)Q(x,\psi_T(x))+R(x,\psi_T(x))y_{1,0}^2\right)\psi_{T,1}^2(x)= 0 \mod (x-x_0)^{l_3}$.\\
$\,$\\
\end{tabular}

Therefore, as  $\big(y(x),y_1(x),y_2(x),y_3(x)\big)=\big(\psi_T(x),\psi_{T,1}(x),\psi_{T,2}(x),\psi_{T,3}(x)\big)$ we deduce
$$4P(x,y)y_1^4+Q(x,y)(3y_2^2-2y_3y_1)+\left(R(x,y)y_{1,0}^2-\left(3\frac{y_{2,0}^2}{y_{1,0}^2}-2\frac{y_{3,0}}{y_{1,0}} \right)Q(x,y)\right)y_1^2 $$
belongs to $\ker \mathcal{E}^{V_3}_{D_3}$.\\

The converse is straightforward.
\end{proof}

\begin{lem}\label{lem:ker_exta_ric}
Consider a non trivial element $$4P(x,y)y_1^4+Q(x,y)(3y_2^2-2y_3y_1)+R(x,y)y_1^2 \in \ker \tilde{\mathcal{E}}^{N}_{D_3}(x_0,y_0),$$ then:
\begin{itemize}
\item If $Q=0$ then $4Py_1^2+R \in \ker \tilde{\mathcal{E}}^{N,2}_{D_1}(x_0,y_0)$.
\item If $Q\neq 0$ and $Q \not \in \ker \tilde{\mathcal{E}}^N_{D_0}(x_0,y_0)$ then:
\begin{eqnarray*}
0&=&y_1(x)^2\Big(4D_0(P/Q)+8A(P/Q)\partial_y(B/A)-2A\partial_y^3(B/A)\Big)\big(x,y(x)\big)\\
&&+y_{1,0}^2 D_0(R/Q)\big(x,y(x)\big).
\end{eqnarray*}
\end{itemize}
\end{lem}

\begin{proof}
We have $4P(x,y)y_1^4+Q(x,y)(3y_2^2-2y_3y_1)+R(x,y)y_1^2 \in \ker \tilde{\mathcal{E}}^{N}_{D_3}(x_0,y_0)$, then Lemma~\ref{lem:eval_exta_ric} and Theorem~\ref{prec_gen} gives:
\begin{eqnarray*}
(\star) & & 4P\big(x,y(x)\big)y_1^4(x)+Q\big(x,y(x)\big)(3y_2^2(x)-2y_3(x)y_1(x))\\
&&+\left( R\big(x,y(x)\big)y_{1,0}^2-\left(3\frac{y_{2,0}^2}{y_{1,0}^2}-2\frac{y_{3,0}}{y_{1,0}} \right)Q\big(x,y(x)\big) \right)y_1^2(x) =0.
\end{eqnarray*}

%As $4P(x,y)y_1^4+Q(x,y)(3y_2^2-2y_3y_1)+R(x,y)y_1^2 \in \ker \tilde{\mathcal{E}}^{N}_{D_3}(x_0,y_0)$, we have using Lemma \ref{lem:eval_exta_ric}
%$$4P\big(x,y(x)\big)y_1^4(x)+Q\big(x,y(x)\big)(3y_2^2(x)-2y_3(x)y_1(x))$$
%$$+\left( R\big(x,y(x)\big)y_{1,0}^2-\left(3\frac{y_{2,0}^2}{y_{1,0}^2}-2\frac{y_{3,0}}{y_{1,0}} \right)Q\big(x,y(x)\big) \right)y_1^2(x) =0 \mod (x-x_0)^{l_3}$$
%Then Theorem \ref{prec_gen} applied with the five variables $x,y,y_1,y_2,y_3$ gives
%$$4P\big(x,y(x)\big)y_1^4(x)+Q\big(x,y(x)\big)(3y_2^2(x)-2y_3(x)y_1(x))$$
%$$+\left( R\big(x,y(x)\big)y_{1,0}^2-\left(3\frac{y_{2,0}^2}{y_{1,0}^2}-2\frac{y_{3,0}}{y_{1,0}} \right)Q\big(x,y(x)\big) \right)y_1^2(x) =0.$$
If $Q=0$  then we have \mbox{$4P\big(x,y(x)\big)y_1^4(x)+y_{1,0}^2R\big(x,y(x)\big)y_1^2(x)=0$}. As $y_1(x)\neq 0$, we get $4P\big(x,y(x)\big)y_1^2(x)+y_{1,0}^2 R\big(x,y(x)\big)=0$. And thus $4Py_1^2+R$ belongs to $\ker \tilde{\mathcal{E}}^{N,2}_{D_1}(x_0,y_0)$ by Lemma \ref{lem:eval_exta_darboux}.\\

If $Q\neq 0$ and $Q \not \in \ker \tilde{\mathcal{E}}_{D_0}^N(x_0,y_0)$ then $Q\big(x,y(x)\big)\neq 0$. We set $F=P/Q$ and $G=R/Q$ then  by dividing $(\star)$ by $Q\big(x,y(x)\big)y_1^2(x)$, we have:
$$(\star \star)\quad 4F\big(x,y(x)\big)y_1^2(x)+\left( 3\frac{y_2^2(x)}{y_1^2(x)}-2\frac{y_3(x)}{y_1(x)} \right)+y_{1,0}^2 G\big(x,y(x)\big)-\left(3\frac{y_{2,0}^2}{y_{1,0}^2}-2\frac{y_{3,0}}{y_{1,0}} \right)=0.$$
The derivation relatively to $x$ of $3\frac{y_2^2(x)}{y_1^2(x)}-2\frac{y_3(x)}{y_1(x)}$ and the relation given by the differential system $(S'_3)$ gives:

\begin{eqnarray*}
\partial_x \left( 3\dfrac{y_2^2(x)}{y_1^2(x)}-2\dfrac{y_3(x)}{y_1(x)} \right)&=& 
\Bigg[3\left(2y_2\Big[y_2\partial_y\Big( \dfrac{B}{A} \Big)+y_1^2\partial_y^2\Big( \dfrac{B}{A} \Big)\Big]\dfrac{1}{y_1^2}\right)\\
&&-6y_2^2 y_1\partial_y\Big(\dfrac{B}{A}\Big)\dfrac{1}{y_1^3}\\
&& -2\left( y_3\partial_y\Big( \dfrac{B}{A} \Big) + 3 y_2y_1\partial_y^2\Big(\dfrac{B}{A} \Big) + y_1^3\partial_y^3 \Big( \dfrac{B}{A} \Big) \right)\dfrac{1}{y_1}\\
&&+2y_3y_1 \partial_y\Big(\dfrac{B}{A} \Big)\dfrac{1}{y_1^2}\Bigg]\big(x,y(x),y_1(x),y_2(x),y_3(x)\big)\\
&=&-2y_1^2(x) \partial_y^3\Big( \dfrac{B}{A} \Big)\big(x,y(x)\big).
\end{eqnarray*}

Then the derivation relatively to $x$ of $(\star \star)$  gives:
\begin{eqnarray*}
0=4A^{-1}(x,y(x))D_0(F)(x,y(x)) y_1^2(x)+8F(x,y(x))y_1^2(x)\partial_y\Big(\frac{B}{A}\Big)\\
-2y_1^2(x) \partial_y^3\Big( \dfrac{B}{A} \Big)\big(x,y(x)\big)+y_{1,0}^2A^{-1}(x,y(x))D_0(G)(x,y(x))
\end{eqnarray*}
Thus
\begin{eqnarray*}
0&=&y_1(x)^2\Big(4D_0(F)+8AF\partial_y(B/A)-2A\partial_y^3(B/A)\big)\big(x,y(x)\Big)\\
&& +y_{1,0}^2 D_0(G)\big(x,y(x)\big).
\end{eqnarray*}
This gives the desired conclusion.
\end{proof}

Now we can state the generalization of Theorem \ref{thm:exta+lfi} for the Riccati case.
\begin{thm}[Riccati extactic curve Theorem]\label{thm:exta+cfi}$\,$
\begin{enumerate}
\item If $\tilde{E}^N_{D_3}(x_0,y_0)=0$ then the derivation $D_0$ has a Riccati first integral with degree smaller than $N$ or a $2$-Darbouxian first integral with degree smaller than $2N+4d-1$ or a rational first integral with degree smaller than \mbox{$4N+10d-3$}. Moreover the defining equation of the first integral, (Rat), (D) or (Ric), has coefficients in $\KK$.
\item If $D_0$ has a rational or a $2$-Darbouxian or a Riccati first integral with degree smaller than $N$ then $\tilde{E}^N_{D_3}(x_0,y_0)=0$.
\end{enumerate}
\end{thm}

\begin{proof}%[Proof of Theorem \ref{thm:exta+cfi}]
First consider a non trivial solution 
$$4P(x,y)y_1^4+Q(x,y)(3y_2^2-2y_3y_1)+R(x,y)y_1^2 \in \ker \tilde{\mathcal{E}}^{N}_{D_3}(x_0,y_0).$$

If $Q=0$ then by Lemma \ref{lem:ker_exta_ric}, then $4Py_1^2+R \in \ker \tilde{\mathcal{E}}^{N,2}_{D_1}(x_0,y_0)$ and Theorem \ref{thm:exta+dfi} gives the existence of a $2$-Darbouxian first integral with degree smaller than $N$ and defining equation (D) with coefficients in $\KK$ or a first integral with degree smaller than $2N+2d-1$ and with  coefficients in $\mathbb{K}$.\\

If $Q\neq0$ and $Q \in \ker \tilde{\mathcal{E}}^N_{D_0}(x_0,y_0)$ then by Theorem \ref{thm:exta+rfi} there exists a rational first integral with degree smaller than $N$ with coefficients in $\mathbb{K}$.\\

If $Q\neq0$ and $Q  \not \in \ker \tilde{\mathcal{E}}^N_{D_0}(x_0,y_0)$ then  we have two situations:\\
In the first situation we have: 
$$4D_0(P/Q)+8A(P/Q)\partial_y(B/A)-2A\partial_y^3(B/A)=0.$$
In this case, Proposition \ref{proprepresent} gives the existence of a Riccati first integral. As $\deg P,\deg Q\leq N$, we deduce the existence of a Riccati first integral with degree smaller than $N$.  Moreover, the equation of type (Ric) giving the existence of a Riccati first integral has coefficients in $\KK$ since $P, Q \in \KK(x_0,y_0)[x,y]$.\\
In the second situation we have 
$$4D_0(P/Q)+8A(P/Q)\partial_y(B/A)-2A\partial_y^3(B/A)\neq 0$$ and thanks to  Lemma \ref{lem:ker_exta_ric}
\begin{eqnarray*}
0&=&y_1(x)^2\Big(4D_0(P/Q)+8A(P/Q)\partial_y(B/A)-2A\partial_y^3(B/A)\big)\big(x,y(x)\Big)\\
&&+y_{1,0}^2D_0(R/Q)\big(x,y(x)\big).
\end{eqnarray*}
In this case, we set 
$$P_1=A^3Q^2\big(4D_0(P/Q)+8A(P/Q)\partial_y(B/A)-2A\partial_y^3(B/A)\big),$$ $$Q_1=A^3Q^2D_0(R/Q)$$
 and we obtain $y_1^2P_1+y_{1,0}^2Q_1 \in \ker \tilde{\mathcal{E}}^{2N+4d-1,2}_{D_1}(x_0,y_0,y_{1,0})$. Therefore by Lemma~\ref{lem:eval_exta_darboux}, $y_1^2P_1+Q_1 \in \ker \mathcal{E}_{D_1}^{2N+4d-1,2}(x_0,y_0)$, thus Theorem \ref{thm:exta+dfi} gives the existence of a $2$-Darbouxian first integral with degree smaller than $2N+4d-1$ or the existence of a rational first integral with degree smaller than $4N+10d-3$. Moreover the equation giving this first integral has coefficients in $\KK$ since $P_1, Q_1 \in \KK(x_0,y_0)[x,y]$.\\

The second part of the theorem is  proved with the strategy used in Theorem~\ref{thm:exta+dfi}.
%Now, we study the second part of the theorem.\\
%If $D_0$ has a Riccati first integral with degree smaller than $N$ then we have the following relation $\partial_{y}^2 \mathcal{F}/\mathcal{F}=P/Q$, with $P,Q \in \bar{\KK}[x,y]$ and $\deg(P)$, $\deg(Q)$ smaller than $N$. Proposition \ref{propinv} implies that $4y_1^2P/Q + 3y_2^2/y_1^2-2y_3/y_1$ is a first integral of $(S'_3)$. This gives 
%$$4y_1^2(x)P/Q\big(x,y(x)\big) + 3y_2^2(x)/y_1^2(x)-2y_3(x)/y_1(x)=c,$$ where $c  \in \bar{\mathbb{K}}(x_0,y_0)$ putting $y_1(x_0)=1,y_2(x_0)=0,y_3(x_0)=0$. Thus
%$$4P\big(x,y(x)\big)y_1^4(x)+Q\big(x,y(x)\big)(3y_2^2(x)-2y_3(x)y_1(x))-cQ\big(x,y(x)\big)y_1^2(x)=0,$$
%and then
%$$4Py_1^4+Q(3y_2^2-2y_3y_1)-cQy_1^2 \in\bar{\mathbb{K}}\otimes_{\mathbb{K}} \ker \tilde{\mathcal{E}}^N_{D_3}(x_0,y_0).$$
%Thus there exists a non trivial element with coefficients in $\mathbb{K}$ in $\ker \tilde{\mathcal{E}}^N_{D_3}(x_0,y_0)$, and thus $\tilde{E}^N_{D_3}(x_0,y_0)=0$. The situation where $D_0$ has a rational or a $2$-Darbouxian first integral with degree smaller than $N$ can be done in the same way.
\end{proof}

As before, we remark that the computation of $F \in \KK(x,y)$ is not restrictive. 
\begin{cor}\label{cor:RiccatidansK}
Suppose that $D_0$ has a Riccati first integral and no $2$-Darbouxian nor rational first integral. We have:\\ 
There exists $F \in \KK(x,y)$ with $\deg F \leq N$ such that equation (Ric) gives a Riccati first integral if and only if 
there exists $\tilde{F} \in \overline{\KK}(x,y)$ with $\deg \tilde{F} \leq N$ such that equation (Ric) gives a Riccati first integral.
\end{cor}

%\begin{proof}
%Given a Riccati first integral with coefficients in $\bar{\mathbb{K}}$, we apply the second part of Theorem \ref{thm:exta+cfi}, giving $\tilde{E}^N_{D_3}(x_0,y_0)=0$. We then apply the first part, and knowing that $2$-Darbouxian and rational first integrals are forbidden, the only possibility left is a the existence of a Riccati first integral of degree $\leq N$ with an equation of type (Ric) with coefficients in $\mathbb{K}$.
%\end{proof}

%%%%%%%%%%%%%%%%%%%%%%%%%%%%%%%%%%%%%%%%%%%%%%%%%%%%%%%%%%%
\section{Evaluations of extactic curves}\label{sec:extaeval}

In our algorithms we will not compute the extactic curves as polynomials in $x_0,y_0$. We will only compute the extactic curves evaluated at a random point $(x_0^{\star},y_0^{\star}) \in \KK^2$. If the extactic curve is a non-zero polynomial, then almost surely its evaluation at a random point will not be zero. However, theoretically this can happen, and thus we want to bound the algebraic set on which such kind of bad situations can happen.

\begin{defi}
We denote by $\Sigma_{D_r,N,k}$ (where $k$ is omitted when $r\neq 1$) the following algebraic variety:
$$\Sigma_{D_r,N,k}= \mathcal{V}\big( p\times p \textrm{ minors of } \tilde{\mathcal{E}}^{N,k}_{D_r}(x_0,y_0), \textrm{ where  } p=\textrm{rank } \tilde{\mathcal{E}}^{N,k}_{D_r}\big).$$
\end{defi}

%With this definition we can now state a specialized version of Theorem \ref{prec_gen}. In the following in order to have a uniform statement we set $y_1^{(2)}=y_2$ and $y_1^{(3)}=y_3$.

\begin{lem} \label{lem:specialized_ker}
Let $(x_0^{\star},y_0^{\star}) \in \KK^2$, $1 \leq r \leq 3$ and $P \in \ker \tilde{\mathcal{E}}^{N,k}_{D_r}(x_0^{\star},y_0^{\star})$.\\
We consider $y_{\star}(x), \ldots, y_{3,\star}(x)$ a solution of $(S'_3)$ with initial condition \mbox{$y_{\star}(x_0^{\star})=y_0^{\star}$}, $y_{1,\star}(x_0^{\star})=1$, $y_{2,\star}(x_0^{\star})=y_{3,\star}(x_0^{\star})=0$.\\
If $(x_0^{\star},y_0^{\star}) \not \in \Sigma_{D_r,N,k}$ then
$P\big(x,y_{\star}(x), \ldots, y_{r,\star}(x)\big)=0.$

\end{lem}

\begin{proof}
 By Lemma \ref{lem:eval_exta_darboux}, Lemma \ref{lem:eval_exta_liouv} and Lemma \ref{lem:eval_exta_ric}, we have 
$$\dim_{\KK(x_0,y_0)} \ker \tilde{\mathcal{E}}^{N,k}_{D_r}(x_0,y_0)=\dim_{\LL_r} \ker \mathcal{E}^{V_r}_{D_r},$$
where $\LL_r=\KK(x_0,y_0,y_{1,0}, \ldots, y_{r,0})$.\\
If $(x_0^{\star},y_0^{\star}) \not \in \Sigma_{D_r,N,k}$ then 
$$\dim_{\KK} \ker \tilde{\mathcal{E}}^{N,k}_{D_r}(x_0^{\star},y_0^{\star})= \dim_{\KK(x_0,y_0)} \ker \tilde{\mathcal{E}}^{N,k}_{D_r}(x_0,y_0)= \dim_{\LL_r} \ker \mathcal{E}^{V_r}_{D_r}.$$
 Thus in this situation, if $P \in \ker \tilde{\mathcal{E}}^{N,k}_{D_r}(x_0^{\star},y_0^{\star})$ then there exists an element $\mathcal{P}(x_0,y_0,\ldots,y_{r,0};x,y,\ldots,y_{r})$ in $ \ker \mathcal{E}^{V_r}_{D_r}$ such that 
$$\mathcal{P}(x_0^{\star},y_0^{\star},1,0,0;x,y,y_1,\ldots, y_{r})=P(x,y,\ldots,y_{r}).$$
By Theorem \ref{prec_gen}, we have 
$$\mathcal{P}(x_0,y_0,\ldots,y_{r,0};x,y(x),\ldots,y_{r}(x))=0$$
then after the specialization we get
$$P\big(x,y_{\star}(x), \ldots, y_{r,\star}(x)\big)=0.$$
\end{proof}

In the following, we will need some explicit bounds on the degree of the minors of $\tilde{\mathcal{E}}^{N,k}_{D_r}(x_0,y_0)$.

\begin{lem}\label{lem:boundBi}
\begin{itemize}
\item The degree of a minor of $\tilde{\mathcal{E}}^N_{D_0}(x_0,y_0)$ is smaller than 
$$\mathcal{B}_0(d,N):= \dfrac{N(N+1)(N+2)}{2}+(d-1)\dfrac{(N+1)^2(N+2)^2-(N+1)(N+2)}{8}.$$
\item The degree of a minor of $\tilde{\mathcal{E}}^{N,k}_{D_1}(x_0,y_0)$ is smaller than
\begin{eqnarray*} 
\mathcal{B}_1(d,N) &  :=  & Nl_1+\dfrac{(2d-1)(l_1-1)l_1}{2}\\
& = & N(N+1)(N+2)+(2d-1)\dfrac{(N+1)^2(N+2)^2-(N+1)(N+2)}{2},
\end{eqnarray*}
where $l_1=\dim_{\KK} V_1$, see Definition \ref{def:exta_Darb}.\\

\item The degree of a minor of $\tilde{\mathcal{E}}^N_{D_2}(x_0,y_0)$ is smaller than 
\begin{eqnarray*}
\mathcal{B}_2(d,N) &  :=   & Nl_2+\dfrac{(3d-1)(l_2-1)l_2}{2}\\
& =  &\dfrac{3N(N+1)(N+2)}{2}+\dfrac{3d-1}{2}\Big[ \Big( \dfrac{3}{2}(N+1)(N+2)\Big)^2-\dfrac{3}{2}(N+1)(N+2)\Big], 
\end{eqnarray*}
where $l_2=\dim_{\KK} V_2$, see Definition \ref{def_exta_Liouv}.\\

\item The degree of a minor of $\tilde{\mathcal{E}}^N_{D_3}(x_0,y_0)$ is smaller than 
\begin{eqnarray*}
\mathcal{B}_3(d,N) & :=  & Nl_2+\dfrac{(4d-1)(l_3-1)l_3}{2}\\
&  =  & \dfrac{3N(N+1)(N+2)}{2}+\dfrac{4d-1}{2}\Big[ \Big( \dfrac{3}{2}(N+1)(N+2)\Big)^2-\dfrac{3}{2}(N+1)(N+2)\Big],
\end{eqnarray*}
where $l_3=\dim_{\KK} V_3$, see Definition \ref{def_exta_Ric}.\\
\end{itemize}
\end{lem}

\begin{proof}
The degree of a minor of $\tilde{\mathcal{E}}^N_{D_0}(x_0,y_0)$ is smaller than the degree of the extactic curve. The announced bound is given in \cite{Cheze}. We apply here the same strategy for $\tilde{\mathcal{E}}^{N,k}_{D_1}(x_0,y_0)$:\\
Let $\{v_i\}$ be basis of $V_1$. The degree in $x_0,y_0$ of a minor $\mathcal{M}$ of $\tilde{\mathcal{E}}^N_{D_1}(x_0,y_0)$ satisfies
$$\deg(\mathcal{M}) \leq \sum_{j=0}^{l_1-1} \deg D_1^j(v_i).$$
As $\deg\big(D_1^j(v_i)\big) \leq j(2d-1)+N$, we get 
$$\deg( \mathcal{M}) \leq \sum_{j=0}^{l_1-1}j(2d-1)+N \leq Nl_1+(2d-1) \sum_{j=0}^{l_1-1}j \leq Nl_1+\dfrac{(2d-1)(l_1-1)l_1}{2}.$$
The bounds for $\tilde{\mathcal{E}}^N_{D_2}(x_0,y_0), \tilde{\mathcal{E}}^N_{D_3}(x_0,y_0)$ are obtained in the same way.
\end{proof}

\begin{cor}\label{cor:inclusion_sigma_hypersurface}
The algebraic variety $\Sigma_{D_r,N,k}$ is included in an algebraic hypersurface with degree smaller than $\mathcal{B}_{r}(d,N)$.
\end{cor}

The following set will be also useful to characterize some special situations.

\begin{defi}
We denote by $\mathfrak{S}_N$ the following set:\\
If $D_0$ has no rational first integrals of degree $\leq N$ then

$$\mathfrak{S}_N =\left\{   (x_0,y_0)\in\KK^2    \ 
\middle\vert \begin{array}{l}
    (x_0,y_0)  \textrm{ vanishes an irreducible} \\
  \textrm{Darboux polynomial of degree } \leq N 
  \end{array}\right\}.
  $$

If $D$ has an indecomposable rational first integral $P/Q$ of degree $p\leq N$ then

$$\mathfrak{S}_N =\left\{   (x_0,y_0)\in\KK^2    \ 
\middle\vert \begin{array}{l}
   
  (x_0,y_0)  \textrm{ vanishes an irreducible}
\\
\textrm{Darboux polynomial of degree} < p
  \end{array}\right\}.
  $$
\end{defi}

This set corresponds to non-generic situations. Thus we will try to avoid them. 
Now, we give a bound on this set:
%More precisely, when $D_0$ has no rational first integral we do not want to get a solution $\big(x,y(x)\big)$ of $(S'_0)$  corresponding to a Darboux polynomial.\\
% When $D_0$ has a rational first integral we do not want to get an orbit of degree strictly smaller than the degree of a generic orbit. At last, when considering two initial conditions  we do not want them to be on the same level of this first integral.

\begin{lem}\label{lem:boundspectre}
The algebraic variety $\mathfrak{S}_{N}$ is included in an algebraic hypersurface with degree smaller than $\big(\frac{d(d+1)}{2}+5\big)N$.
\end{lem}
\begin{proof}
If $D_0$ has no rational first integral then by the Darboux-Jouanolou theorem $D_0$ has at most $d(d+1)/2$ irreducible Darboux polynomials, see e.g. \cite{Jou} or \cite{Singer,DLA}. Therefore, if $(x_0,y_0) \in \mathfrak{S}_N$ then $(x_0,y_0)$  vanishes the product of $d(d+1)/2$ bivariate polynomials with degree smaller than $N$. This gives a bound on the degree  which is lower than the bound of the Lemma.\\

If $D_0$ has a rational first integral with degree $p \leq N$ then all irreducible Darboux polynomials divide a linear combination $\lambda P - \mu Q$ where $P/Q$ is an indecomposable rational first integral with degree $p$. By the Darboux-Jouanolou theorem we know that all but finitely many irreducible Darboux polynomials are of the form $\lambda P - \mu Q$ and have degree $p$. The set $\sigma(P,Q)$ of $(\lambda:\mu) \in \PP^1(\overline{\KK})$ such that $\lambda P - \mu Q$ is reducible or has a degree strictly smaller than $p$ is the set of remarkable values. Sometimes this set is called  the spectrum of $P/Q$. It is proved in \cite{ChezeDarbouxJouanolou} that $|\sigma(P,Q)| \leq d(d+1)/2 + 5$. So if $(x_0,y_0) \in \mathfrak{S}_N$ then it vanishes  a polynomial $\lambda P - \mu Q$ where $(\lambda:\mu)$ belongs to $\sigma(P,Q)$. Therefore, if $D_0$ has a rational first integral then $\mathfrak{S}_N$ is included in an algebraic hypersurface with degree smaller than $\big(\frac{d(d+1)}{2}+5\big)N$.
\end{proof}

%%%%%%%%%%%%%%%%%%%%%%%%%%%%%%%%%%%%%%%%%%%%%%%%%%%%%%%%%%%%%%%%%%%%%%%%%%%%
\section{First integral algorithms}\label{sec:algo}
In the following sections we are going to describe our algorithms. As mentionned before we are going to compute rational first integrals for the derivations $D_0, D_1, D_2, D_3$. These rational first integrals are computed thanks to the extactic curves.
 We are going to consider one point $(x_0^{\star},y_0^{\star})\in \KK^2$ and to compute a non trivial element in $\ker \tilde{\mathcal{E}}_{D_r}^{N}(x_0^{\star},y_0^{\star})$. We will see that if $(x_0^{\star},y_0^{\star})$ avoids an algebraic variety then we can compute a symbolic first integral from this element.\\
 %Our strategy is to compute a non trivial element in $\ker \tilde{\mathcal{E}}^{N,k}_{D_r}(x_0^{\star},y_0^{\star})$, where $x_0^{\star},y_0^{\star} \in \KK$.\\

\textbf{\textsf{Compute flow series}}\\
\texttt{Input:} $A(x,y), B(x,y) \in \KK[x,y]$, $(x_0^{\star}, y_0^{\star}) \in \KK^2$, $N \in \NN$, $r \in  [[0;3]]$.\\
\texttt{Output:} $r+1$ series $y_{\star}(x),\ldots,y_{r,\star}(x)$ solutions of $(S'_r) \, \mod (x-x_0^{\star})^{\sigma}$\\  where $\sigma=\min(r+1,3)\frac{(N+1)(N+2)}{2}$, with initial condition $y_{\star}(x_0^{\star})=y_0^{\star}$, \\
$y_{1,\star}(x_0^{\star})=1$, $y_{2,\star}(x_0^{\star})=y_{3,\star}(x_0^{\star})=0$.\\

This subroutine is performed with the algorithm given in \cite{BCLOSSS}.\\

For the following subroutine, we need a weighted degree in order to specified the output. We use the following weighted degree:\\
$\textrm{w-deg}\big(P(x,y,y_1,y_{2},y_{3})\big)=\deg\Big( P\big(x,y,y_1^{N+1},y_{2}^{2N+2},y_{3}^{3N+3}\big)\Big)$.\\

This weighted degree will be useful for the following reason:
$$\textrm{w-deg}\big(y_1^kP(x,y)+Q(x,y)\big) > \textrm{w-deg}\big(\mathcal{Q}(x,y)\big),$$
 for all $P,Q , \mathcal{Q} \in \KK[x,y]_{\leq N}$, when $P \neq 0$.\\
 
Now, if we have an element $y_1 P(x,y)+Q(x,y) \in \ker \tilde{\mathcal{E}}^N_{D_1}(x_0^{\star},y_0^{\star})$ then we can build a Darbouxian first integral from $P$ and $Q$. However, if we have obtained an element $\mathcal{Q}(x,y) \in  \ker \tilde{\mathcal{E}}^N_{D_1}(x_0^{\star},y_0^{\star})$ then we can build a rational first integral. Indeed, thanks to Lemma \ref{lem:ker_exta_darboux} we know that $\mathcal{Q}(x,y) \in \ker \tilde{\mathcal{E}}^N_{D_0}(x_0^{\star},y_0^{\star})$.\\
Therefore, the computation of an element in  $\ker \tilde{\mathcal{E}}^N_{D_1}(x_0^{\star},y_0^{\star})$ with minimal weighted degree allows one  to get a rational first integral instead of a Darbouxian first integral. This means that the computation of a non trivial in $ \ker \tilde{\mathcal{E}}^N_{D_r}(x_0^{\star},y_0^{\star})$ with minimal weighted degree gives symbolic first integral with degree bounded by $N$ in the simplest class of first integral.\\

% Thus, if we compute a non trivial element in $\ker \tilde{\mathcal{E}}^{N}_{D_1}(x_0^{\star},y_0^{\star})$ with minimal weighted degree, this allows one to get a symbolic first integral with degree bounded by $N$ in the simplest class of first integrals.
%% In practice, this allows us to give a first integral with size smaller than $N$ in the simplest class of first integral with this size bound. 
% Indeed, suppose that we want to compute a Darbouxian first integral with degree smaller than $N$ and that we have obtained $\mathcal{Q}(x,y) \in \ker \tilde{\mathcal{E}}^{N}_{D_1}(x_0^{\star},y_0^{\star})$. Then thanks to Lemma~\ref{lem:ker_exta_darboux}, we know that \mbox{$\mathcal{Q}(x,y) \in \ker \tilde{\mathcal{E}}^{N}_{D_0}(x_0^{\star},y_0^{\star})$} and then $\mathcal{Q}$ can be used to construct a rational first integral.\\
%The same remark is valid for the Liouvillian and Riccati case.\\

%Indeed, in order to obtain a non trivial element in $\ker \tilde{\mathcal{E}}_{D_r}^{N,k}$ we are going to compute an Hermite-Padé approximation. We use this strategy because it gives a better complexity than the resolution with a classical method of the system $\tilde{\mathcal{E}}_{D_r}^{N,k}.v=0$, where $v$ is the vector to compute.\\

\textbf{\textsf{Compute solution extactic kernel}}\\
\texttt{Input:} $A(x,y), B(x,y) \in \KK[x,y]$, $(x_0^{\star}, y_0^{\star}) \in \KK^2$, $N \in \NN$, $r \in  [[0;3]]$,\\ 
 $r+1$ series $y_{\star}(x),\ldots,y_{r,\star}(x)$ solutions of $(S'_r) \, \mod (x-x_0^{\star})^{\sigma}$\\ 
 where $\sigma=\min(r+1,3)\frac{(N+1)(N+2)}{2}$, with initial condition $y_{\star}(x_0^{\star})=y_0^{\star}$, $y_{1,\star}(x_0^{\star})=1$, $y_{2,\star}(x_0^{\star})=0$,    $y_{3,\star}(x_0^{\star})=0$.\\
\texttt{Output:}  A non trivial element in $\ker \tilde{\mathcal{E}}^{N,k}_{D_r}(x_0^{\star},y_0^{\star})$, if it exists, with minimal weighted degree, or ``None".\\

This subroutine can be reduced  to a linear algebra problem. Indeed, we just have to find a polynomial $\mathcal{S}(x,y,y_1,y_2,y_3)$ such that  
$$\mathcal{S}\big( x, y_{\star}(x),\ldots,y_{3,\star}(x)\big) = 0 \mod (x- x_0^{\star})^{\sigma},$$ where $\mathcal{S} \in V_r$ and  $\deg_{x,y}(\mathcal{S}) \leq N$.\\
 We recall that in our study $V_r$ is the vector space associated to the extactic curve $\tilde{E}^{N,k}_{D_r}(x_0,y_0)$, see Definition~\ref{def:exta_Darb}, Definition~\ref{def_exta_Liouv} and Definition~\ref{def_exta_Ric}.\\
 We will give some details in Section~\ref{sec:complexity}, in order to explain how we can compute $\mathcal{S}$ efficiently.

%\begin{enumerate}
%\item Compute a Hermite-Pad\'e approximation of 
%$$\big(y(x),\ldots,y^N(x),y_1(x),\ldots,y_1(x)y^N(x),\ldots,y_1^{(r)}(x),\ldots,y_1^{(r)}(x)y^N(x)\big)$$ with minimal weighted degree, see Section \ref{sec:complexity}.
%\item Construct from this approximation a polynomial $J \in \KK[x,y,y_1,y_1^{(2)},y_1^{(3)}]$, such that $J\big(x,y(x),\ldots,y^{(r)}(x)\big)=0 \mod (x-x_0^{\star})^{\sigma}$.
%\item If $J \in V_r$ and $\deg_{x,y} J \leq N$ then Return $J$, Else Return ``None".\\
%\end{enumerate}

%%%%%%%%%%%%%%%%%%%%%%%%%%%%%%%%%%%%%%%%%%%%%%%%%%%%%%%%%%%%
\subsection{Rational first integrals}%$\,$\\
An algorithm which computes a rational first integral with degree smaller than $N$ has been described in \cite{BCCW}. Here, we use the same kind of approach. However, we  only use one random point $(x_0^{\star},y_0^{\star}) \in \KK^2$. In  \cite{BCCW}, two random points were used.\\

\textbf{\textsf{Compute Rational first integral}}\\
\texttt{Input:} $A,B \in \KK[x,y]$, $(x_0^{\star},y_0^{\star}) \in \KK^2$, $N \in \NN$\\
\texttt{Output:} An equation $(Eq_0): \mathcal{F}-F=0$ where $F(x,y) \in \KK(x,y)\setminus \KK$, or ``None" or ``I don't know".

\begin{enumerate}
\item \label{step:testARational}  If $A(x_0^{\star},y_0^{\star})=0$ then Return ``I don't know".
\item \textsf{Compute flow series}($A,B,x_0^{\star},y_0^{\star},N,0$)=: $y_{\star}(x)$.
\item \label{step:rational_compute_exta_ker} \textsf{Compute solution extactic kernel}($A,B,y_{\star}(x),N,0$)=:$\mathcal{S}$.\\
If $\mathcal{S}$=``None", then Return ``None", else $\mathcal{S}=: P$.
\item \label{step:rational_calculPi} \textsf{Build Rational first integral}($A,B,P,x_0^{\star},y_0^{\star}$).
\end{enumerate}

Now, we describe the algorithm  \textsf{Build Rational first integral}.\\

\textbf{\textsf{Build Rational first integral}}\\
\texttt{Input:} $P,A,B \in \KK[x,y]$, $(x_0^{\star},y_0^{\star})\in\mathbb{K}^2$.\\
\texttt{Output:} An equation $(Eq_0): \mathcal{F}-F=0$, where $F(x,y) \in \KK(x,y)\setminus \KK$, or ``I don't know".
\begin{enumerate}
\item Compute $P_{red}=P/\gcd(P, \partial_x P, \partial_y P)$.
\item Compute the factorization $\gcd\big(P_{red},D_0(P_{red})\big)=\prod_{j=1}^l L_j(x,y)$,  where $L_j$ are irreducible in $\KK[x,y]$ and set $i:=1$.
\item While $L_i(x_0^{\star},y_0^{\star}) \neq 0$ do $i:=i+1$.
\item If $i>l$ then Return(``I don't know"), \\
Else\begin{enumerate}
\item $\Omega:=\dfrac{D_0(L_i)}{L_i}$;
\item\label{step:build_rat_syst} Compute a basis $\mathcal{B}=\{b_1,b_2,\ldots\}$ of the kernel of the linear system $D_0(Q)=\Omega Q$, where $Q  \in \KK[x,y]_{\leq \deg(L_i)}$.
\item If $|\mathcal{B}|=1$ then Return(``I don't know"),\\
 Else Return$\Big( (Eq_0): \mathcal{F}-\dfrac{b_1}{b_2}=0\Big)$.
\end{enumerate}

\end{enumerate}
\begin{prop}\label{prop:algo_rat_proba_correct}
The algorithm \textsf{Compute Rational first integral} satisfies the following properties:
\begin{itemize}
\item If it returns ``None" then the derivation $D_0$ has no rational first integral with degree smaller than $N$.
\item If it returns an equation $(Eq_0)$ then this equation leads to a rational first integral of $D_0$.
\item If it returns ``I don't know'', then $(x_0^{\star},y_0^{\star})$ belongs to
$$\Sigma_0 \cup \mathfrak{S}_{N},$$
where $\Sigma_0=\mathcal{V}(A)\cup \Sigma_{D_0,N}$.
\end{itemize}
\end{prop}

\begin{proof}
If the algorithm returns ``None", this means that $\tilde{E}^N_{D_0}(x_0^{\star},y_0^{\star}) \neq 0$, thus $\tilde{E}^N_{D_0}(x_0,y_0) \neq 0$. Therefore, Theorem~\ref{thm:exta+rfi} implies that $D_0$ has no rational first integral.\\

If the algorithm returns $(Eq_0)$ then by construction we have $D_0(b_i)=\Omega b_i$, for $i=1,2$. Thus 
$$D_0\Big( \dfrac{b_1}{b_2} \Big)=\dfrac{D_0(b_1)b_2-b_1D_0(b_2)}{b_2^2}=0.$$
Furthermore, $\frac{b_1}{b_2} \not \in \KK$ because $b_1$ and $b_2$ are linearly independent. Then, we deduce that $\frac{b_1}{b_2}$ is a rational first integral.\\

Now, we suppose that $(x_0^{\star},y_0^{\star}) \not \in \Sigma_0 \cup \mathfrak{S}_N$ and we are going to prove that the algorithm does not return ``I don't know".\\

First, if $\tilde{E}^N_{D_0}(x_0^{\star},y_0^{\star}) \neq 0$, then as seen before, the algorithm returns the correct output ``None".\\

Second, suppose that $\dim_{\KK} \tilde{\mathcal{E}}^N_{D_0}(x_0^{\star},y_0^{\star}) \neq 0$.\\
 In Step \ref{step:rational_compute_exta_ker} of \textsf{Compute Rational first integral} we have: $\mathcal{S} \in \ker  \tilde{\mathcal{E}}^N_{D_0}(x_0^{\star},y_0^{\star})$. As $(x_0^{\star},y_0^{\star}) \not \in \Sigma_{D_0,N}$, Lemma \ref{lem:specialized_ker} implies  $\mathcal{S}\big(x,y_{\star}(x)\big)=0$. Then an irreducible factor $L_i$ of $\mathcal{S}$ is an irreducible Darboux polynomial which vanishes at $(x_0^{\star},y_0^{\star})$.\\
As $(x_0^{\star},y_{0}^{\star}) \not \in \mathfrak{S}_N$, we deduce that there exists  a rational first integral $R/Q$ with degree $p \leq N$ and that the factor $L_i$ of $\mathcal{S}$ is a Darboux polynomial of degree $p$ of the form $\lambda R - \mu Q$, where $Q,R \in\KK[x,y]$. Moreover, a rational first integral is defined up to an homography. Then the rational first integral can  be written $L_i/Q \in \KK(x,y) \setminus\KK$, where $Q$ is linearly independent of $L_i$ over $\KK$.

 We have then
$$D_0\Big( \dfrac{L_i}{Q} \Big) =0 \Rightarrow \dfrac{ D_0(L_i)Q-L_i D_0(Q)}{Q^2}=0 \Rightarrow \dfrac{D_0(L_i)}{L_i}=\dfrac{D_0(Q)}{Q}.$$
The last equality defines the polynomial $\Omega$ and we remark that the linear system considered in Step~\ref{step:build_rat_syst} of \textsf{Build Rational first integral}  has two independent solutions: $L_i$ and $Q$. Thus the basis $\mathcal{B}$ contains two distinct elements $b_1$ and $b_2$. As seen before, these two elements give a rational first integral.\\
In conclusion, if $(x_0^{\star},y_0^{\star}) \not \in \Sigma_0 \cup \mathfrak{S}_N$, the algorithm does not return ``I don't know".
\end{proof}

%%%%%%%%%%%%%%%%%%%%%%%%%%%%%%%%%%%%%%%%%%%%%%%%%%%%
\subsection{Darbouxian first integrals}%$\,$\\
This section describes how to use the results given in the previous section in order to get an efficient probabilistic algorithm computing Darbouxian first integrals. The strategy is the following: First, we compute a non trivial element $y_1^kP +Q \in \ker \tilde{\mathcal{E}}_{D_1}^{N,k}(x_0^{\star},y_0^{\star})$. Second, from $P$ and $Q$ we build a Darbouxian first integral.\\

\textbf{\textsf{Compute Darbouxian first integral}}\\
\texttt{Input:} $A,B \in \KK[x,y]$, $(x_0^{\star},y_0^{\star}) \in \KK^2$, $N \in \NN$, $k\in\mathbb{N}^*$ (by default $k=1$)\\
\texttt{Output:} An equation $(Eq_1): \partial_y \mathcal{F} -F^{1/k}=0$, where $F(x,y) \in \KK(x,y)\setminus\{0\}$, or an equation $(Eq_0): \mathcal{F}-F=0$ where $F(x,y) \in \KK(x,y)\setminus \KK$,\\
 or ``None" or ``I don't know".
\begin{enumerate}
\item \label{step:testADarboux}  If $A(x_0^{\star},y_0^{\star})=0$ then Return ``I don't know".
\item \textsf{Compute flow series}($A,B,x_0^{\star},y_0^{\star},N,1$)=: $y_{\star}(x),y_{1,\star}(x)$.
\item \label{step:darboux_compute_exta_ker} \textsf{Compute solution extactic kernel}($A,B,y_{\star}(x),y_{1,\star}(x),N,1,k$)=:$\mathcal{S}$.\\
If $\mathcal{S}$=``None", then Return ``None", Else $\mathcal{S}=: y_1^kP+Q$.
\item \label{step:darboux_calculPi} \textsf{Build Darbouxian first integral}($A,B,P,Q,x_0^{\star},y_0^{\star},k$).
\end{enumerate}

Now, we describe the algorithm \textsf{Build Darbouxian first integral}.\\

\textbf{\textsf{Build Darbouxian first integral}}\\
\texttt{Input:} $A(x,y),B(x,y)$, $P(x,y), Q(x,y) \in \KK[x,y]$ with $(P,Q)\neq (0,0)$, $(x_0^{\star},y_0^{\star})$ in $\mathbb{K}^2$, $k\in\mathbb{N}^*$ (by default $k=1$).\\
\texttt{Output:} An equation $(Eq_1): \partial_y \mathcal{F} -F^{1/k}=0$, where $F(x,y) \in \KK(x,y)\setminus\{0\}$, or an equation $(Eq_0): \mathcal{F}-F=0$ where $F(x,y) \in \KK(x,y) \setminus \KK$\\
 or ``I don't know".

\begin{enumerate}
\item\label{step_build_darboux:p=0} If $P=0$ then Return\big(\textsf{Build Rational first integral}($Q,A,B,x_0^{\star},y_0^{\star}$)\big).
\item\label{step_build_darboux:q=0} If $Q=0$ then Return\big(\textsf{Build Rational first integral}($P,A,B,x_0^{\star},y_0^{\star}$)\big).
\item $R_1:=APQ(P/Q)^{-1/k}\big(D_0((P/Q)^{1/k}) + A (P/Q)^{1/k}\partial_y(B/A)\Big).$
\item\label{step_build_darboux:fin} If $R_1=0$ then Return $(Eq_1): \partial_y \mathcal{F} -(P/Q)^{1/k}=0$,\\
Else Return\big(\textsf{Build Rational first integral} ($R_1,A,B,x_0^{\star},y_0^{\star}$)\big).
\end{enumerate}

\begin{prop}\label{prop:algo_darboux_proba_correct}
The algorithm \textsf{Compute Darbouxian first integral} satisfies the following properties:
\begin{itemize}
\item If it returns ``None" then the derivation $D_0$ has  no $k$-Darbouxian nor rational first integral with degree smaller than $N$.
\item If it returns an equation $(Eq_0)$ or $(Eq_1)$ then this equation leads to a first integral of $D_0$.
\item If it returns ``I don't know'', then $(x_0^{\star},y_0^{\star})$ belongs to
$$\Sigma_1 \cup \mathfrak{S}_{2N+2d-1},$$
where $\Sigma_1=\mathcal{V}(A)\cup \Sigma_{D_0,N}\cup \Sigma_{D_1,N,k}$.
\end{itemize}
\end{prop}

\begin{proof}
If the algorithm returns ``None", this means that $\tilde{E}^{N,k}_{D_1}(x_0^{\star},y_0^{\star}) \neq 0$. Theorem \ref{thm:exta+dfi} implies that $D_0$ has no rational nor $k$-Darbouxian first integral with degree smaller than $N$.\\

If the algorithm returns $(Eq_1)$ this means that we have $R_1=0$ in \textsf{Build Darbouxian first integral}. Proposition \ref{proprepresent} gives then the desired result.\\

If the algorithm returns $(Eq_0)$ then this result comes from \textsf{Build Rational first integral}. We have seen in Proposition~\ref{prop:algo_rat_proba_correct} that this output is correct.\\

Now, we prove the last point of the proposition. We suppose that $(x_0^{\star},y_0^{\star})$ does not belong to 
 $\Sigma_1 \cup \mathfrak{S}_{2N+2d-1}$,
%$$(\KK^2\times(A^{-1}(0)\cup \Sigma_{D_0,N}\cup \Sigma_{D_1,N,k}))\cup ((A^{-1}(0)\cup \Sigma_{D_0,N}\cup \Sigma_{D_1,N,k}) \times \KK^2) \cup \mathfrak{S}_{2N+2d-1},$$
and we are going to prove that the algorithm does not return ``I don't know".\\

First, if $\dim_{\KK} \ker \tilde{\mathcal{E}}^{N,k}_{D_1}(x_0^{\star},y_0^{\star})=0$ then in Step \ref{step:darboux_compute_exta_ker} of \textsf{Compute Darbouxian first integral} we have $\mathcal{S}=$``None". Thus the algorithm returns ``None".\\

Second, we suppose that
$\dim_{\KK} \ker \tilde{\mathcal{E}}^{N,k}_{D_1}(x_0^{\star},y_0^{\star})\neq 0$.\\
In Step \ref{step:darboux_compute_exta_ker} of \textsf{Compute Darbouxian first integral} we have $\mathcal{S}=y_1^kP+Q$.\\

$\bullet$ If $P=0$, then  $Q\neq 0$ and $Q \in \ker \tilde{\mathcal{E}}^N_{D_0}(x_0^{\star},y_0^{\star})$ thanks to Lemma~\ref{lem:ker_exta_darboux}. Then in Step~\ref{step_build_darboux:p=0} of \textsf{Build Darbouxian first integral} the algorithm returns an equation $(Eq_0)$ thanks to Proposition~\ref{prop:algo_rat_proba_correct} and the inclusion  $\Sigma_0 \cup \mathfrak{S}_N \subset \Sigma_1 \cup\mathfrak{S}_{2N+2d-1}$.

$\bullet$ If $P\neq 0$ then $Q\neq0$. Indeed, if $Q=0$ then as $(x_0^{\star},y_0^{\star}) \not \in \Sigma_{D_1,N,k}$ we have, thanks to Lemma \ref{lem:specialized_ker}, $y_{1,\star}(x)^kP\big(x,y_{\star}(x)\big)=0$. Since $y_{1,\star}(x) \neq 0$ we deduce that $P\big(x,y_{\star}(x)\big)=0$. Therefore a factor $\mathcal{P}$ of $P$ is a Darboux polynomial which vanishes at  $(x_0^{\star},y_0^{\star})$. It would give a non trivial element in $\ker \tilde{\mathcal{E}}^{N,k}_{D_1}(x_0^{\star},y_0^{\star})$. This is absurd since \textsf{Compute solution extactic kernel} returns a solution with minimal weighted degree. It follows $Q \neq 0$.\\
Furthermore, we have $Q \not \in \ker \tilde{\mathcal{E}}^N_{D_0}(x_0^{\star},y_0^{\star})$. Indeed, since $(x_0^{\star},y_0^{\star}) \not \in \Sigma_{D_0,N}$, the contrary would imply $Q\big(x,y_{\star}(x)\big)=0$, and gives a non trivial element in $\ker \tilde{\mathcal{E}}^{N,k}_{D_1}(x_0^{\star},y^{\star})$. This is absurd since $y_1^kP+Q$ is a solution with minimal weighted degree.\\
In \textsf{Build Darbouxian first integral}, we thus compute $R_1$. \\
If $R_1=0$ then by Proposition \ref{proprepresent} we get a Darbouxian first integral.\\
Now, we suppose that $R_1 \neq 0$, then in Step~\ref{step_build_darboux:fin} of \textsf{Build Darbouxian first integral}, we use the algorithm \textsf{Build Rational first integral}.\\
As $(x_0^{\star},y_0^{\star}) \not \in \Sigma_{D_1,N,k}$ we have $y_{1,\star}(x)^kP\big(x,y_{\star}(x)\big)+Q\big(x,y_{\star}(x)\big)=0$ thanks to Lemma~\ref{lem:specialized_ker}. Therefore, with the  strategy used in Lemma \ref{lem:ker_exta_darboux}, we obtain $R_1\big(x,y_{\star}(x)\big)=0$ and then $R_1 \in \ker \tilde{\mathcal{E}}_{D_0}^{2N+2d-1}(x_0^{\star},y_0^{\star})$. Thus, we deduce that the  algorithm returns an equation $(Eq_0)$ giving a rational first integral thanks to Proposition~\ref{prop:algo_rat_proba_correct} and the inclusion $\Sigma_0 \cup \mathfrak{S}_{2N+2d-1} \subset \Sigma_1 \cup\mathfrak{S}_{2N+2d-1}$.\\

In conclusion if $(x_0^{\star}, y_0^{\star}) \not \in  \Sigma_1 \cup\mathfrak{S}_{2N+2d-1}$ the algorithm does not return ``I don't know".
%
% Therefore a factor $\mathcal{P}$ of $R_1$ gives an irreducible Darboux polynomial with degree smaller than $2N+2d-1$ which vanishes at $(x_0^{\star},y^{\star})$. Thus for $(x_0^{\star},y_0^{\star})$ (resp. $(x_0^{\star},\mathtt{y}_0^{\star})$) we obtain an irreducible Darboux polynomial denoted by $\mathcal{P}_0$ (resp. $\mathcal{P}_1$). As $(x_0^{\star},y_0^{\star},x_0^{\star},\mathtt{y}_0^{\star}) \not \in \mathfrak{S}_{2N+2d-1}$, this implies that $D_0$ admits a rational first integral of degree $\leq 2N+2d-1$ and $\mathcal{P}_0,\mathcal{P}_1$ define two different levels of this first integral. Then $\mathcal{P}_0/\mathcal{P}_1$ is not constant and thus gives a rational first integral. Therefore, the test at Step \ref{algo:darboux_test_ipr} is satisfied, and the algorithm returns a rational first integral. Thus when $R_1 \neq 0$ the algorithm does  not return ``I don't know". This concludes the proof.
\end{proof}

\begin{prop}\label{prop:condition_proba_algo_darboux}
We set
$$\mathcal{D}(d,N)=d+\mathcal{B}_0(d,N)+ \mathcal{B}_1(d,N)+\Big(\dfrac{d(d+1)}{2}+5\Big)(2N+2d-1).$$
There exists a polynomial $H$ with degree smaller than $\mathcal{D}(d,N)$ such that:\\
If $H(x_0^{\star},y_0^{\star}) \neq 0$ then  \textsf{Compute Darbouxian first integral} returns ``None" or an equation leading to a first integral.
\end{prop}

\begin{proof}
From Corollary \ref{cor:inclusion_sigma_hypersurface} we deduce the existence of a polynomial $\tilde{H}_1$ such that:
$$\Sigma_{D_0,N} \cup \Sigma_{D_1,N,k}\cup \mathcal{V}(A) \subset \mathcal{V}(\tilde{H}_1)$$
where $\deg(\tilde{H}_1) \leq d+\mathcal{B}_0(d,N)+\mathcal{B}_1(d,N)$. We also have from Lemma \ref{lem:boundspectre} the existence of a polynomial $\tilde{H}_2$ of degree $\Big(\dfrac{d(d+1)}{2}+5\Big)(2N+2d-1)$ such that:
$$\mathfrak{S}_{2N+2d-1} \subset \mathcal{V}(\tilde{H}_2).$$
Thus the polynomial
$$H(x_0,y_0)=\tilde{H}_1(x_0,y_0)\tilde{H}_2(x_0,y_0)$$
vanishes on the set given in Proposition \ref{prop:algo_darboux_proba_correct} in the ``I don't know'' part. So if  $H(x_0^{\star},y_0^{\star}) \neq 0$ then  \textsf{Compute Darbouxian first integral} returns ``None" or an equation leading to a first integral, and the degree of $H$ satisfies the degree bound.
\end{proof}

\begin{cor}\label{prop:proba_algo_darboux}
Let $\Omega$ a finite subset of $\KK$ of cardinal $|\Omega|$ greater than $\mathcal{D}(d,N)$ and assume that in \textsf{Compute Darbouxian first integral}  $x_0^{\star}$, $y_0^{\star}$ are chosen independently and uniformly at random in $\Omega$. Then,  \textsf{Compute Darbouxian first integral} returns ``None'' or an equation leading to a first integral with probability at least
$$1-\dfrac{\mathcal{D}(d,N)}{|\Omega|}.$$
\end{cor}

\begin{proof}
This follows from Proposition \ref{prop:condition_proba_algo_darboux},  and Zippel-Schwartz's lemma, see \cite{GG}.
\end{proof}

\begin{rem}
In practice, see Section~\ref{sec:examples}, we never obtain the output ``I don't know".
%In practice the ``practical" probability will be much better, see Section~\ref{sec:examples}.
\end{rem}

\begin{prop}\label{prop:darboux-deg-min}
If $D_0$ admits a rational or Darbouxian first integral with degree smaller than $N$ then \textsf{Compute Darbouxian first integral} returns an equation with minimal degree.
\end{prop}

\begin{proof}
This follows directly from the fact that \textsf{Compute solution extactic kernel} returns a solution with minimal weighted degree.
\end{proof}

%%%%%%%%%%%%%%%%%%%%%%%%%

\subsection{Liouvillian first integrals}
Here, in order to compute a Liouvillian first integral, we follow the same strategy as before:\\
 First, we  compute a non trivial element $P(x,y) y_1^2+Q(x,y)y_2+R(x,y)y_1$ in $\ker \tilde{\mathcal{E}}^N_{D_2}(x_0^{\star},y_0^{\star})$. Second, we build from $P$, $Q$, $R$ a Liouvillian first integral.\\

\textbf{\textsf{Compute Liouvillian first integral}}\\
\texttt{Input:} $A,B \in \KK[x,y]$, $(x_0^{\star},y_0^{\star})\in \KK^2$, $N \in \NN$\\
\texttt{Output:} An equation $(Eq_2):\partial_y^2 \mathcal{F} - F\partial_y\mathcal{F}=0$, where $F(x,y) \in \KK(x,y)$,\\
 or $(Eq_1): \partial_y \mathcal{F} -F=0$, where $F(x,y) \in \KK(x,y) \setminus \{0\}$\\ or  $(Eq_0): \mathcal{F}-F=0$ where $F(x,y) \in \KK(x,y) \setminus \KK$,\\
  or ``None" or ``I don't know''.

\begin{enumerate}
\item \label{step:testA liouv} If $A(x_0^{\star},y_0^{\star})=0$ then Return ``I don't know".
\item \textsf{Compute flow series}($A,B,x_0^{\star},y^{\star},N,2$)=: $y_{\star}(x),y_{1,\star}(x),y_{2,\star}(x)$.
\item \label{algo:liouv_computeker}\textsf{Compute solution extactic kernel}($A,B,y_{\star}(x),y_{1,\star}(x),y_{2,\star}(x),N,2$)=:$\mathcal{S}$.\\
If $\mathcal{S}$=``None", then Return(``None"), \\
Else $\mathcal{S}=: P(x,y)y_1^2+Q(x,y)y_2+R(x,y)y_1$.
\item \textsf{Build Liouvillian first integral}($A,B,P,Q,R,x_0^{\star},y_0^{\star}$).\\

%=: $\mathcal{P}_i$.
%\item If $\mathcal{P}_i \not \in \KK[x,y]$ then Return($\mathcal{P}_i$).
%\item $i:=1+1$.
%\end{enumerate}
%\item \label{algo:liouv_test_ipr} If $\mathcal{P}_0/\mathcal{P}_1 \notin \mathbb{K}$ and $D_0(\mathcal{P}_0/\mathcal{P}_1)=0$ then Return($(Eq_0): \mathcal{F}-\mathcal{P}_0/\mathcal{P}_1=0)$,\\
%    else Return ``I don't know".\\
\end{enumerate}

\textbf{\textsf{Build Liouvillian first integral}}\\
\texttt{Input:} $A(x,y), B(x,y)$, $P(x,y),Q(x,y),R(x,y) \in \KK[x,y]$ such that \mbox{$(P,Q,R)\neq 0$}, $(x_0^{\star},y_0^{\star})\in\mathbb{K}^2$.\\
\texttt{Output:}An equation $(Eq_2):\partial_y^2 \mathcal{F}-F(x,y)\partial_y\mathcal{F}=0$, or $(Eq_1): \partial_y \mathcal{F} -F=0$, where $F(x,y) \in \KK(x,y)\setminus\{0\}$, or $(Eq_0): \mathcal{F}-F=0$, where $F(x,y) \in \KK(x,y)\setminus \KK$, or ``I don't know".

\begin{enumerate}
\item If $Q=0$ then Return(\textsf{Build Darbouxian first integral}($A,B,P,R,x_0^{\star},y_0^{\star}$)).
\item Compute $P_1:=A^3Q^2\big(D_0(P/Q)+A(P/Q)\partial_y(B/A)+A\partial_y^2(B/A)\big)$,\\ $Q_1:=A^3Q^2D_0(R/Q)$.
\item \label{step:daroux_dans_liouv} If $P_1=0$ then Return $(Eq_2): \partial_y^2 \mathcal{F}-(P/Q) \partial_y\mathcal{F}=0$,\\
Else Return\big(\textsf{Build Darbouxian first integral}($A,B,P_1,Q_1,x_0^{\star},y_0^{\star}$)\big).\\
\end{enumerate}

\begin{prop}\label{prop:algo_liouv_proba_correct}
The algorithm \textsf{Compute Liouvillian first integral}  satisfies the following properties:
\begin{itemize}
\item If it returns ``None" then the derivation $D_0$ has no Liouvillian  nor Darbouxian nor rational first integral with degree smaller than $N$.
\item If it returns an equation $(Eq_0)$ or $(Eq_1)$ or $(Eq_2)$ then this equation leads to a non-trivial first integral.
\item If it returns ``I don't know'', then $(x_0^{\star},y_0^{\star})$ belongs to
$$\Sigma_{2} \cup \mathfrak{S}_{4N+8d-3},$$
where 
$\Sigma_{2}=\mathcal{V}(A)\cup \Sigma_{D_0,N}\cup \Sigma_{D_1,N}\cup \Sigma_{D_2,N}$.
\end{itemize}
\end{prop}

\begin{proof}
If the algorithm returns ``None", this means that $\tilde{E}^N_{D_2}(x_0^{\star},y_0^{\star}) \neq 0$. Theorem \ref{thm:exta+lfi} implies that $D_0$ has no rational nor Darbouxian nor Liouvillian first integral with degree smaller than $N$.\\

If the algorithm returns $(Eq_2)$ this means that  we have $P_1=0$ in \textsf{Build Liouvillian first integral}. Proposition \ref{proprepresent} gives then the desired result.\\

If the algorithm returns $(Eq_1)$ this result is correct thanks to Proposition~\ref{prop:algo_darboux_proba_correct}. Indeed, the algorithm returns $(Eq_1)$ when \textsf{Build Liouvillian first integral} uses \textsf{Build Darbouxian first integral}.\\

If the algorithm returns $(Eq_0)$ then the output is correct as shown in Proposition \ref{prop:algo_rat_proba_correct}.\\

Now, we prove the last point of the proposition and we suppose that $(x_0^{\star},y_0^{\star})$ does not belong to
$\Sigma_{2} \cup \mathfrak{S}_{4N+8d-3}$.\\

%$$(\KK^2\times(\mathcal{V}(A)\cup \Sigma_{D_0,N}\cup \Sigma_{D_1,N}\cup \Sigma_{D_2,N} ))\cup ((\mathcal{V}(A)\cup \Sigma_{D_0,N}\cup \Sigma_{D_1,N}\cup \Sigma_{D_2,N}) \times \KK^2) \cup \mathfrak{S}_{4N+8d-3}.$$

First, if $\dim_{\KK}\ker \tilde{\mathcal{E}}^N_{D_2}(x_0^{\star},y_0^{\star}) =0$ then in Step \ref{algo:liouv_computeker} of \textsf{Compute Liouvillian first integral} we have $\mathcal{S}=$``None". Thus the algorithm returns ``None".\\

Second, we suppose that $\dim_{\KK} \ker \tilde{\mathcal{E}}^N_{D_2}(x_0^{\star},y_0^{\star}) \neq 0$.\\
In Step \ref{algo:liouv_computeker} of \textsf{Compute Liouvillian first integral} we have $\mathcal{S}=Py_1^2+Qy_2+R y_1$.\\

$\bullet$ If $Q=0$ then $Py_1+R \in \ker \tilde{\mathcal{E}}_{D_1}^N(x_0^{\star},y_0^{\star})$ as shown in Lemma~\ref{lem:ker_exta_liouv}. Then we are in the Darbouxian case. Therefore Proposition~\ref{prop:algo_darboux_proba_correct} and the inclusion $\Sigma_1 \cup \mathfrak{S}_{2N+2d-1} \subset \Sigma_2 \cup \mathfrak{S}_{4N+8d-3}$ allow us to conclude in this situation.\\

%$\bullet$ If $P=Q=0$ then as in Proposition \ref{prop:algo_darboux_proba_correct} we deduce that in this situation the algorithm returns a rational first integral with minimal degree because we have the inclusion $\Sigma_1 \cup \mathfrak{S}_{2N+2d-1} \subset \Sigma_2 \cup \mathfrak{S}_{4N+8d-3}$.\\
%$\bullet$ If $Q=0$ and $P\neq 0$ then we deduce that $Py_1+R \in \tilde{\mathcal{E}}^N_{D_1}(x_0^{\star},y_0^{\star})$. Then  Proposition~\ref{prop:algo_darboux_proba_correct} allows us to conclude in this situation.\\
$\bullet$ If $Q\neq 0$ then in \textsf{Build Liouvillian first integral}, we compute $P_1$.\\
If $P_1=0$ then by Proposition \ref{proprepresent} we get a Liouvillian first integral.\\
Now, we suppose $P_1 \neq 0$.\\
 We claim $$Q \not \in \ker \tilde{\mathcal{E}}^N_{D_0}(x_0^{\star},y_0^{\star}).$$ Indeed, if $Q \in \ker  \tilde{\mathcal{E}}^N_{D_0}(x_0^{\star},y_0^{\star})$  then $Q\big(x,y_{\star}(x)\big)=0$ because $(x_0^{\star},y_0^{\star})$ does not belong to $ \Sigma_{D_0,N}$. Thus a factor $\mathcal{R}$ of $Q$ is a Darboux polynomial.\\
  Therefore, $\mathcal{R}$  would give a non-trivial element in $\ker \tilde{\mathcal{E}}^N_{D_2}(x_0^{\star},y_0^{\star})$. This is impossible since the computed solution $\mathcal{S}$ has a minimal weighted degree and $\textrm{w-deg}(\mathcal{S}) \geq \textrm{w-deg}(\mathcal{R})$.\\

As $(x_0^{\star},y_0^{\star}) \not \in \Sigma_{D_2,N}$,  we have  thanks to Lemma \ref{lem:specialized_ker}
$$P\big(x,y_{\star}(x)\big)y_{1,\star}^2(x)+ Q\big(x,y_{\star}(x)\big)y_{2,\star}(x)+R\big(x,y_{\star}(x)\big)y_{1,\star}(x)=0.$$ 
 Since $Q \not \in \ker \tilde{\mathcal{E}}^{N}_{D_0}(x_0^{\star},y_0^{\star})$ we can use the strategy of Lemma \ref{lem:ker_exta_liouv}  and we get  
 $$P_1\big(x,y_{\star}(x)\big)y_{1,\star}(x)+Q_1\big(x,y_{\star}(x)\big)=0,$$ 
 where $\deg(P_1), \deg(Q_1) \leq 2N+3d-1$.\\
 Then, as we have  $\Sigma_1  \subset \Sigma_2 $,  Proposition~\ref{prop:algo_darboux_proba_correct} implies that the algorithm \textsf{Build Darbouxian first integral} applied to $P_1$ and $Q_1$ gives either a Darbouxian first integral with degree smaller than $2N+3d-1$, or a rational first integral.\\
 
Thus when $(x_0^{\star},y_0^{\star}) \not \in \Sigma_{2} \cup \mathfrak{S}_{4N+8d-3}$ the algorithm does not return ``I don't know".
 % a Darboux polynomial with degree smaller than $4N+8d-3$. In this last case, we continue to Step \ref{algo:liouv_test_ipr} and we have a Darboux polynomial $\mathcal{P}_i$ with degree smaller than $4N+8d-3$ vanishing at $(x_0^{\star},y^{\star})$. As $(x_0^{\star},y_0^{\star},x_0^{\star},\mathtt{y}_0^{\star})$ does not belong to $\mathfrak{S}_{4N+8d-3}$, this implies that $D_0$ admits a rational first integral of degree $\leq 4N+8d-3$ and $\mathcal{P}_0,\mathcal{P}_1$ define two different levels of this first integral. Then $\mathcal{P}_0/\mathcal{P}_1$ is not constant and thus gives a rational first integral. So the checking at Step \ref{algo:liouv_test_ipr} is satisfied, and the algorithm returns a rational first integral.
\end{proof}

\begin{prop}\label{prop:condition_proba_algo_liouv}
We set
$$\mathcal{L}(d,N)=d+\mathcal{B}_0(d,N)+ \mathcal{B}_1(d,N)+\mathcal{B}_2(d,N)+\Big(\dfrac{d(d+1)}{2}+5\Big)(4N+8d-3).$$
There exists a polynomial $H_L$ with degree smaller than $\mathcal{L}(d,N)$ such that:\\
If $H_L(x_0^{\star},y_0^{\star}) \neq 0$ then  \textsf{Compute Liouvillian first integral} returns ``None" or an equation leading to a first integral.
\end{prop}

\begin{proof}
The proof is done exactly in the same way as the proof of Proposition~\ref{prop:condition_proba_algo_liouv}.
\end{proof}

\begin{cor}\label{prop:proba_algo_liouv}
Let $\Omega$ a finite subset of $\KK$ of cardinal $|\Omega|$ greater than $\mathcal{L}(d,N)$ and assume that in \textsf{Compute Liouvillian first integral}  $x_0^{\star}$, $y_0^{\star}$ are chosen independently and uniformly at random in $\Omega$. Then,  \textsf{Compute Liouvillian first integral} returns ``None " or an equation leading to a first integral with probability at least
$$1-\dfrac{\mathcal{L}(d,N)}{|\Omega|}.$$
\end{cor}

\begin{prop}\label{prop:liouv-deg-min}
If $D_0$ admits a rational or Darbouxian or Liouvillian first integral with degree smaller than $N$ then \textsf{Compute Liouvillian first integral} returns an equation with minimal degree.
\end{prop}

\begin{proof}
As in the Darbouxian case this is a direct consequence of the minimality of the weighted degree of a solution in \textsf{Compute solution extactic kernel}.
\end{proof}

%%%%%%%%%%%%%%%%%%%%%%%%%%%%%%%%%%%%%%%%%%%%%%%%%%%
\subsection{Riccati first integrals}
In order to compute a Riccati first integral, we follow the same strategy as before:\\
First, we compute a non trivial element $4P(x,y)y_1^4+Q(x,y)(3y_2^2-2y_3y_1)+R(x,y)y_1^2$ in $\tilde{\mathcal{E}}_{D_3}^N(x_0^{\star},y_0^{\star})$. Second, we build from $P$, $Q$, $R$ a Riccati first integral.\\

\textbf{\textsf{Compute Riccati first integral}}\\
\texttt{Input:} $A,B \in \KK[x,y]$, $(x_0^{\star},y_0^{\star})\in \KK^2$, $N \in \NN$\\
\texttt{Output:} An equation $(Eq_3):\partial_y^2 \mathcal{F} - F\mathcal{F}$, where $F(x,y) \in \KK(x,y)$,\\
 or $(Eq_1): \partial_y \mathcal{F} -\sqrt{F}=0$,  where $F(x,y) \in \KK(x,y) \setminus\{0\}$,\\
  or  $(Eq_0): \mathcal{F}-F=0$, where $F(x,y) \in \KK(x,y) \setminus \KK$,\\
   or ``None" or ``I don't know''.
\begin{enumerate}
\item \label{step:testA ric} If $A(x_0^{\star},y_0^{\star})=0$ then Return ``I don't know".
\item \textsf{Compute flow series}($A,B,x_0^{\star},y_0^{\star},N,3$)=: $y_{\star}(x),y_{1,\star}(x),y_{2,\star}(x),y_{3,\star}(x)$.
\item \label{algo:ric_computeker}\textsf{Compute solution extactic kernel}($A,B,y_{\star}(x),y_{1,\star}(x),y_{2,\star}(x),y_{3,\star}(x),N,3$)=:$\mathcal{S}$.\\
If $\mathcal{S}$=``None", then Return(``None"), \\
else $\mathcal{S}=: 4P(x,y)y_1^4+Q(x,y)(3y_2^2-2y_3y_1)+R(x,y)y_1^2$.
\item \textsf{Build Riccati first integral}($A,B,P,Q,R,x_0^{\star},y_0^{\star}$).\\
\end{enumerate}

\textbf{\textsf{Build Riccati first integral}}\\
\texttt{Input:} $A(x,y), B(x,y)$, $P(x,y),Q(x,y),R(x,y) \in \KK[x,y]$ such that \mbox{$(P,Q,R)\neq 0$}, $(x_0^{\star},y_0^{\star})\in\mathbb{K}^2$.\\
\texttt{Output:}An equation $(Eq_3):\partial_y^2 \mathcal{F}-F(x,y)\mathcal{F}=0$, where $F(x,y) \in \KK(x,y)$,\\
 or \mbox{$(Eq_1): \partial_y \mathcal{F} -\sqrt{F(x,y)}=0$}, where $F(x,y) \in \KK(x,y) \setminus \{0\}$, \\
 or $(Eq_0): \mathcal{F}-F=0$, where $F(x,y) \in \KK(x,y) \setminus \KK$, or ``I don't know''.
\begin{enumerate}
\item If $Q=0$ then Return(\textsf{Build Darbouxian first integral}($A,B,4P,R,x_0^{\star},y_0^{\star},2$)).
\item Compute $P_1:=A^4Q^2\big(4D_0(P/Q)+8A(P/Q)\partial_y(B/A)-2A\partial_y^3(B/A)\big)$,\\ $Q_1:=A^4Q^2D_0(R/Q)$.
\item \label{step:daroux_dans_ric} If $P_1=0$ then Return $(Eq_3): \partial_y^2 \mathcal{F}-(P/Q) \mathcal{F}=0$\\
Else Return\big(\textsf{Build Darbouxian first integral}($A,B,P_1,Q_1,x_0^{\star},y_0^{\star},2$)\big).\\
\end{enumerate}
\begin{prop}\label{prop:algo_ric_proba_correct}
The algorithm \textsf{Compute Riccati first integral}  satisfies the following properties:
\begin{itemize}
\item If it returns ``None" then the derivation $D_0$ has no Riccati nor $2$-Darbouxian nor rational first integral with degree smaller than $N$.
\item If it returns an equation $(Eq_0)$ or $(Eq_1)$ or $(Eq_3)$ then this equation leads to a non-trivial first integral.
\item If it returns ``I don't know'', then $(x_0^{\star},y_0^{\star})$ belongs to
$$\Sigma_3  \cup \mathfrak{S}_{4N+10d-3},$$
where 
$\Sigma_3=\mathcal{V}(A)\cup \Sigma_{D_0,N}\cup \Sigma_{D_1,N,2}\cup \Sigma_{D_3,N} $. 
\end{itemize}
\end{prop}

\begin{proof}
The proof of this proposition is similar to the one given for  Proposition~\ref{prop:algo_liouv_proba_correct}.
\end{proof}

As before we deduce the following results:
\begin{prop}\label{prop:condition_proba_algo_ric}
We set
$$\mathcal{R}(d,N)=d+\mathcal{B}_0(d,N)+\mathcal{B}_1(d,N)+ \mathcal{B}_3(d,N)+\Big(\dfrac{d(d+1)}{2}+5\Big)(4N+10d-3)$$
There exists a polynomial $H_R$ with degree smaller than $\mathcal{R}(d,N)$ such that:\\
If $H_R(x_0^{\star},y_0^{\star}) \neq 0$ then  \textsf{Compute Riccati first integral} returns ``None" or an equation leading to a first integral.
\end{prop}

\begin{proof}
The proof is done exactly in the same way as the proof of Proposition~\ref{prop:condition_proba_algo_liouv}.
\end{proof}

\begin{cor}\label{prop:proba_algo_ric}
Let $\Omega$ a finite subset of $\KK$ of cardinal $|\Omega|$ greater than $\mathcal{R}(d,N)$ and assume that in \textsf{Compute Riccati first integral}  $x_0^{\star}$, $y_0^{\star}$ are chosen independently and uniformly at random in $\Omega$. Then,  \textsf{Compute Riccati first integral} returns ``None" or an equation leading to a first integral with probability at least
$$1-\dfrac{\mathcal{R}(d,N)}{|\Omega|}.$$
\end{cor}

\begin{prop}\label{prop:ric-deg-min}
If $D_0$ admits a rational or $2$-Darbouxian or Riccati first integral with degree smaller than $N$ then \textsf{Compute Riccati first integral} returns an equation with minimal degree.
\end{prop}

\begin{proof}
As in the Darbouxian case this is a direct consequence of the minimality of the weighted degree of a solution in \textsf{Compute solution extactic kernel}.
\end{proof}

\subsection{Deterministic algorithms}
In this section we show how to get a deterministic algorithm from our probabilistic ones. We give explicitly the deterministic algorithm for the Riccati case below. The Darbouxian and Liouvillian can be obtained in the same way.\\

\textbf{\textsf{Deterministic computation Riccati first integral}}\\
\texttt{Input:} $A,B \in \KK[x,y]$, such that $A(x,y) \neq 0$, $N \in \NN$\\
\texttt{Output:} An equation $(Eq_3):\partial_y^2 \mathcal{F} - F\mathcal{F}$, where $F(x,y) \in \KK(x,y)$,\\
 or $(Eq_1): \partial_y \mathcal{F} -\sqrt{F}=0$,  where $F(x,y) \in \KK(x,y) \setminus\{0\}$,\\
  or  $(Eq_0): \mathcal{F}-F=0$, where $F(x,y) \in \KK(x,y) \setminus \KK$,\\
   or ``None" or ``I don't know''.
\begin{enumerate}
\item Set  $c:=0$, $x_0^{\star}:=-1$.
\item While $c \leq \mathcal{R}(d,N)+1$ do
\begin{enumerate}
\item $x_0^{\star}:=x_0^{\star}+1$, $\Omega:=\emptyset$.
\item While $A(x_0^{\star},y)=0$ do $x_0^{\star}:=x_0^{\star}+1$.
\item While $|\Omega| \leq \mathcal{R}(d,N)+1$ do
\begin{enumerate}
\item \label{step:algo_det_ric} Choose a random element $y_0^{\star} \in \KK\setminus \Omega$ such that  $A(x_0^{\star},y_0^{\star})\neq 0$.
\item $\mathcal{E}:=$\textsf{Compute Riccati first integral}$(A,B,(x_0^{\star},y_0^{\star}),N)$.
\item If $\mathcal{E}=$``None", then Return ``None".
\item If $\mathcal{E}=$``I don't know" then $\Omega:=\Omega \cup \{y_0^{\star}\}$, Else Return $\mathcal{E}$.
\end{enumerate}
\item $c:=c+1$.
\end{enumerate}
\item Return ``None".

\end{enumerate}

\begin{prop}
The algorithm \textsf{Deterministic computation Riccati first integral} is correct.
\end{prop}
\begin{proof}
The deterministic algorithm repeats the probabilistic algorithm. If the probabilistic  algorithm returns an equation or ``None" then this output is correct thanks to Proposition \ref{prop:algo_ric_proba_correct}.\\
Now, we want to get $x_0^{\star},y_0^{\star}$ such that $H_R(x_0^{\star},y_0^{\star}) \neq 0$.\\
 As we use the probabilistic algorithm with at most $\mathcal{R}(d,N)+1$  different values for $x_0^{\star}$ and $\mathcal{R}(d,N)+1$  different values for $y_0^{\star}$ we necessarily avoid situations where $H_R(x_0^{\star},y_0^{\star})$ is equal to zero. Then Proposition \ref{prop:condition_proba_algo_ric} implies that the probabilistic algorithm returns an output different from ``I don't know" and we get the desired output.
 \end{proof}

%%%%%%%%%%%%%%%%%%%%%%%%%%%%%%%%%%%%%%%%%%%%%%%%%%%%%%%%%%
\section{Complexity results}\label{sec:complexity}
In this section we study the arithmetic complexity of our algorithms.
We focus on the dependency on the degree bound $N$ and we recall that we assume that $N\geq d$, where $d=\max(\deg(A),\deg(B))$ denotes the degree of the polynomial vector field. This hypothesis is natural because if a derivation has a polynomial first integral of degree $N$, then necessarily $N-1 \geq d$. More precisely, we suppose that $d$ is fixed and $N$ tends to infinity.\\

All the complexity estimates are given in terms of arithmetic operations in $\KK$. 
We use the notation $f \in \bigOsoft(g)$, roughly speaking this means that we neglect the  logarithmic factors in the expression of the complexity. For a precise definition, see \cite[Definition 25.8]{GG}. \\
We suppose that the Fast Fourier Transform can be used so that two univariate polynomials with coefficients in $\KK$ and degree bounded by $r$ can be multiplied in  $\bigOsoft(r)$, see \cite[Corollary 8.19]{GG}.\\
 We further assume that  two matrices of size $n$ with entries in $\KK$ can be multiplied using $\bigO(n^\omega)$,  where $2 \leq \omega \leq 3$ is the matrix multiplication exponent, see \cite[Chapter~12]{GG}. We  also  recall  that  a  basis  of  solutions  of  a  linear  system
composed  of $m$ equations  and $n \leq m$ over $\KK$ can  be  computed  using $\bigO(mn^{\omega -1})$ operations in $\KK$, see \cite[Chapter 2]{BiniPan}.\\

The algorithm \textsf{Compute flow series} is a direct application of the algorithm given in \cite{BCLOSSS}. In our situation, the number of arithmetic operations needed to perform this subroutine is in $\bigOsoft(L\sigma+\sigma)$. Here $L$ is the number of arithmetic operations needed to evaluate the rational functions defining the system $(S'_r)$. As $d$ is assumed to be fixed, we have $L \in \bigO(1)$. Furthermore,  $\sigma$ is the precision on the power series, then $\sigma \in \bigO(N^2)$. \\
It thus follows that the computation modulo  $(x-x_0^{\star})^{\sigma}$ of $y_{\star}(x)$, $y_{1,\star}(x)$, $y_{2,\star}(x)$, $y_{3,\star}(x)$ can done with at most $\bigO(N^2)$ arithmetic operations.\\

In  \textsf{Compute solution extactic kernel} we need to find a non trivial  element in $\ker \tilde{\mathcal{E}}_{r,D_r}^N(x_0^{\star},y_0^{\star})$. This can be done with an Hermite-Pad\'e approximation. We recall this setting:\\
 We have $m$ polynomials $f_i(x) \in \KK[x]$, a precision $\sigma$, a shift $s=(s_1,\ldots, s_m)$ and we want to compute $m$ polynomials $p_i(x) \in \KK[x]$ such that 
$$\sum_{i=1}^m p_i.f_i= 0 \mod x^{\sigma}.$$
The set of all solutions $(p_1,\ldots,p_m)$ is a $\KK[x]$-module. A $s$-minimal approximate basis is a basis of this module and furthermore an element of this basis has minimal $s$-degree among all solutions of the problem. We recall that the $s$-degree of $(p_1,\ldots,p_m)$ is $\max_i \deg(p_i +s_i)$.\\
We can compute such a basis with $\bigOsoft\big(m^{\omega-1}(\sigma +\xi)\big)$ arithmetic operations in $\KK$, where $\xi= \sum_i (s_i-\min(s))$, see \cite{BL}, \cite[Theorem 5.3]{ZhouLabahn} and \cite{Villard}.\\
In our situation we have $r \in [[0;3]]$, $m=(r+1)(N+1)$, $\sigma=(r+1)\frac{(N+1)(N+2)}{2}$.\\

When $r=0$ we set:
$$(f_1,\ldots,f_m)=\big(1,y_{\star}(x),y_{\star}^2(x),\ldots,y_{\star}^N(x)\big),$$
$$s=(0,1,2,\ldots,N).$$

When $r=1$ we set:
\begin{eqnarray*}
(f_1,\ldots,f_m)&=&\big(1,y_{\star}(x),y_{\star}^2(x),\ldots,y_{\star}^N(x),\\
&&y_{1,\star}(x),y_{1,\star}(x)y_{\star}(x),y_{1,\star}(x)y_{\star}^2(x), \ldots, y_{1,\star}(x)y_{\star}^N(x)\big),
\end{eqnarray*}
$$s=(0,1,2,\ldots,N,N+1,\ldots,2N+1).$$

When $r=2$ we set
\begin{eqnarray*}
(f_1,\ldots,f_m)&=&\big(y_{1,\star}(x),y_{1,\star}(x)y_{\star}(x),y_{1,\star}(x)y_{\star}^2(x), \ldots,y_{1,\star}(x)y_{\star}^N(x),\\
&& y_{1,\star}^2(x),y_{1,\star}^2(x)y_{\star}(x),y_{1,\star}^2(x)y_{\star}^2(x), \ldots,y_{1,\star}^2(x)y_{\star}^N(x),\\
&& y_{2,\star}(x),y_{2,\star}(x)y_{\star}(x),y_{2,\star}(x)y_{\star}^2(x), \ldots,y_{2,\star}(x)y_{\star}^N(x)\big),
\end{eqnarray*}
 $$s=(0,1,2,\ldots,N,N+1,\ldots,2N+1,2N+2,\ldots,3N+2).$$
 
When $r=3$ we set 
\begin{eqnarray*}
(f_1,\ldots,f_m)&=&\big(y_{1,\star}^4(x),y_{1,\star}^4(x)y_{\star}(x),y_{1,\star}^4(x)y_{\star}^2(x), \ldots,y_{1,\star}^4(x)y_{\star}^N(x),\\
&& \Psi(x),\Psi(x)y_{\star}(x),\Psi(x)y_{\star}^2(x), \ldots,\Psi(x)y_{\star}^N(x),\\
&& y_{1,\star}^{2}(x),y_{1,\star}^{2}(x)y_{\star}(x),y_{1,\star}^{2}(x)y_{\star}^2(x), \ldots,y_{1,\star}^{2}(x)y_{\star}^N(x)\big)
\end{eqnarray*}
where $\Psi(x)=3y_{2,\star}^2(x)-2y_{3,\star}(x)y_{1,\star}(x)$, and
 $$s=(0,1,2,\ldots,N,N+1,\ldots,2N+1,2N+2,\ldots,3N+2).$$

We remark that from a solution $(p_1,\ldots,p_m)$ we get:
\begin{itemize}
\item when $r=1$, a polynomial 
$$Q(x,y)+P(x,y)y_1=\sum_{i=0}^N p_i(x)y^i+\sum_{i=0}^{N}p_{N+1+i}(x)y^iy_1,$$
\item when $r=2$, a polynomial 
\begin{eqnarray*}
R(x,y)y_1+P(x,y)y_1^2+Q(x,y)y_{2}&=&\sum_{i=0}^N p_i(x)y^iy_1
+\sum_{i=0}^{N}p_{N+1+i}(x)y^iy_1^2\\
&&+\sum_{i=0}^Np_{2N+2+i}(x)y^iy_{2},
\end{eqnarray*}
\item when $r=3$,  a polynomial
\begin{eqnarray*}
P(x,y)y_1^4+Q(x,y)\Psi+R(x,y)y_1^2&=& \sum_{i=0}^N p_i(x)y^iy_1^4+\sum_{i=0}^{N}p_{N+1+i}(x)y^i\Psi\\
&&+\sum_{i=0}^Np_{2N+2+i}(x)y^iy_1^{2},
\end{eqnarray*}
where $\Psi=3y_2^2-2y_3y_1$.
\end{itemize}

Therefore a solution with a minimal $s$-degree corresponds to a polynomial in $\ker \tilde{\mathcal{E}}_{r,D_r}^N(x_0^{\star},y_0^{\star}) $ with minimal weighted degree.
Thus the subroutine  \textsf{Compute solution extactic kernel} can be done with at most $\bigOsoft(N^{\omega-1}N^2)=\bigOsoft(N^{\omega+1})$ arithmetic operations in $\KK$.\\

The algorithm \textsf{Build Rational first integral} computes  a gcd of bivariate polynomials with degree in $\bigO(N)$. This subroutine can be done with at most $\bigOsoft(N^2)$ arithmetic operations in $\KK$, see \cite{GG}. Furthermore, we need to factorize a bivariate polynomial with degree at most $N$, this can be done in a probabilistic (respectively deterministic) way with $\bigOsoft(N^3)$  (respectively $\bigOsoft(N^{\omega+1})$) arithmetic operations plus the factorization of an univariate polynomial in $\KK[T]$ with degree $N$, see \cite{BoLeSaScWi04,Lec05}. \\

At last, in \textsf{Build Rational first integral} we solve the linear system 
$$D(Q)=\Omega Q, \textrm{ where }\deg(Q) \leq N.$$
This step can done with $\bigO(N^{\omega+1})$ arithmetic operations. The approach is the following. We set
$$A(x,y)=\sum_{j=0}^d a_j(y/x)x^j,\; B(x,y)= \sum_{j=0}^d b_j(y/x)x^j,$$
$$\Omega(x,y)=\sum_{j=0}^{d-1} \omega_{j+1}(y/x)x^j,\; \omega_0=0,\; Q(x,y)=\sum_{j=0}^N q_{N-j}(y/x)x^j,$$
and we also set $q_j=0$, for all $j \not \in \{0,\dots,N\}$.\\
 We want to solve the equation 
 \begin{equation}\label{eq:linsys}
 A\partial_x Q+B\partial_y Q-\Omega Q=0.
 \end{equation}\\
  Substituting in Equation \ref{eq:linsys} the above expressions with $y=zx$, we obtain for the coefficient of  $x^{N+d-1-j}$
$$\sum_{i=0}^d \big(b_{d-i}(z)-za_{d-i}(z)\big) q'_{j-i}(z) + \big((N-j+i)a_{d-i}(z)-\omega_{d-i}(z)\big)q_{j-i}(z)$$
Then Equation \ref{eq:linsys} gives
\begin{eqnarray}\label{eqpol}
&&\big(za_d(z)-b_d(z)\big) q'_j(z)+ \big(\omega_d(z)-(N-j)a_d(z)\big)q_j(z) \\
&=&\sum_{i=1}^d \big(b_{d-i}(z)-za_{d-i}(z)\big) q'_{j-i}(z) +\big((N-j+i)a_{d-i}(z)-\omega_{d-i}(z)\big)q_{j-i}(z) \nonumber
\end{eqnarray}
So the equation $A\partial_x Q+B\partial_y Q-\Omega Q=0$ is equivalent to the system of equations \eqref{eqpol} for $j=0,\dots, N+d$. This system is triangular as the righthandside of \eqref{eqpol} only involves $q_l$ with $l<j$. Let us now remark that there can be at most one $j=j_0$ such that
$$za_d(z)=b_d(z),\;\; \omega_d(z)=(N-j)a_d(z)$$
as we cannot have $a_d=b_d=\omega_d=0$. Thus equation \eqref{eqpol} seen as a differential equation in $q_j$  always admits an affine space of polynomial solutions of dimension $\leq 1$ for $j\neq j_0$.

We now solve the system of equations \eqref{eqpol} by induction on $j$. For $j=0$, equation \eqref{eqpol} is a linear differential equation of order $1$, and thus admits a vector space of polynomial solutions of dimension $\leq 1$ if $j_0\neq 0$ or dimension $N+1$ for $j=j_0$. This vector space can be found by solving a linear system of $N+1$ unknowns and $N+d$ equations, which costs $\bigO(N^\omega)$.

Let us now assume $(q_{N-j},\dots,q_N)$ belongs to a known vector space $\mathcal{V}_j$. We look at Equation \eqref{eqpol} for $j+1$. This is a linear differential equation in $q_{N-j-1}\in\mathbb{K}[z]_{\leq N}$ and $(q_{N-j},\dots,q_N)\in \mathcal{V}_j$. As it is of order $1$ in $q_{N-j-1}$, the dimension of the space of solutions $\mathcal{V}_{j+1}$ can grow at most by $1$ if $j\neq j_0$ or $N+1$ if $j=j_0$. Such linear system with $N+1+\dim \mathcal{V}_j$ unknowns and $N+d$ equations can be solved in $\bigO\big((N+d)(N+1+\dim \mathcal{V}_j)^{\omega-1}\big)$ operations.

Now, as $\dim \mathcal{V}_j$ grows at most by one except at most for one $j=j_0$ where it grows at most by $N+1$, we always have $\dim \mathcal{V}_j \leq 2N+d$. Thus each step of the resolution of \eqref{eqpol} costs at most $\bigO\big((N+d)(N+1+2N+d)^{\omega-1}\big)=\bigO(N^\omega)$. As there are $N+d$ steps, the overall cost is $\bigO(N^{\omega+1})$.\\

%our algorithms we test if $D_0(\mathcal{P}_0/\mathcal{P}_1)=0$. This step corresponds to the multiplication of bivariate polynomials with degree in $\bigO(N)$. Therefore this step can be done with at most $\bigOsoft(N^2)$ arithmetic operations in $\KK$.\\

In conclusion our probabilistic algorithms use at most $\bigOsoft(N^{\omega+1})$ arithmetic operations in $\KK$ plus the factorization of a univariate polynomial with degree at most $N$. This is  the complexity given in Theorem \ref{thmmain1}.\\

As $\mathcal{D}(d,N), \mathcal{L}(d,N)$ and $\mathcal{R}(d,N)$ are in $\bigO(N^4)$, the deterministic algorithm uses at most \mbox{$\bigOsoft(N^{\omega+9})$} arithmetic operations in $\KK$ and $\bigO(N^8)$ factorizations of  univariate polynomials in $\KK[T]$ with degree at most $N$.

%%%%%%%%%%%%%%%%%%%%%%%%%%%%%%%%%%%%%%%%%%%%%%%%%%%%%%%%%%%%%
%%%%%%%%%%%%%%%%%%%%%%%%%%%%%%%%%%%%%%%%%%%%%%%%%%%%%%%%%%%%%
\section{Examples}\label{sec:examples}
The algorithms developed in the previous sections have been implemented in Maple. This implementation is available with some examples at:\\
 \texttt{http://combot.perso.math.cnrs.fr/software.html},\\
 \texttt{https://www.math.univ-toulouse.fr/$\sim$cheze/Programme.html}.\\
The computations for the following examples have been done on a Macbook pro 2013, intel core i7 2.8 Ghz.\\
For practical reasons, the implemented version of our algorithms do not use the Hermite-Padé algorithm to  find a solution of the extactic kernel. We just solve a linear system. Furthermore, the solutions $y_{\star}(x)$,\ldots, $y_{3,\star}(x)$ are computed from $y_{\star}(x)$ and then integrated. For example, we compute $y_{1,\star}(x)$ with the formula: 
$$ y_{1,\star}(x)=exp \Big(\int \dfrac{B}{A}(x,y_{\star}(x)) dx \Big).$$ 

\subsection{The Darbouxian case}
 Let us consider the system
$$\dot{x}=x^2+2xy+y^2-4x+4y-2 ,\;\; \dot{y}=x^2+2xy+y^2+4x-4y-2.$$
The algorithm \textsf{Compute Darbouxian first integral} returns in $0.2s$, when $N=3$:
$$\frac{\partial \mathcal{F}}{\partial y}+\frac{14(x^2+2xy+y^2-4x+4y-2)}{11(x-y)(x^2+2xy+y^2-2)}=0,$$
 which after integration leads to the Darbouxian first integral
$$\mathcal{F}(x,y)=\sqrt{2}\ln\left(x+y-\sqrt{2}\right)-\sqrt{2}\ln\left( x+y+\sqrt{2}\right)+\ln(x-y).$$
Now, we set
$$z=\frac{x+y+\sqrt{2}}{x+y-\sqrt{2}},\quad w=x-y,$$
and we have for the level set $\mathcal{F}(x,y)=c$
$$w=e^c z^{\sqrt{2}}.$$
This curve is not algebraic for almost all $c$, and thus the system does not admit a rational first integral.\\

 The initial point used in the execution of the algorithm above was $(1,8)$.  To get ``I don't know", we need for example to use a bad point, i.e. a point vanishing a Darboux polynomial. From this point, we will obtain a Darboux polynomial, and thus the algorithm will try to deduce from this polynomial a rational first integral. This will not work as the vector field has no rational first integral. We can choose for example $(1,1)$ or $(1,\sqrt{2}-1)$. Such  initial points were never encountered when using (small) random initial points. In particular, the probabilistic algorithm is the only algorithm necessary to use in practice, and we never have to rerun it with several initial points.\\
 
If we use the algorithm \textsf{Compute Darbouxian first integral} with $N=2$ then the output is ``None". This is correct and means that there exists no Darbouxian first integral with degree smaller than 2.\\

Now, we slightly modify  the previous example
$$\dot{x}=2\lambda^2x-2\lambda^2y+\lambda^2-x^2-2xy-y^2, \;\; \dot{y}=2\lambda^2y-2\lambda^2x+\lambda^2-x^2-2xy-y^2.$$
The algorithm \textsf{Compute Darbouxian first integral} returns with $\lambda=100$ and $N=3$ 
$$ \frac{\partial \mathcal{F}}{\partial y}-\frac{312(x^2+2xy+y^2-20000x+20000y-10000)}{469(x-y)(y+100+x)(y-100+x)}=0$$
in $0.2s$ which after integration leads to the Darbouxian first integral
$$\mathcal{F}(x,y)=100\ln\left(x+y-100\right)-100\ln\left( x+y+100\right)+\ln(x-y).$$
Now the exponential of $\mathcal{F}$ gives a rational first integral
$$\left(\frac{x+y-100}{x+y+100}\right)^{100}(x-y),$$
which is of degree $101$.\\

We remark that if we want to compute a rational first integral we can use \textsf{Compute Darbouxian first integral}. In this case the bound $N$ is a bound on the degree of the product of the irreducible Darboux polynomials  used to write the rational first integral (3 in the previous example) and not a bound on the degree of the first integral (101 in the previous example). The difference between these bounds is important when the rational first integral has one or several factors with large multiplicities.

\subsection{Comparison with the Avelar-Duarte-da Mota's algorithm}

We compare our algorithm with the algorithm proposed by J. Avellar, L.G.S. Duarte, L.A.C.P. da Mota, denoted in the following by: ADM algorithm, see  \cite{PSsolve}.\\
First, we consider  a vector field of the form $(-\partial_y G, \partial_x G)$ with
$$G=x+\sum\limits_{i=1}^m i\ln(x+y-i),$$
after multiplying by a common denominator.\\
 We find the Darbouxian first integral in the following times.
\begin{center}\begin{tabular}{|c|c|c|c|c|c|c|c|c|c|}\hline
$m$                                        & $1$   & $2$   &  $3$   & $4$  &  $5$  & $6$   \\\hline
\textsf{Compute Darbouxian first integral} & 0.016 & 0.031 & 0.063 & 0.234 & 0.951 & 5.1 \\\hline
ADM algorithm                              & 0.094 & 0.047 & 0.078 & 0.109 & 0.109 & 0.187 \\\hline
\end{tabular}\end{center}
Unexpectedly, the ADM algorithm fails with $m\geq 7$. We see that the ADM algorithm computation times are much better than ours. This is because all Darboux polynomials are of degree $1$, and the ADM algorithm computes them first. After there is a  combinatorial step, with an exponential complexity, but here it is  negligible at those low $m$.\\

 Secondly, we study a growing degree Darboux polynomial case
$$G=x+\ln(x+y^m-1).$$
We find the Darbouxian first integral in the following times.
\begin{center}\begin{tabular}{|c|c|c|c|c|c|c|c|c|c|}\hline
$m$        & $1$   & $2$   &  $3$   & $4$ & $5$ &$6$ \\\hline
\textsf{Compute Darbouxian first integral} & 0.015 & 0.015 & 0.31 & 0.125 & 0.359 & 0.889 \\\hline
ADM algorithm  & 0.094 & 0.078 & 1.123 & $>10^3$  & $>10^3$ & $>10^3$  \\\hline
\end{tabular}\end{center}
The timings of our algorithm have the same order of magnitude, but the ADM algorithm becomes almost unusable. This is because the computation of Darboux polynomials is very expansive even for low degrees. In other words, as soon as the Darboux polynomials are not linear, the ADM algorithm is not usable. Our algorithm never computes Darboux polynomials, and thus avoids this problem.

\subsection{The Liouvillian case}
Consider the system
$$\dot{x}=2x^2-2y^2-1,\;\dot{y}=2x^2-2y^2-3.$$
The algorithm \textsf{Compute Liouvillian first integral} returns in $0.3s$ when $N=3$ 
$$\frac{\partial^2 \mathcal{F}}{\partial y^2}-\frac{2(x+y)(2x^2-4xy+2y^2-1)}{2x^2-2y^2-1} \frac{\partial \mathcal{F}}{\partial y}=0.$$
After integration, this gives the first integral
$$\mathcal{F}(x,y)=\sqrt{\pi}\hbox{erf}(x-y)+(x+y)e^{-(x-y)^2}.$$
Now, we set $z=x-y,w=x+y$, and the equation $\mathcal{F}(x,y)=c$ gives
$$w=(c-\sqrt{\pi}\hbox{erf}(z))e^{z^2}$$
This function is never algebraic for $c\in\mathbb{C}$ as $\hbox{erf}$ is not even elementary. Thus the system does not admit a rational first integral.\\
The function $\mathcal{F}$ is holomorphic and all solution curves (outside the straight line at infinity) are of the form $\mathcal{F}(x,y)=c\in\mathbb{C}$. As none of these curves are algebraic, the system does not admit any Darboux polynomial. As the poles of a Darbouxian first integral are Darboux polynomials, a Darbouxian first integral should be a polynomial, which is again not possible as there are no rational first integrals. Thus the system admits a Liouvillian first integral but no first integral of lower class.\\

Now, we consider the example 185 of Kamke which is an Abel equation with a Liouvillian first integral
$$\dot{x}=-x^7,\; \dot{y}=y^2(5x^3+2x^2y+2y).$$
\textsf{Compute Liouvillian first integral} gives in $1.95$s with $N=7$:
$$\frac{\partial^2 \mathcal{F}}{\partial y^2}+ \frac{x^6+7x^3y+6x^2y^2+6y^2}{2y(x^6+2x^3y+x^2y^2+y^2)} \frac{\partial\mathcal{F}}{\partial y}=0.$$
This system admits a lower class first integral, a $4$-Darbouxian first integral of degree $32$, which can be recovered by integration of this equation, giving
$$\tilde{\mathcal{F}}(x,y)=\int \frac{y^{3/2}\sqrt{x}(5x^3+2x^2y+2y)}{(x^6+2x^3y+x^2y^2+y^2)^{5/4}}dx+ \frac{x^{15/2}}{(x^6+2x^3y+x^2y^2+y^2)^{5/4}\sqrt{y}}dy.$$
This integral can also be searched directly as a $4$-Darbouxian first integral with $N=32$, and is obtained in $22166s$.

\subsection{The Riccati case}
The example 43 of Kamke is an Abel equation. When we set $a=3$ and $b=17$ in this equation, we get
$$\dot{x}=1,\; \dot{y}=-(9x^2+36x+17)y^3-3xy^2.$$
This equation admits a Riccati first integral.\\
\textsf{Compute Riccati first integral} gives in $10.9s$ with $N=9$:
$$\frac{\partial^2 \mathcal{F}}{\partial y^2}-\frac{3P}{4(9x^2y+36xy+17y-6)^2y^3}\mathcal{F}=0,$$
with
\begin{eqnarray*}
P&=&81x^4y^3+648x^3y^3-18x^3y^2+1602x^2y^3-180x^2y^2+1224xy^3\\
&&+3x^2y-466xy^2+289y^3+24xy-204y^2+36y-2.
\end{eqnarray*}

This equation together with the equation of the first integral defines a PDE system with a two dimensional space of solutions. The first integral in Kamke's book is written using Bessel functions. Thus the solutions of this PDE system can be expressed in terms of Bessel functions here, but this is not an easy task.\\

In general, Abel equations are of the form
$$\frac{\partial y}{\partial x} =f_3(x) y^3+f_2(x) y^2+f_1(x)y+f_0(x).$$
These can be seen as a generalization of the Riccati equation. However, in contrary to the Riccati equation, they are not all solvable in algebraic-differential terms. Still many integrable families are known, Abel integrability is typically searched by looking into a known table list up to some transformations. Our algorithm  can  detect any integrable cases, even belonging to an unknown new integrable family.\\

\subsection{The generic case} In a generic situation a vector field has no symbolic first integral. Let us now consider a random quadratic vector field
$$\dot{x}=2x^2+xy-2y^2-1,\;\; \dot{y}=2x^2-2y^2+y-3.$$
We do not find any Liouvillian nor Riccati first integrals (and thus neither Darbouxian, $2$-Darbouxian or rational first integral) up to degree $9$, with the following timings.\\

\begin{center}\begin{tabular}{|c|c|c|c|c|c|c|c|c|c|}\hline
$N$        & $1$   & $2$   &  $3$   & $4$  & $5$ &   $6$ &   $7$ &   $8$  & $9$\\\hline
Liouvillian& 0.015 & 0.047 & 0.620 & 0.219 & 0.639 & 1.576 & 4.072 & 8.736 & 17.44 \\\hline
Riccati    & 0.015 & 0.047 & 0.094 & 0.280 & 0.843 & 2.262 & 4.898 & 11.62 & 23.32 \\\hline
\end{tabular}\end{center}

\subsection{Rational first integral with degree bigger than $N$}
Let us consider the following example
$$\dot{x}=\lambda x^3-\lambda xy^2-2\mu y^2-\lambda x,\;\; \dot{y}= \lambda x^2y-\lambda y^3-2\mu xy-\lambda y$$
with $\lambda,\mu\in\mathbb{Z}$. This vector field always admits the first integral
$$I(x,y)=\lambda\ln\left(\frac{x}{y}-\frac{\sqrt{x^2-y^2}}{y}\right)+\mu\ln\left(\frac{x^2-y^2+1}{x^2-y^2-1}-2\frac{\sqrt{x^2-y^2}}{x^2-y^2-1}\right)$$
which is a $2$-Darbouxian first integral, which is of degree $8$.
Indeed, we have: $\partial_y I - F=0$, where $F^2=P/Q$ and 
\begin{eqnarray*}
P&=&{\lambda}^{2}{x}^{6}-2\,{\lambda}^{2}{x}^{4}{y}^{2}+{\lambda}^
{2}{x}^{2}{y}^{4}-4\,\lambda\,\mu\,{x}^{3}{y}^{2}+4\,\lambda\,\mu\,x{y
}^{4}-2\,{\lambda}^{2}{x}^{4}+2\,{\lambda}^{2}{x}^{2}{y}^{2}\\
&&+4\,{\mu}^
{2}{y}^{4}+4\,\lambda\,\mu\,x{y}^{2}+{\lambda}^{2}{x}^{2}\\
Q&=&{x}^{6}{y}^
{2}-3\,{x}^{4}{y}^{4}+3\,{x}^{2}{y}^{6}-{y}^{8}-2\,{x}^{4}{y}^{2}+4\,{
x}^{2}{y}^{4}-2\,{y}^{6}+{x}^{2}{y}^{2}-{y}^{4}.
\end{eqnarray*}

 As $\lambda/\mu \in\mathbb{Q}$, we can however build from this a rational first integral (with degree depending on $\lambda/\mu$). This is also a particular case of a Liouvillian first integral, which is then of degree $8$. This kind of example is build by searching radical extension of $\bar{\mathbb{K}}(x,y)$ with groups of unit of rank $\geq 2$. The first integral is then a linear combination of logs of these units.\\

For this example, we here display the timings in seconds of the algorithms Rational, Darbouxian, Liouvillian and Ricatti first integrals with initial point $(2,5)$. The degree columns are the minimum $N$ for which the output is not ``None".\\

\begin{tabular}{|c|c|c|c|c|c|c|c|c|}\hline
$(\lambda,\mu)$ & Rat deg & time &  D deg & time & L deg & time & Ric deg & time\\\hline
(1,0)& 1 & 0.031 & 1 & 0.016 & 1 & 0.016 &  1&0.016\\\hline
(0,1)& 2 &0.016  & 1 & 0.015 & 1 & 0.016 & 2 &0.031\\\hline
(1,1)& 3 &0.031  & 2 & 0.031 & 2 & 0.078 & 3 & 0.328\\\hline
(2,1)& 4 &0.031  &3  &0.047  & 2 & 0.062 & 3 & 0.281\\\hline
(1,2)& 5 & 0.031 & 4 & 0.110 & 4 & 0.343 & 5 & 1.934\\\hline
(3,1)& 5 & 0.047 & 4 &0.234  & 3 & 0.203 & 5 & 2.168\\\hline
(1,3)& 7 & 0.203 & 5 & 0.483 & 5 & 1.935 & 7 & 16.27 \\\hline
(4,1)& 6 & 0.109 & 5 &0.437  & 4 & 0.733 &5  & 2.652\\\hline
(3,2)& 7 & 0.219 & 5 & 0.359 & 5 & 1.825 & 6 & 6.012 \\\hline
(2,3)& 8 & 0.359 & 6 & 1.030 & 5 & 1.981 & 6 & 12.43\\\hline
(1,4)& 9 & 1.482 & 6 & 2.559 & 6 & 11.62 & 8 & $74.1^*$ \\\hline
(5,1)& 7 & 0.453 & 5 & 0.873 & 5 & 3.853 & 7 & 26.59\\\hline
(1,5)&11 & 3.807 & 7 & 4.290 & 7 & 20.87 & 8 & $54.2^*$ \\\hline
(6,1)& 8 & 1.264 & 6 & 2.980 & 6 & 15.14 & 7 & 40.62\\\hline
(5,2)& 9 & 0.889 & 6 & 1.326 & 6 & 6.381 & 7 & 20.37\\\hline
(4,3)&10 & 3.354 & 7 & 6.443 & 6 & 15.50 & 7 & 39.57\\\hline
(3,4)&11 & 4.945 & 7 & 5.694 & 6 & 12.76 & 8 & $73^*$\\\hline
(2,5)&12 & 8.392 & 8 & 9.267 & 7 & 30.60 & 8 & $83^*$\\\hline
(1,6)&13 & 10.99 & 8 & 10.59 & 8 & 55.08 & 8 & $64.1^*$\\\hline
\end{tabular}\\

%Remark that in several cases, the degree $N$ necessary to obtain an output from the algorithm is smaller than the minimal degree of the rational first integral.
% In particular, 
In many cases (all except those with $\star$), Darbouxian, Liouvillian and Ricatti algorithms have returned the rational first integral even if it is of degree larger than $N$. For example, with $(\lambda,\mu)=(3,4)$ the Liouvillian algorithm with $N=6$ returns a rational first integral of degree $11$. Remark that the degree cannot be higher than $8$ for Liouvillian or Ricatti first integrals, because the $2$-Darbouxian first integral is always present and its degree is not growing.\\

%\appendix
%\section*{Appendix}

\textbf{Acknowledgements}
Guillaume Chèze acknowledges the support and hospitality of the Laboratorio Fibonacci (Pisa) where a part of this work has been developed during November 2016 when he was visiting Thierry Combot.
%\textbf{Acknowledgements}

\bibliographystyle{plain}      % basic style, author-year citations
\bibliography{ref-int-prem}   % name your BibTeX data base

%% Non-BibTeX users please use
%\begin{thebibliography}{}
%%
%% and use \bibitem to create references. Consult the Instructions
%% for authors for reference list style.
%\end{thebibliography}

\end{document}